%% file: main-v6.tex
\RequirePackage{fix-cm}
\documentclass[a4paper]{article}

\newif\ifshowcomments
\showcommentstrue 

\usepackage[margin=1in]{geometry}  
\usepackage[english]{babel}
\usepackage[utf8x]{inputenc}
\usepackage[T1]{fontenc}
\usepackage{parskip}
\usepackage[dvipsnames]{xcolor}
\usepackage{graphicx}
\usepackage{subcaption}
\usepackage{wrapfig}
\usepackage{longtable}
\usepackage{multirow}
\usepackage{listings}
\usepackage{makecell}
\usepackage{array}
\usepackage{float}
\usepackage{dsfont}
\usepackage{rotating}
\usepackage{booktabs}
\usepackage{enumerate}
\usepackage{extarrows}
\usepackage{tikz}
\usetikzlibrary{chains}
\usetikzlibrary{arrows}
\usetikzlibrary{backgrounds,positioning}
\usepackage{pgf}

\usepackage[normalem]{ulem}

\usepackage{amsmath}
\usepackage{amssymb}
\usepackage{amsthm}
\usepackage{bm,bbm}
\usepackage[linesnumbered,ruled,vlined,resetcount]{algorithm2e}
\usepackage{xcolor}
\usepackage{algorithmicx}
\usepackage{mathtools}
\usepackage{mathrsfs}
\usepackage{ stmaryrd } 

\usepackage{natbib}
\usepackage{bibunits}
\let\given\relax
\usepackage[
  colorlinks,
  citecolor=blue,
  linkcolor=red,
  anchorcolor=red,
  urlcolor=blue
]{hyperref}
\mathtoolsset{showonlyrefs}
\usepackage{authblk}
\usepackage{todonotes}
\input{macro.tex}

\renewcommand{\hat}[1]{\widehat{#1}}

\renewcommand{\widetilde}{\tilde} 

\newcommand{\iid}{\mathrm{iid}}

\newcommand{\nmin}{n_{\min}}

\newcommand{\indiv}{\mathrm{indiv}}
\newcommand{\pooled}{\mathrm{pooled}}

\newcommand{\cbt}{\mathrm{CBT}} 
\newcommand{\sss}{\mathrm{SSS+}}

\newcommand{\rh}{\mathrm{rh}}

\newcommand{\rb}{\mathrm{CE}}
\newcommand{\ew}{\mathrm{ew}}

\graphicspath{{./}{figures_and_tables/}}
\makeatletter
\def\input@path{{./}}
\makeatother
\long\def\comment#1{}

\title{Compound p-values for competition-based testing}

\author[1]{Junu Lee\thanks{Corresponding author: \href{mailto:junulee@wharton.upenn.edu}{\nolinkurl{junulee@wharton.upenn.edu}}}}
\author[2]{Rina Foygel Barber} 
\author[3]{Richard J. Samworth}
\author[1]{Zhimei Ren}

\affil[1]{Department of Statistics and Data Science, University of Pennsylvania}
\affil[2]{Department of Statistics, University of Chicago}
\affil[3]{Statistical Laboratory, University of Cambridge}

\date{}
\begin{document}
\begin{bibunit}[apalike]
\maketitle

\begin{abstract}
  Standard competition-based multiple testing methods---such 
  as the knockoff filter and the target-decoy procedure---are
  typically formulated as head-to-head competitions between a 
  test unit and its assigned negative control unit. 
  Motivated by real-world applications, 
  which leverage any available negative controls, we study this competition-based framework in more intricate settings where  
  competitions vary in size and structure.   We propose new methods based on {\em compound p-values}---p-values
  that are superuniform  only on aggregate---that enable 
  powerful multiple testing with broader 
  competition structures, via the Bonferroni and Benjamini--Hochberg procedures.  The former provides exact familywise error rate control, 
  whereas the latter may incur a false discovery rate inflation factor, for which we derive 
  upper bounds across a range of settings.  
  We showcase the performance of our general competition-based testing 
  methodology through extensive simulations and real-data applications, 
  including analyses of virtual spatial transcriptomics and CRISPR 
  screen datasets that exhibit precisely the complex competition 
  structures that motivate our methods. 
\end{abstract}

\section{Introduction}
\label{sec:intro}
\input{sections/intro}

\section{Preliminaries}
\label{sec:prelim}
\input{sections/prelim}

\section{Compound p-values for competition-based testing}
\label{sec:cbt}
\input{sections/cbt}

\section{Competition-based testing with shared control competitors}
\label{sec:shared_cbt}
\input{sections/shared_cbt}


\section{Simulations}
\label{sec:sims}
\input{sections/sims}

\section{Real data analysis}
\label{sec:real_data}
\input{sections/real_data}

\section{Discussion}
\label{sec:discussion}
\input{sections/discussion}
\subsection*{Acknowledgements}

J.L. is supported by a Graduate Research Fellowship from the NSF.
R.F.B. is supported by the Office of Naval Research via grant N00014-24-1-2544.  R.J.S. is supported by European Research Council Advanced Grant 101019498.  
Z.R. is supported by the National Science Foundation  under grant DMS-2413135 and Wharton Analytics.
The authors thank Kaitian Jin for his helpful suggestions regarding the virtual ST application in Section~\ref{sec:application-virtual-st}.
The authors also thank William Bekerman, T. Tony Cai, Patrick Chao, Eugene Katsevich, Ziang Niu, Dylan Small, and Jeffrey Zhang for their constructive feedback at various stages of this work.

The first author employed generative AI models provided by OpenAI, 
most recently GPT-6 Astra, to assist in this work.  AI tools contributed to the codebase (for simulations 
and experiments) and the visualizations (diagrams and plots), as well as aiding with integration techniques in Appendix~\ref{sec:appd-fdr-supermartingale}.

\putbib[ref]
\end{bibunit}

\clearpage 

\appendix
\begin{bibunit}[apalike]
\makeatletter
\def\hyper@natlinkstart#1{%
  \Hy@backout{#1}%
  \hyper@linkstart{cite}{suppcite.#1}%
  \def\hyper@nat@current{#1}%
}
\def\hyper@natlinkbreak#1#2{%
  \hyper@linkend#1\hyper@linkstart{cite}{suppcite.#2}%
}
\def\hyper@natanchorstart#1{%
  \Hy@raisedlink{\hyper@anchorstart{suppcite.#1}}%
}
\makeatother
\section{Compound p-values for random hypotheses}
\label{sec:appd-random-compound}
\input{sections_appendix/appd_random_compound}

\input{sections_appendix/appd_technical_proofs}

\input{sections_appendix/appd_evalues}

\input{sections_appendix/appd_fdr_supermartingale.tex}

\input{sections_appendix/appd_equal_weights}

\input{sections_appendix/appd_miscellaneous_results}

\input{sections_appendix/appd_deferred_sims}

\clearpage 

\renewcommand{\bibsection}{\section*{Supplementary References}}
\putbib[ref]
\end{bibunit}

\end{document}

%% file: macro.tex
\theoremstyle{plain}

\newtheorem{definition}{Definition}
\newtheorem{theorem}{Theorem}
\newtheorem{lemma}[theorem]{Lemma}

\newtheorem{corollary}[theorem]{Corollary}
\newtheorem{remark}{Remark}
\newtheorem{assumption}{Assumption}
\newtheorem{proposition}[theorem]{Proposition}

\newcommand{\RR}{\mathbb{R}}

\newcommand{\NN}{\mathbb{N}}
\newcommand{\EE}{\mathbb{E}}
\newcommand{\PP}{\mathbb{P}}
\newcommand{\ind}{\mathbbm{1}}

\newcommand{\Bern}{\text{Bern}}
\newcommand{\Bin}{\text{Bin}}

\newcommand{\unif}{\text{Unif}}

\newcommand{\fwer}{\textnormal{FWER}}
\newcommand{\fdr}{\textnormal{FDR}}
\newcommand{\fdp}{\textnormal{FDP}}

\newcommand{\bonf}{\textnormal{Bonf}}
\newcommand{\bh}{\textnormal{BH}}
\newcommand{\ebh}{\textnormal{eBH}}

\newcommand{\blk}{\mathrm{blk}}
\newcommand{\rmax}{r_{\max}}

\newcommand{\argmax}[1]{\underset{#1}{\arg\!\max}}
\newcommand{\argmin}[1]{\underset{#1}{\arg\!\min}}

\usepackage{xspace}

\newcommand{\given}{{\,|\,}}
\newcommand{\biggiven}{\,\big{|}\,}
\newcommand{\Biggiven}{\,\Big{|}\,}
\newcommand{\bigggiven}{\,\bigg{|}\,}

\def\@#1\@{\begin{align}#1\end{align}}
\def\$#1\${\begin{align*}#1\end{align*}}

\definecolor{myblue}{rgb}{.8, .8, 1}
\definecolor{mathblue}{rgb}{0.2472, 0.24, 0.6} 
\definecolor{mathred}{rgb}{0.6, 0.24, 0.442893}
\definecolor{mathyellow}{rgb}{0.6, 0.547014, 0.24}

\newcommand{\cA}{{\mathcal{A}}}

\newcommand{\cC}{{\mathcal{C}}}

\newcommand{\cE}{{\mathcal{E}}}
\newcommand{\cF}{{\mathcal{F}}}
\newcommand{\cG}{{\mathcal{G}}}
\newcommand{\cH}{{\mathcal{H}}}
\newcommand{\cI}{{\mathcal{I}}}
\newcommand{\cJ}{{\mathcal{J}}}
\newcommand{\cK}{{\mathcal{K}}}
\newcommand{\cL}{{\mathcal{L}}}
\newcommand{\cM}{{\mathcal{M}}}
\newcommand{\cN}{{\mathcal{N}}}

\newcommand{\cP}{{\mathcal{P}}}

\newcommand{\cR}{{\mathcal{R}}}
\newcommand{\cS}{{\mathcal{S}}}
\newcommand{\cT}{{\mathcal{T}}}

%% file: sections/intro.tex

%

Consider a modern multiple testing setting where our goal is to 
test $m$ null hypotheses $H_1, \dots, H_m$ simultaneously using 
corresponding test statistics (or ``scores'') $S_{1, 0}, \dots, 
S_{m,0}$, for which the null distributions are unknown. 
Such settings commonly arise in fields such as 
genomics~\citep{dempster2025controlling},
drug discovery~\citep{jin2023selection} 
and natural language processing~\citep{lee2025full}, 
which leverage contemporary, state-of-the-art machine learning 
models to provide accurate, but black-box predictions with unknown or intractable null distributions.  To obtain p-values for each null, then, the analyst can 
turn to {\em negative controls} to learn a suitable 
nonparametric empirical null distribution 
for each $S_{j,0}$ under~$H_j$~\citep[e.g.][]{gao2023simultaneous}. 

There are two main paradigms for exploiting such negative controls. 
First, one can collect {\em individual} negative controls: each $H_j$ 
is assigned $n_j$ negative control scores 
$(S_{j,1}, \dots, S_{j, n_j})$ that are representative
{of the distribution} of $S_{j,0}$
under its null. Alternatively, one may collect a  {\em pooled} set  of 
negative controls 
{$(S_1, \dots, S_N)$}
used by all hypotheses 
as  a shared empirical null distribution. Once the 
negative control set is established, one can construct p-values 
for each $H_j$ by evaluating the resulting empirical 
CDF (eCDF) at $S_{j,0}$, yielding {\em empirical p-values}{
\begin{equation}
    p_j^\indiv \coloneqq  \frac{1 + \sum_{k=1}^{n_j} \ind\{  S_{j,k}\ge S_{j, 0} \} }{n_j+1} \quad \text{or}\quad 
    p_j^\pooled \coloneqq  \frac{1 + \sum_{i=1}^N \ind\{  S_i \ge S_{j,0} \} }{N+1}.
\end{equation}
}

The individual and pooled control paradigms each have their own 
drawbacks. When using individual controls, the empirical p-values
rely on $n_j$ being sufficiently large to attain reasonable statistical power; otherwise, they suffer from poor granularity, {
as evidenced by the denominator of $p_j^\indiv$ above.  
As an alternative, one can pool the individual control groups together and 
compute p-values with the pooled eCDF, enabling a finer granularity of $N +1 =\sum_{j=1}^m n_j + 1$.}
{However,} the validity of the empirical p-values may be compromised when the null distributions of the negative controls vary significantly across $H_1,\ldots,H_m$.  Thus the analyst may be forced into a tradeoff between validity and statistical power.

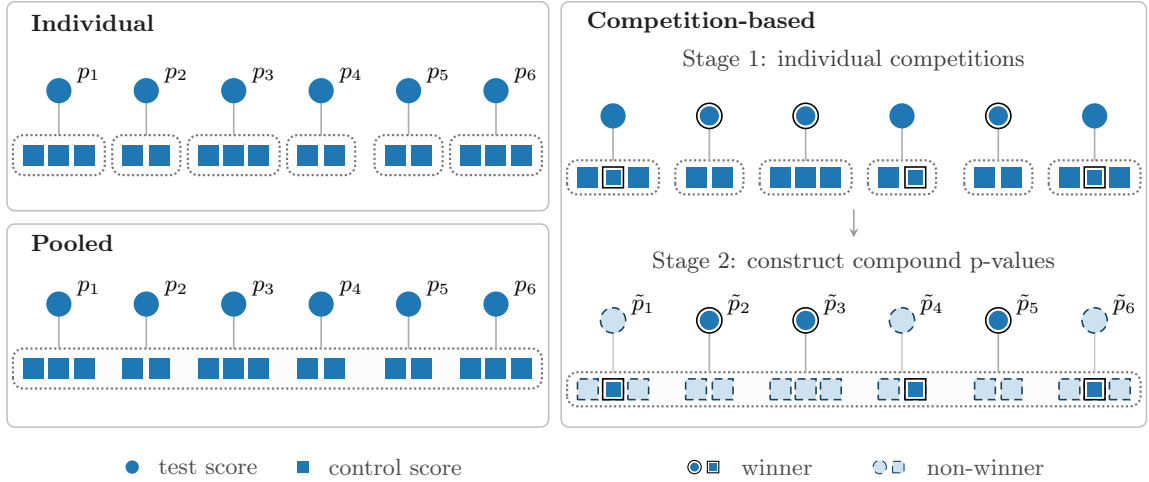
\begin{figure}[tbp!]
    \centering
    \resizebox{0.96\textwidth}{!}{%
        \input{figures_and_tables/negative-control-paradigms-comparison}%
    }
    \caption{
    {
        A visualization of the three negative control-based multiple 
        testing paradigms. Test scores are indicated by circles; 
        control scores are indicated by squares. Individual testing:
        each test score is compared with its own collection of control scores.
        Pooled testing: each test score is compared with the combined 
        collection of all control scores, pooled across all hypotheses. 
        Competition-based testing: for each hypothesis, Stage 1 compares
        their test and control scores and selects the highest scoring 
        unit as the winner of this ``competition''. Then,
        in Stage 2, each test score that won in Stage 1
        is compared against all control scores that also won in Stage 1.
        For individual and pooled testing, the comparisons result in strict p-values $p_j$; for competition-based testing, the comparisons result in 
        compound p-values $\tilde p_j$.
    }
    }
    \label{fig:negative-control-paradigms}
\end{figure}

In this paper, we develop a general methodology based on 
the {\em competition-based testing} framework 
\citep{emery2020multiple}, which initially
ranks each test score among its 
individual control group {(holding a metaphorical ``competition'' between
the test and control units based on their scores)}
and then further aggregates information 
about these comparisons across all hypotheses to boost statistical power. 
Competition-based testing is an established approach for 
multiple testing with negative controls, with the two most prominent examples being the \emph{target-decoy method}~\citep{elias2007target} in 
proteomics and the \emph{knockoff filter}~\citep{barber2015controlling,
candes2018panning,barber2020robust} for genomics-inspired variable selection.
In these use cases (and, to our knowledge, all other previous literature),
the framework is specifically instantiated for settings where the practitioner is able to generate the same number of negative controls for each hypothesis.  
However, many other application areas that employ negative controls  may not be able to enforce equal control group 
sizes, perhaps because the controls are collected unevenly across  hypotheses and/or are assigned to different hypotheses from a pooled sample (rather than being shared) using randomness or side information \citep{gui2026adaptive}.  
Existing competition-based  methods are therefore not well adapted to these practical real data settings. 

Our proposal overcomes this methodological gap and handles 
competitions with significantly greater complexity than those composed of
equal-sized control groups. Central to our methods is the 
concept of {\em compound p-values}  $\tilde{p}_1,\ldots,\tilde{p}_m$ \citep{armstrong2022false}, whose null tail probabilities are superuniform ``on aggregate'':
\begin{equation}
    \label{eq:compound-pval-defn}
    \sum_{j\in\cH_0}\PP(\tilde p_j\le t)\le mt,
    \qquad \text{for every }t\in[0,1].
\end{equation}
We begin with the important observation that 
the existing competition-based testing algorithm that underlies the 
target-decoy procedure and
knockoff filter 
can be reframed as running the Benjamini--Hochberg (BH) 
procedure~\citep{benjamini1995controlling} on a 
collection of compound p-values. 
This  motivates a general compound p-value construction 
that accommodates broader competition structures, including 
competitions of varied size and a novel setup with 
certain competitions sharing the same negative controls.
For each setting, we also investigate
familywise error rate (FWER) and false discovery rate (FDR) control,
providing nuanced inflated control guarantees for the latter.
Figure~\ref{fig:negative-control-paradigms} visualizes our 
competition-based proposal and juxtaposes it with the 
individual and pooled paradigms for negative control-based 
multiple testing. 

\subsection{An illustrative example in virtual spatial transcriptomics}
\label{sec:virtual-st-preview}

\begin{figure}[tbp!]
    \centering
    \begin{subfigure}[t]{2.25in}
        \vspace{0pt}
        \centering
        \includegraphics[width=\linewidth,keepaspectratio,clip]{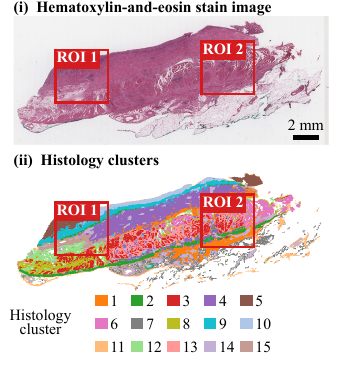}
        \caption{Gastric cancer sample with clusters.}
        \label{fig:utopia-preview-tissue}
    \end{subfigure}%
    \hspace{0.04in}%
    \begin{subfigure}[t]{3.97in}
        \vspace{0pt}
        \centering
        \includegraphics[width=\linewidth,keepaspectratio,clip]{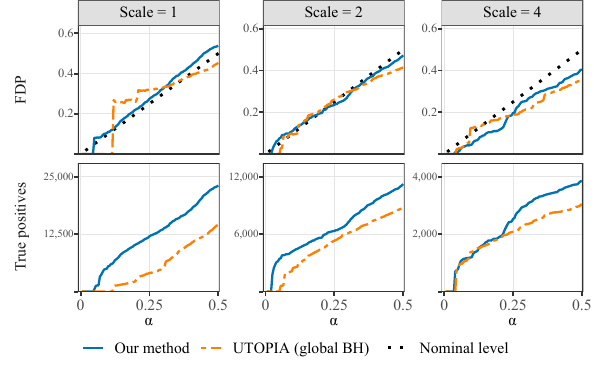}
        \caption{Preview of real data analysis using {our competition-based testing method}.}
        \label{fig:utopia-preview-results}
    \end{subfigure}
    \caption{
    Preview of the gastric cancer virtual ST application in Section~\ref{sec:application-virtual-st}.
    (a) (i) The H\&E-stained tissue and (ii) corresponding 15 histology clusters at the $32\,\mu\mathrm{m}$-per-pixel resolution with the two  measured regions of interest (ROIs)---the negative controls---outlined in red.
    The cluster composition differs inside and outside the measured ROIs, so the ROI controls do not form a representative global reference pool for the unmeasured locations outside the ROIs.
    (b) FDP and true positives at multiple spatial aggregation scales for both our proposed method and BH applied to the combined UTOPIA p-values over all clusters; the black dotted line in the FDP panels is the nominal FDP level.  The initial visualization and clustering were due  to~\citet{jin2026multiscale}, and were recreated for this paper.}
    \label{fig:utopia-tissue-histology-preview}
\end{figure}

We illustrate the tension between using a pooled control group 
and using individual controls per hypothesis, as well as 
the benefit of using our competition-based methodology, 
by previewing our real data example in virtual spatial 
transcriptomics (ST).
Virtual ST uses black-box prediction models to 
predict gene expression levels across a tissue sample 
cross-section by using spatial features derived from 
a hematoxylin-and-eosin (H\&E) stain of the 
sample (Figure~\ref{fig:utopia-preview-tissue}, (i)).
The workflow of~\citet{jin2026multiscale}, titled UTOPIA, quantifies the 
uncertainty of these predictions 
(or scores) to identify 
tissue locations with non-negligible gene expression by 
re-interpreting the task as a multiple hypothesis testing 
problem with FDR control. 
UTOPIA uses the predicted expression at each 
location to test the null hypothesis that 
its true underlying gene expression level 
is at most a user-specified threshold $c \in [0,1]$.
As the null distributions of the scores are unknown, 
UTOPIA leverages small regions of interests (ROIs) 
with known underlying expressions (measured by standard ST technology) and generates negative control 
scores from the prediction model. 

In the specific gastric cancer case study of~\citet{jin2026multiscale},
the ROI scores are unfit for use as a pooled collection of negative 
controls due to the heterogeneity inside and outside the ROI regions. 
Their preliminary analysis reveals the presence 
of fifteen distinct histological clusters in the sample
(Figure~\ref{fig:utopia-preview-tissue}, (ii)).
Since varying histological composition can affect 
both the genetic expression and the performance of the 
prediction model, the ROI scores need not be 
representative of null test scores outside the ROIs; {hence, naïvely computing $(p_j^\pooled)_{j\in[m]}$ using all the ROI scores will lead to severely miscalibrated empirical p-values.} 
Instead of pooling the negative controls,
UTOPIA therefore conducts its analysis cluster-by-cluster: {for each $j\in [m]$, 
$p_j^\pooled$ is constructed using only 
the scores of the negative controls corresponding to the same histological cluster as
location $j$.}

The authors find several drawbacks of this workaround. First, the resulting {cluster-specific}
empirical p-values can have poor granularity, so they 
choose to run the BH procedure within each cluster rather than 
globally---representing a shift in the statistical objective, as 
within-cluster FDR control does not in general imply overall 
FDR control. Second, per-cluster analysis does not completely 
address other sources of heterogeneity (e.g., from batch effects), so ad hoc corrections are needed to produce valid 
empirical p-values. These deficiencies are further exacerbated 
when the ROIs are limited in size, as is typical due to the cost of producing the standard ST measurements.

In Section~\ref{sec:application-virtual-st}, we implement our 
proposed  competition-based methodology to the gastric cancer 
case study.  We first address the heterogeneity present within clusters by further partitioning the controls into smaller, but more plausibly homogeneous, groups based on available covariates. 
This step itself is not part of our methodology,
but rather enabled by it. The standard UTOPIA approach would then result in p-values with extremely poor 
granularity; on the other hand, our methods can leverage these small competitions. 
Another problem is that even with smaller competitions, the fact that 
the ROIs are only a small proportion of the tissue means that not every test score can be 
assigned its own negative controls; thus, existing competition-based methods cannot be employed.  As a remedy, we capitalize on a simple observation: many ROI scores 
can serve as reasonable controls to a subset of all test scores. 
This observation inspires our investigation of settings where
the hypotheses are partitioned into blocks, 
with the test scores within each block 
sharing the same negative control competitors. 
Figure~\ref{fig:utopia-preview-results} 
demonstrates that the resulting competition-based method 
exhibits empirical evidence of FDR control 
while selecting many more locations with nonnull gene expression 
levels than the per-cluster p-value baseline. 

\subsection{Organization of the paper}

The rest of the paper is organized as follows. 
Section~\ref{sec:prelim} contains multiple testing preliminaries,
including the aforementioned key observation which relates existing 
competition-based testing procedures to compound p-values.  In
Section~\ref{sec:cbt}, we introduce our methods for the
general competition-based testing framework, while in Section~\ref{sec:shared_cbt} we extend the ideas to the setting of shared control competitors.  
We investigate the empirical performance of our methods through simulations
in Section~\ref{sec:sims} and applications to real-world datasets for virtual ST and 
CRISPR screens in Section~\ref{sec:real_data}.%
\footnote{Code for the simulations and data analyses is available at \url{https://github.com/leejunu/competition-based-testing}.}
We conclude with a discussion in Section~\ref{sec:discussion}.

%% file: figures_and_tables/negative-control-paradigms-comparison.tex
\begingroup%
\begin{tikzpicture}[x=1cm,y=1cm]
\definecolor{competitionblue}{HTML}{1F77B4}

\tikzset{
  test score/.style={
    circle,
    draw=none,
    fill=competitionblue,
    minimum size=3.4mm,
    inner sep=0pt,
    outer sep=0pt
  },
  control score/.style={
    rectangle,
    draw=none,
    fill=competitionblue,
    minimum size=2.8mm,
    inner sep=0pt,
    outer sep=0pt
  },
  test winner/.style={
    circle,
    draw=black,
    line width=0.60pt,
    fill=white,
    minimum size=3.4mm,
    inner sep=0pt,
    outer sep=0pt,
    path picture={
      \node[circle,draw=none,fill=competitionblue,minimum size=2.6mm,inner sep=0pt,outer sep=0pt]
        at (path picture bounding box.center) {};
    }
  },
  control winner/.style={
    rectangle,
    draw=black,
    line width=0.60pt,
    fill=white,
    minimum size=2.8mm,
    inner sep=0pt,
    outer sep=0pt,
    path picture={
      \node[rectangle,draw=none,fill=competitionblue,minimum size=2.0mm,inner sep=0pt,outer sep=0pt]
        at (path picture bounding box.center) {};
    }
  },
  test nonwinner/.style={
    circle,
    draw=competitionblue!58!black,
    densely dashed,
    line width=0.60pt,
    fill=competitionblue!10,
    minimum size=3.4mm,
    inner sep=0pt,
    outer sep=0pt
  },
  control nonwinner/.style={
    rectangle,
    draw=competitionblue!58!black,
    densely dashed,
    line width=0.60pt,
    fill=competitionblue!10,
    minimum size=2.8mm,
    inner sep=0pt,
    outer sep=0pt
  },
  test background/.style={
    circle,
    draw=competitionblue!55!black,
    densely dashed,
    line width=0.60pt,
    fill=competitionblue!22,
    minimum size=3.4mm,
    inner sep=0pt,
    outer sep=0pt
  },
  control background/.style={
    rectangle,
    draw=competitionblue!55!black,
    densely dashed,
    line width=0.60pt,
    fill=competitionblue!22,
    minimum size=2.8mm,
    inner sep=0pt,
    outer sep=0pt
  },
  comparison edge/.style={
    draw=black!32,
    line width=0.62pt,
    line cap=round
  },
  muted comparison edge/.style={
    draw=black!22,
    line width=0.55pt,
    line cap=round
  },
  control pool/.style={
    rounded corners=4pt,
    draw=black!54,
    densely dotted,
    line width=0.82pt,
    fill=black!1
  },
  method box/.style={
    rounded corners=3pt,
    draw=black!25,
    line width=0.58pt,
    fill=white
  },
  panel title/.style={
    font=\fontsize{9}{10.5}\selectfont\bfseries,
    text=black!88,
    anchor=west
  },
  step title/.style={
    font=\fontsize{9}{10.5}\selectfont,
    text=black!76,
    anchor=center
  },
  pvalue label/.style={
    font=\fontsize{9}{10.5}\selectfont,
    text=black,
    anchor=west,
    yshift=2.15mm,
    inner sep=0pt
  },
  legend test/.style={
    circle,
    draw=none,
    fill=competitionblue,
    minimum size=1.8mm,
    inner sep=0pt
  },
  legend control/.style={
    rectangle,
    draw=none,
    fill=competitionblue,
    minimum size=1.55mm,
    inner sep=0pt
  },
  legend test winner/.style={
    circle,
    draw=black,
    line width=0.42pt,
    fill=white,
    minimum size=1.8mm,
    inner sep=0pt,
    path picture={
      \node[circle,draw=none,fill=competitionblue,minimum size=1.2mm,inner sep=0pt,outer sep=0pt]
        at (path picture bounding box.center) {};
    }
  },
  legend control winner/.style={
    rectangle,
    draw=black,
    line width=0.42pt,
    fill=white,
    minimum size=1.55mm,
    inner sep=0pt,
    path picture={
      \node[rectangle,draw=none,fill=competitionblue,minimum size=0.95mm,inner sep=0pt,outer sep=0pt]
        at (path picture bounding box.center) {};
    }
  },
  legend test nonwinner/.style={
    circle,
    draw=competitionblue!58!black,
    densely dashed,
    line width=0.42pt,
    fill=competitionblue!22,
    minimum size=1.8mm,
    inner sep=0pt
  },
  legend control nonwinner/.style={
    rectangle,
    draw=competitionblue!58!black,
    densely dashed,
    line width=0.42pt,
    fill=competitionblue!22,
    minimum size=1.55mm,
    inner sep=0pt
  },
  legend text/.style={
    font=\fontsize{9}{10.5}\selectfont,
    text=black!70,
    anchor=west
  }
}

  \path[use as bounding box] (0,0.38) rectangle (15.60,7.15);
  \fill[white] (0,0.38) rectangle (15.60,7.15);

  \draw[method box] (0.12,4.28) rectangle (7.43,7.10);
  \draw[method box] (0.12,1.36) rectangle (7.43,4.10);
  \draw[method box] (7.60,1.36) rectangle (15.50,7.10);

  \node[panel title] at (0.32,6.82) {Individual};

  \begin{scope}[yshift=3.20cm]
    \foreach \xmin/\xmax in {
      0.20/1.44,1.54/2.46,2.56/3.80,
      3.90/4.82,5.08/6.00,6.10/7.34}
      \draw[control pool] (\xmin,1.56) rectangle (\xmax,2.10);

    \foreach \x in {0.82,2.00,3.18,4.36,5.54,6.72}
      \draw[comparison edge] (\x,2.70) -- (\x,2.10);

    \foreach \x [count=\j] in {0.82,2.00,3.18,4.36,5.54,6.72}{
      \node[test score] (individual-t\j) at (\x,2.70) {};
      \node[pvalue label,xshift=0.75mm] at (individual-t\j.east) {$p_{\j}$};
    }
    \foreach \x in {0.48,0.82,1.16,1.82,2.18,2.84,3.18,3.52,
                      4.18,4.54,5.36,5.72,6.38,6.72,7.06}
      \node[control score] at (\x,1.83) {};
  \end{scope}

  \node[panel title] at (0.32,3.84) {Pooled};

  \begin{scope}[yshift=0.32cm]
    \draw[control pool] (0.20,1.56) rectangle (7.34,2.10);

    \foreach \x in {0.82,2.00,3.18,4.36,5.54,6.72}
      \draw[comparison edge] (\x,2.70) -- (\x,2.10);

    \foreach \x [count=\j] in {0.82,2.00,3.18,4.36,5.54,6.72}{
      \node[test score] (pooled-t\j) at (\x,2.70) {};
      \node[pvalue label,xshift=0.75mm] at (pooled-t\j.east) {$p_{\j}$};
    }
    \foreach \x in {0.48,0.82,1.16,1.82,2.18,2.84,3.18,3.52,
                      4.18,4.54,5.36,5.72,6.38,6.72,7.06}
      \node[control score] at (\x,1.83) {};
  \end{scope}

  \node[panel title] at (7.82,6.82) {Competition-based};
  \node[step title] at (11.55,6.31) {Stage 1: individual competitions};

  \begin{scope}[xshift=-0.5mm]
  \foreach \xmin/\xmax in {
    7.73/8.97,9.19/10.11,10.33/11.57,
    11.79/12.71,13.09/14.01,14.23/15.47}
    \draw[control pool] (\xmin,4.50) rectangle (\xmax,4.96);

  \foreach \x in {8.35,9.65,10.95,12.25,13.55,14.85}
    \draw[comparison edge] (\x,5.56) -- (\x,4.96);

  \foreach \x in {8.35,9.65,10.95,12.25,13.55,14.85}
    \node[test score] at (\x,5.56) {};
  \foreach \x in {8.01,8.35,8.69,9.47,9.83,10.61,10.95,11.29,
                    12.07,12.43,13.37,13.73,14.51,14.85,15.19}
    \node[control score] at (\x,4.73) {};

  \foreach \x in {9.65,10.95,13.55}
    \node[test winner] at (\x,5.56) {};
  \foreach \x in {8.35,12.43,14.85}
    \node[control winner] at (\x,4.73) {};
  \end{scope}

  \draw[-stealth,draw=black!40,line width=0.55pt] (11.55,4.30) -- (11.55,3.92);
  \node[step title] at (11.55,3.58) {Stage 2: construct compound p-values};

  \begin{scope}[xshift=-0.5mm]
  \draw[control pool] (7.73,1.64) rectangle (15.47,2.10);

  \foreach \x in {8.35,12.25,14.85}
    \draw[muted comparison edge] (\x,2.80) -- (\x,2.10);
  \foreach \x in {9.65,10.95,13.55}
    \draw[comparison edge] (\x,2.80) -- (\x,2.10);

  \node[test background] (aggregate-t1) at (8.35,2.80) {};
  \node[test winner] (aggregate-t2) at (9.65,2.80) {};
  \node[test winner] (aggregate-t3) at (10.95,2.80) {};
  \node[test background] (aggregate-t4) at (12.25,2.80) {};
  \node[test winner] (aggregate-t5) at (13.55,2.80) {};
  \node[test background] (aggregate-t6) at (14.85,2.80) {};

  \foreach \j in {1,...,6}
    \node[pvalue label,xshift=0.75mm] at (aggregate-t\j.east) {$\tilde{p}_{\j}$};

  \foreach \x in {8.01,8.35,8.69,9.47,9.83,10.61,10.95,11.29,
                    12.07,12.43,13.37,13.73,14.51,14.85,15.19}
    \node[control background] at (\x,1.87) {};
  \foreach \x in {8.35,12.43,14.85}
    \node[control winner] at (\x,1.87) {};
  \end{scope}

  \node[legend test] at (1.81,0.82) {};
  \node[legend text] at (2.04,0.82) {test score};

  \node[legend control] at (4.11,0.82) {};
  \node[legend text] at (4.34,0.82) {control score};

  \node[legend test winner] at (9.40,0.82) {};
  \node[legend control winner] at (9.65,0.82) {};
  \node[legend text] at (9.91,0.82) {winner};

  \node[legend test nonwinner] at (11.90,0.82) {};
  \node[legend control nonwinner] at (12.15,0.82) {};
  \node[legend text] at (12.41,0.82) {non-winner};
\end{tikzpicture}%
\endgroup

%% file: sections/prelim.tex

\subsection{Problem setup}
\label{sec:setup}

For the collection of null hypotheses $H_1, \dots, H_m$, we denote by $\cH_0 \subseteq [m]\coloneqq \{1, \dots, m\}$ 
the index set of true nulls and by 
$\cH_1\coloneqq [m] \setminus \cH_0 $ the index set of false nulls 
(true nonnulls). The general objective is to identify a data-driven rejection set $\hat R \subseteq [m] $ that 
captures as much of $\cH_1$ as possible while avoiding the selection of too many elements of $\cH_0$. In this paper, we 
design procedures with the goal of controlling either the 
familywise error rate (FWER) or the false discovery rate (FDR), 
defined as follows: 
\begin{align}
    \label{eq:fwer-defn}
    \fwer(\hat R)
    &\coloneqq \PP\bigl(|\hat R\cap\cH_0|\ge 1\bigr);\\
    \label{eq:fdr-defn}
    \fdr(\hat R)
    &\coloneqq \EE\bigl[\fdp(\hat R)\bigr],
    \quad \text{where }
    \fdp(\hat R)
    \coloneqq \frac{|\hat R\cap\cH_0|}{|\hat R|\vee 1}.
\end{align}
(Above, FDP stands for ``false discovery proportion''.)

Classical multiple testing approaches assume we have {\em strict} p-values $(p_j)_{j\in [m]}$, so that $\PP_{H_j}(p_j \le t) \le t$ for all $t\in[0,1]$---in other words, $p_j$ is superuniform under the null $H_j$.  For FWER control at level $\alpha \in(0,1)$, the well-known Bonferroni correction~\citep{dunn1961multiple} returns the set
\begin{equation}
    \label{eq:bonferroni}
    \hat R^\bonf(p_1,\dots,p_m)
    \coloneqq \{j\in[m]:p_j\le \alpha/m\},
\end{equation}
which is guaranteed to satisfy $\fwer\bigl(\hat R^\bonf(p_1,\dots,p_m)\bigr) \le \alpha$.
On the other hand, for FDR control, the BH procedure~\citep{benjamini1995controlling} order the p-values as $p_{(1)} \leq \cdots \leq p_{(m)}$ and again rejects hypotheses according to a threshold-based rejection rule: 
\begin{equation}
    \label{eq:bh-rejection-set}
    \hat R^\bh(p_1,\dots,p_m)
    \coloneqq \bigl\{j\in[m]:p_j \le \alpha\hat k/m\bigr\}, \quad \text{where } \hat k 
    \coloneqq
    \max\bigl\{k\in[m]:p_{(k)}\le \alpha k/m\bigr\}.
\end{equation} 
By convention, $\hat k = 0$ in the case where the set is empty. Unlike the Bonferroni correction, which always controls the FWER, the BH procedure only guarantees $\fdr\bigl(\hat R^\bh(p_1,\dots,p_m)\bigr) \le \alpha$ under certain conditions on the dependence structure of the p-values---for example, when $p_1, \dots, p_m$ are independent or follow the positive regression dependence on a subset (PRDS) condition~\citep{benjamini2001control}. 

\subsubsection{Compound p-values for multiple testing}
\label{sec:compound-mt-prelim}

The popularity of the Bonferroni and BH procedures makes them  
natural choices for compound p-values %
\citep{ignatiadis2025empirical, barber2025false, ignatiadis2026empirical}.
The Bonferroni procedure is particularly well-suited, as demonstrated by the following proposition, which, although obvious, does not appear to have been made explicit in the literature.
\begin{proposition}
    \label{prop:fwer-compound}
    If $\tilde p_1,\dots,\tilde p_m$ are compound valid, then
    \begin{equation}
        \fwer\bigl(\hat R^\bonf
        (\tilde p_1,\dots,\tilde p_m)\bigr)\le\alpha.
    \end{equation}
\end{proposition}
\begin{proof}
By the definition of compound p-values at $t = \alpha/m$,
\[
\PP\biggl(\sum_{j\in\cH_0}\ind\{ \tilde  p_j \le \alpha/m \} \ge 1 \biggr) \le  \EE\biggl[\sum_{j\in\cH_0}\ind\{ \tilde  p_j \le \alpha/m \} \biggr] \le \sum_{j \in \cH_0}\PP\Bigl(\tilde p_j \le \frac{\alpha}{m}\Bigr) \le m\cdot \frac{\alpha}{m} = \alpha,
\]
as required.
\end{proof} 
\begin{remark}  
    \label{rem:weighted-p-compound}
    Proposition \ref{prop:fwer-compound} is similar to the underlying FWER control of the weighted Bonferroni correction~\citep{roeder2009genome, ignatiadis2021covariate}, 
    which associates a strict p-value~$p_j$ and positive weight $w_j$ 
    to each $H_j$, with $\sum_{j=1}^m w_j =1$. The weighted Bonferroni correction then returns the rejection
    set $\hat R^\bonf(p_1/w_1, \dots, p_m/w_m)$.    
    In fact, Proposition \ref{prop:fwer-compound} is the more general result, as $p_1/w_1,\ldots,p_m/w_m$ are compound p-values.  Similarly, the random-weights variant of~\citet{ignatiadis2023evalues} yields compound p-values when the weights are independent of the p-values.
\end{remark}
\begin{remark}
    \label{rem:holm-failure}
    For FWER control with compound p-values, one may ask whether Holm's procedure~\citep{holm1979simple}, which uniformly improves on Bonferroni correction, is valid. We answer in the negative: a construction where  the FWER is violated is provided in Appendix~\ref{sec:appd-holm-counterexample}. 
\end{remark}

In contrast, BH does not retain guaranteed nominal FDR control when run on compound p-values.  \citet{armstrong2022false} proves that the FDR is asymptotically controlled at its nominal level when the compound p-values are weakly dependent in the sense of~\citet{storey2004strong}, and provides counterexamples to show that asymptotic FDR control may fail in the absence of such a condition.  In terms of finite-sample bounds,~\citet{barber2025false} establish that $\fdr\bigl(\hat R^\bh(\tilde p_1, \dots, \tilde p_m)\bigr) \le 1.93\alpha$ when $(\tilde p_j)_{j\in\cH_0}$ are mutually independent and further independent of $(\tilde p_j)_{j\in\cH_1}$, and give an example of independent compound p-values for which $\fdr\bigl(\hat R^\bh(\tilde p_1, \dots, \tilde p_m)\bigr) \ge 7\alpha/6$.  Furthermore,~\citet{deng2026randomized} argue that
compound p-values typically fail to satisfy the PRDS condition. The competition-based testing methodology introduced in this paper induces a complicated dependence structure among the compound p-values, meaning that its FDR properties must be established from scratch.

\subsection{Multiple testing with negative controls}

In the introduction, we discussed the use of negative control 
scores to construct empirical p-values, to which standard 
p-value-based multiple testing procedures may be applied. 
We now define these p-values under both the pooled and 
individual negative control paradigms and discuss their validity. 
Subsequently, we  {discuss} the Selective SeqStep+ (SSS+)
algorithm~{\citep{barber2015controlling}}, which is the most well-known competition-based testing algorithm 
and which underlies both the knockoff filter
and the target-decoy procedure.

\subsubsection{Individual controls}

In the individual control paradigm, we calibrate evidence against $H_j$ by collecting negative controls $S_{j,1}, \dots, S_{j, n_j}$ 
to form a representative null empirical distribution for $S_{j,0}$.
This allows us to construct an individual empirical p-value 
$p_j^\indiv$ by  evaluating the eCDF at $S_{j,0}$:
\begin{equation}
    \label{eq:indiv-pval-defn}
    p_j^\indiv = \frac{1 + \sum_{k=1}^{n_j} \ind\{  S_{j,k}\ge S_{j, 0} \} }{n_j+1}.
\end{equation}
The random variable $p_j^\indiv$ is a strict p-value when the scores $(S_{j,k})_{k=0}^{n_j}$ 
are jointly exchangeable. For ease of future exposition, we assume a 
slightly stronger exchangeability assumption for the individual control paradigm.
\begin{assumption}[Individual exchangeability]
    \label{assu:individual}
    Conditional on $\cG_{-j} \coloneqq \sigma(\{ S_{\ell, k}\colon \ell \neq j, 0 \le k \le n_\ell \})$, the test and control scores $(S_{j,k})_{k=0}^{n_j}$ are exchangeable for $j\in\cH_0$.
\end{assumption}
One appealing consequence of Assumption~\ref{assu:individual} 
is that it ensures
$\PP_{H_j}\bigl(p_j^\indiv \le t \given (p^\indiv_{\ell})_{\ell \neq j}\bigr) \le t$ 
for all $t\in[0,1]$ and $j \in [m]$, 
which is sufficient for nominal FDR control of BH~\citep[Section~2.1]{fithian2020conditional}.  

{We observe that $p_j^\indiv$ can never take values below $1/(n_j+1)$}, while Bonferroni or BH thresholds are of order $\alpha/m$.  
{
Therefore, the granularity of $p_j^\indiv$  can prevent $H_j$ from being rejected
regardless of the magnitude of $S_{j,0}$. For example, when 
$\min_j 1 / (n_j+1) \gg \alpha/m$, it is impossible to reject $H_j$ via the 
Bonferroni correction. In realistic settings, the number of negative controls 
per hypotheses is  limited; as a consequence,  the poor granularity of the 
$p_j^\indiv$ typically results in a powerless procedure. 
}

\subsubsection{Pooled controls}

As an alternative to individual controls, we may instead proceed on the basis that the distributions of the test scores under the null hypotheses are homogeneous in some manner, 
thereby justifying their calibration 
via the same shared pool of negative controls. For example, 
in the outlier detection problem, the 
goal is to determine which unlabeled units are outliers in the sense 
that they do not follow an inlier distribution $P$. Nonparametric 
approaches~\citep{bates2023testing,mary2022semi,gao2023simultaneous}
compute empirical p-values via a reference set sampled 
independently from $P$.  Writing $(S_1, \dots, S_N)$ for the negative control set, we may construct p-values as
\begin{equation}
    \label{eq:pooled-pval-defn}
    p_j^\pooled = \frac{1 + \sum_{i=1}^N \ind\{  S_i \ge S_{j,0} \} }{N+1}.
\end{equation}
These are also known as \emph{conformal p-values} in the literature~\citep{bates2023testing}.
Under the following joint exchangeability assumption, each null 
$p_j^\pooled$ is a strict p-value:
\begin{assumption}[Pooled exchangeability]
    \label{assu:pooled}
    Conditional on $\sigma( \{ S_{j,0}\colon j\in \cH_1 \} )$, the concatenated control and null  test scores $(S_i)_{i\in[N]} \cup (S_{j,0})_{j\in\cH_0}$ are  exchangeable. 
\end{assumption}
Moreover, under Assumption~\ref{assu:pooled},~\citet{bates2023testing} 
prove that the p-values $p_1^\pooled,\ldots,p_m^\pooled$ satisfy the PRDS property, so the 
BH procedure guarantees nominal FDR control---even though the sharing of negative controls immediately precludes the 
independence of the p-values. 

As previously mentioned, one way to obtain a pooled control set is
to combine all individual controls. In this case, $p_j^\pooled$
has denominator $\sum_{j=1}^m n_j + 1$, so it improves on 
$p_j^\indiv$ in terms of granularity. This benefit is offset by the fact
that Assumption~\ref{assu:individual} does not imply 
Assumption~\ref{assu:pooled}, putting the validity of $p_j^\pooled$ 
in doubt; see Section~\ref{sec:virtual-st-preview} for a real data example. 

\subsubsection{Competition-based testing with individual control groups}
\label{sec:prelim-cbt}

We now turn to the definition of the Selective SeqStep+ (SSS+) algorithm,
designed for individual control groups. For each hypothesis $H_j$, 
SSS+ holds a competition between the corresponding test and control scores, before aggregating the outcomes of all competitions to decide 
which test scores exhibit enough evidence to reject their 
associated hypotheses.
SSS+ was initially implemented in the simplest case where each hypothesis was assigned exactly one negative control, i.e. $n_j = 1$ for each $j \in [m]$.  Later work extended the procedure to the setting where each control group 
is of the same size: $n_j = n \in \NN$~\citep{gimenez2019improving, emery2020multiple}.
For general $n$, SSS+ returns the rejection set 
\begin{equation}
\label{eq:multi-kn}
    \begin{aligned}
    \hat R^\sss
    &\coloneqq
    \{j : S_{j,0} \ge T,\; W_{j,0}=1\},\\
    \text{where }
    T
    &\coloneqq 
    \min \bigg\{
    t\in\cS :
    \frac{
    \frac{1}{n}
    +
    \sum_{\ell=1}^m \frac{1}{n}\sum_{k=1}^{n}
    \ind\{W_{\ell,k}=1,\; S_{\ell,k}\ge t\}
    }{
    1\vee
    \sum_{\ell=1}^m
    \ind\{W_{\ell,0}=1,\; S_{\ell,0}\ge t\}
    }
    \le \alpha
    \bigg\}, \\
    \cS &\coloneqq  \{S_{j,k}\colon j \in [m], 0 \le k \le n\}, \\W_{j,k} &\coloneqq \ind\{S_{j,k} = \max( S_{j,0}, \dots, S_{j, n})\}.
    \end{aligned}
\end{equation}
Thus, $W_{j,k}$ is a ``winner indicator'' for whether the $k$th unit won the $j$th competition, so when there are no ties among $S_{j,1},\ldots,S_{j,n}$, at most one of the terms in the inner sum in the numerator in the definition of $T$ makes a positive contribution.  The fraction defining $T$ can be interpreted as a rough estimate of 
the FDP when we reject all winning test scores that are at least  $t$; by convention, we take $\min \varnothing = \infty$ in the definition of $T$.  Under Assumption~\ref{assu:individual} and the above no-ties assumption, $\hat R^\sss$ has guaranteed nominal FDR control~\citep{barber2015controlling}.

\citet{emery2020multiple} present empirical evidence that for small $n$,
SSS+ can report significant power even when
BH applied to the $p_j^\indiv$ is
powerless. 
For this reason, they argue that SSS+ is the superior 
method under the individual controls paradigm. 
One finding of this paper is that the 
equivalence between SSS+ and BH (see Proposition~\ref{prop:sss-compound-bh} 
below) can explain this power gap. 
The compound p-values we construct have 
a granularity of $1/(mn)$, compared to the granularity of 
$1/(n+1)$ exhibited by $p_j^\indiv$; hence, 
the former is more favorable for the BH rejection rule.
Furthermore, we find that this advantage extends to settings 
with broader competition structures.

Outside the equal-size setting, there is almost no prior literature 
on competition-based testing. 
An exception is the work of~\citet{gui2026adaptive}, who develop adaptive 
screening methods for subgroup selection in observational studies. 
They interpret the matched sets between treatment and control units, 
which vary in size, as competitions; however, their subgroup selection task is tailored to the causal 
application and does not carry over to the 
standard multiple testing problem.

\subsection{Key insight: compound p-values underlie SSS+ when competitions have equal size}
\label{sec:key-insight}

We close our preliminaries with a key observation: 
in the equal-size competition setting, SSS+ is equivalent to 
BH applied to a collection of compound p-values.

\begin{proposition}
    \label{prop:sss-compound-bh}
    For some $n\in \NN$, suppose that $n_j = n$ for all $j\in [m]$.
    Assume further that Assumption~\ref{assu:individual} holds, and that there are no ties among the scores within each competition.\footnote{Formally, here and below, we only require that our no ties assumptions hold almost surely, but we omit this for brevity.} For each $j\in[m]$, define
    \begin{equation}
        \label{eq:sss-compound-pval}
        \tilde p_j
        \coloneqq
        \left[
        \frac{
        \frac{1}{n}
        +\frac{1}{n}\sum_{\ell=1}^m\sum_{k=1}^n
        \ind\{W_{\ell,k}=1,\ S_{\ell,k}\ge S_{j,0}\}
        +m\ind\{W_{j,0}=0\}
        }{m}
        \right]\wedge 1.
    \end{equation}
    Then $\tilde p_1,\dots,\tilde p_m$ are compound p-values.  Moreover, for every $\alpha\in(0,1)$, 
    \begin{equation}
        \hat R^{\sss}
        =
        \hat R^\bh(\tilde p_1,\dots,\tilde p_m).
    \end{equation}
\end{proposition}
The proof of Proposition~\ref{prop:sss-compound-bh} is deferred to Appendix \ref{sec:appd-proof-sss-compound-bh}.  
To preview our results for varying-size competitions in 
Section~\ref{sec:cbt}, we provide a few high-level observations about the form of these compound p-values. First, we can rewrite $\tilde p_j$ as 
\begin{equation}
    \label{eq:sss-compound-pval-rewrite}
    \tilde p_j =
    \left[\frac{1+ \sum_{\ell=1}^m \sum_{k=1}^n \ind\{W_{\ell,k}=1, S_{\ell,k} \ge S_{j,0} \} }{mn} + \ind\{W_{j,0} =0 \} \right] \wedge 1.
\end{equation}
When $S_{j,0}$ is not the winner of the $j$th competition, we have $\ind\{W_{j,0} = 0\}=1$, so $\tilde p_j =1$.  This precludes the rejection of $H_j$ through either Bonferroni or BH.

The numerator of the fraction in \eqref{eq:sss-compound-pval-rewrite} 
resembles that of a pooled empirical p-value~\eqref{eq:pooled-pval-defn}, 
but in fact these differ in an important aspect.  Whereas the empirical p-value compares $S_{j,0}$ to all 
negative control scores $  \{S_{\ell,k}\}_{\ell \in [m], k \in [n]}$, 
our p-value compares a winning $S_{j,0}$ only to negative controls 
that also won their respective competitions. Conceptually, we treat the control winners themselves as negative controls for the test winners. 

The odds of a given control score winning a null competition is $n$ to 1.  Therefore, the contribution of each control 
winner to the numerator is weighted by 
$1/n$ in~\eqref{eq:sss-compound-pval}. 
The weighted contribution to the numerator 
from the event that a control won competition $\ell$ with 
a score of at least $S_{j,0}$ imitates the contribution
from the event that a true null test score won competition $\ell$ 
with a score of at least~$S_{j,0}$. 
The weighting means that the compound p-values attain the 
favorable $1/(mn)$-granularity discussed in 
Section~\ref{sec:prelim-cbt}.

\subsection{Additional related work}
\label{sec:related-work}

Multiple testing with negative controls is a popular 
approach, with both the individual and pooled paradigms well studied \citep[e.g.][]{gao2023simultaneous}. \citet{chakraborty2026power} compare the theoretical 
power of representative methods of the two paradigms, which they 
refer to as ``one-to-one'' and ``one-to-many'' augmentation
schemes, respectively. \citet{nguyen2020aggregation} showed the BH and SSS+ equivalence described in Proposition~\ref{prop:sss-compound-bh} in the case where $n_j=1$ for all $j$.  The authors prove that 
these are ``intermediate'' (rather than compound) p-values.\footnote{Specifically, 
they show that their p-values $(\pi_j)_{j\in[m]}$ satisfy 
$\PP(\pi_j  \le t) \le \kappa  \frac{m}{|\cH_0|} t$ for 
some inflation factor $\kappa \approx 3.24$. } Our methodology builds on the novel observation that these intermediate p-values are in fact compound valid, which inspires the compound p-value approach to multiple testing with broader competition structures. 

\citet{deng2026randomized} give an alternative unification of SSS+ and BH with either uniform p-values or conformal p-values.  Specifically, they prove that if $J\sim\unif(\cH_0)$ is independent of random variables $p_1,\ldots,p_m$ that satisfy a randomized PRDS condition, and if $p_J$ is superuniform, then BH applied to $p_1,\ldots,p_m$ has FDR validity.  The superuniformity of $p_J$ is similar to the compound 
p-value property~\eqref{eq:compound-pval-defn}, and their results are complementary to our goals 
of designing 
multiple testing methods involving competitions of varying size and structure, motivated by the messy realities of application areas
that employ  negative controls; see Section~\ref{sec:virtual-st-preview}.

Finally, we relate our contributions to the rapidly expanding line of work on multiple testing with {\em e-values}~\citep{ramdas2025hypothesis}, which are nonnegative 
random variables with expectation at most 1 under the 
null.~\citet{wang2022false}
propose the e-BH procedure (an e-value variant of BH), which 
controls the FDR at the nominal level regardless of the dependence 
among the e-values.
This FDR guarantee continues to hold for {\em compound e-values} $\tilde{e}_1,\ldots,\tilde{e}_m$, whose defining property is that $\sum_{j \in\cH_0} \EE[\tilde e_j] \le m$ 
\citep{wang2022false,ren2024derandomised}.  \citet{ren2024derandomised} leverage this fact to show 
an equivalence between SSS+ and  e-BH with compound 
e-values, en route to derandomizing the knockoff filter. 
We can similarly construct compound e-values for competition-based testing
in general competitions; see
Section~\ref{sec:fdr-cbt} for an example.

%% file: sections/cbt.tex

%
%
%

%

We now present a more general version of Proposition~\ref{prop:sss-compound-bh}, covering the setting where the number of controls~$n_j$ is 
allowed to vary across hypotheses.
The probability that a test unit wins a null competition
may now differ across $j$, and this heterogeneity must be addressed in order to retain compound validity. 
Our definition in~\eqref{eq:cbt-pval-thm} below therefore re-weights the contribution of a negative control winner in 
competition~$j$ by $1/n_j$.

\begin{theorem}
    \label{thm:cbt-compound}
    Suppose that Assumption~\ref{assu:individual} holds and that there are no ties among the scores within each competition. For each $j\in [m]$, choose 
    a deterministic offset $c_j \geq 0$ and define
    \begin{equation}
        \label{eq:cbt-pval-thm}
        \tilde p_j(c_j)
        \coloneqq
        \left[
        \frac{
            c_j
        +\sum_{\ell=1}^m\frac{1}{n_\ell}\sum_{k=1}^{n_\ell}
        \ind\{W_{\ell,k}=1,\ S_{\ell,k}\ge S_{j,0}\}
        +m\ind\{W_{j,0}=0\}
        }{m}
        \right]\wedge 1,
        \qquad j\in[m].
    \end{equation}
    Then, when $c_j \ge 1/n_j$ for all $j\in[m]$, we have that $\tilde p_1(c_1), \dots, \tilde p_m(c_m)$ are compound p-values. 
\end{theorem}
As in the equal-size setting, 
$W_{j,k} \coloneqq \ind\{S_{j,k} = \max( S_{j,0}, \dots, S_{j, n_j})\}$
denotes the winner indicator for~$S_{j,k}$ in competition $j$. The varying-size setting also permits a p-value offset term  specific to $H_j$, denoted $c_j$, 
which must be at least $1/n_j$.  The choice $c_j = 1/n_j$ 
may be thought of as the default, as it results in the smallest possible values that guarantee compound validity.
We denote
these default compound p-values as $\tilde p_j$, dropping the offset specification.  %
Our rationale for retaining the flexibility of larger offsets arises from the fact that the 
FDR of the BH procedure is easier to analyze 
for compound p-values that share a common offset $c_0$ across hypotheses; see Section~\ref{sec:fdr-cbt}.

Finally, we highlight that the default compound p-value $\tilde p_j$ can be as small as $1/(mn_j)$, providing a favorable effective granularity for multiplicity correction.  
As an example, consider a scenario where $S_{j,0}$ 
is the largest element of $\cS = \{S_{j,k}\colon j \in [m], 0 \le k \le n_j\}$. 
A Bonferroni correction applied to $(\tilde p_j)_{j \in [m]}$ 
rejects $H_j$ when $n_j \ge  \lceil 1/\alpha \rceil$. 
By contrast, for the individual empirical p-values $(p_j^\indiv)_{j \in [m]}$, 
the Bonferroni method only rejects $H_j$ when $n_j + 1  \ge  \lceil m/\alpha \rceil$. Thus the advantage of the compound p-values becomes more pronounced when the number of hypotheses $m$ 
is large relative to the number of individual controls.

\subsection{FDR control for competition-based testing}
\label{sec:fdr-cbt}

The main goal of this subsection is to derive FDR control results when BH is applied to the compound p-values $\tilde{p}_j(c_0)$ from Theorem~\ref{thm:cbt-compound}, with a common offset $c_0 \ge 1/\nmin$, where $\nmin \coloneqq \min_{j\in [m]} n_j$. 
Central to our analysis is the following equivalence,
inspired by Proposition~\ref{prop:sss-compound-bh},
which relates a general version of SSS+ with 
the BH procedure applied to $\tilde{p}_j(c_0)$. 

\begin{proposition}
    \label{prop:cbt-bh-equivalence}
    Suppose that the conditions of Theorem~\ref{thm:cbt-compound} hold. For a prespecified level $\alpha\in(0,1)$,  define the competition-based testing rejection set
    \begin{equation} 
    \label{eq:cbt-rejection-set}
        \begin{aligned} 
            \hat R^\cbt
            &\coloneqq
            \{j : S_{j,0} \ge T,\; W_{j,0}=1\},\\
            \text{where }
            T
            &\coloneqq 
            \min \bigg\{
            t\in\cS :
            \frac{
            c_0 
            +
            \sum_{\ell=1}^m \frac{1}{n_\ell}\sum_{k=1}^{n_\ell}
            \ind\{W_{\ell,k}=1,\; S_{\ell,k}\ge t\}
            }{
            1\vee
            \sum_{\ell=1}^m
            \ind\{W_{\ell,0}=1,\; S_{\ell,0}\ge t\}
            }
            \le \alpha
            \bigg\}.
        \end{aligned}
    \end{equation}
    Then, for every $c_0 \geq 0$,
    \begin{equation}
        \hat R^\cbt
        =
        \hat R^\bh\bigl(\tilde p_1(c_0   ),\dots,\tilde p_m(c_0 )\bigr). 
    \end{equation}
\end{proposition} 
By Theorem~\ref{thm:cbt-compound}, when $c_0 \geq 1/n_{\min}$, we have that $\tilde{p}_1(c_0),\ldots,\tilde{p}_m(c_0)$ are compound p-values, but Proposition~\ref{prop:cbt-bh-equivalence} shows that this condition is not necessary for the equivalence of the two rejections sets.  
Similar to the individual offsets $c_j$ in~\eqref{eq:cbt-pval-thm},
we can  think of  $c_0 = 1/\nmin$ as the default choice of a common offset: not only 
does this yield the smallest possible compound p-values, 
but $\hat R ^\cbt$ then also coincides with SSS+ for competitions of equal size.

Theorem~\ref{thm:inflated-fdr-control} is our main result on the (inflated) FDR control of $\hat{R}^{\mathrm{CBT}}$ applied to the compound p-values $\tilde{p}_1(c_0),\ldots,\tilde{p}_m(c_0)$ with default common offset.  For $t > 0$, we write $\mathfrak{h}_t\coloneqq \int_0^1 (1-x^t)/(1-x)\,dx$, recalling that this function is an analytic continuation of the harmonic series to $(0,\infty)$.

\begin{theorem}
    \label{thm:inflated-fdr-control}
    Suppose that the conditions of Theorem~\ref{thm:cbt-compound} hold and that the competition-winning scores
    $\bigl(\max_{0\le k\le n_j}S_{j,k}\bigr)_{j\in[m]}$, are distinct almost surely. Define $\bm{n} = (n_1,\ldots,n_m)$,
    \begin{equation}
        \Lambda(\bm n) \coloneqq\sum_{j=1}^m\frac{{\nmin}}{n_j+1}
        \quad \text{and} \quad 
        C(\bm n)\coloneqq
        \sum_{j=1}^m\frac{1}{h(n_j)+n_j/\nmin-1},
    \end{equation}
    where $h(d)\coloneqq|\{j\in[m]:n_j=d\}|$
    counts the number of hypotheses with $d$ controls in its competition.
    Then, for $c_0 = 1/\nmin$ and every $\alpha\in(0,1)$,
    \begin{equation}
        \fdr(\hat R^\cbt)
        \le
        \alpha\min\bigl\{
        {\mathfrak{h}_{\Lambda(\bm n)}},\ 
        C(\bm n)
        \bigr\}.
    \end{equation}
\end{theorem} 
The FDR inflation factor in Theorem~\ref{thm:inflated-fdr-control} is a minimum of two distinct upper bounds, both of which depend only on the vector $\bm n$ of competition sizes (which is known to the practitioner).   Since $\Lambda(\bm n) \leq m$, the first term in the minimum immediately implies the harmonic upper bound $\fdr \le \alpha \mathfrak{h}_m$ for BH applied to arbitrarily dependent strict~\citep{benjamini2001control}
and compound~\citep{armstrong2022false} p-values.  %
The second term, on the other hand, can be 
interpreted as a ``counting'' bound, since $C(\bm n) \leq \sum_{j=1}^m 1/h(n_j)$, which is the number of
distinct values in $\bm n$. 

\begin{remark}
    \label{rem:three-different-settings}
To illustrate the behaviors of the two components, we present 
several examples of competition-size vectors $\bm n$ and their resulting FDR bounds. First, consider the simplest case of equal-sized competitions, where every element 
in $\bm n$ is equal to some $n\in\NN$. Here, $C(\bm n) = 1$, so
$\fdr(\hat R^\cbt)\le\alpha$; this aligns with existing FDR control results of \citet{barber2015controlling} on 
the SSS+ procedure; see Section~\ref{sec:prelim-cbt}.  By contrast, $\Lambda(\bm n) = mn/(n+1)$, which yields a very loose upper bound that is almost the worst-case bound of \citet{armstrong2022false} for arbitrarily dependent compound p-values.

As a second example, we study a setting where one control group has a different size from the rest: $n_1 = \cdots = n_{m-1} = n$ and $n_m= N$, with 
$n \in \NN$ and $N = \kappa n$ for some positive $\kappa\neq 1$.  When $\kappa \in (0,1)$, we have $\mathfrak h_{\Lambda(\bm n)} \leq \mathfrak h_{1 + (m-1) \kappa}$, while $C(\bm n) = 1 + (m-1)\kappa / \bigl(1 - \kappa + (m-1)\kappa\bigr) \leq 2$.  If $m=2$ and $\kappa$ is close to 1, then the former bound is the stronger, but if $(m-1)\kappa \geq 3$ then the latter bound is better.  On the other hand, when $\kappa > 1$, we have $\Lambda(\bm n) \approx m$, while $C(\bm n) = 1 + 1 / \kappa$, providing a constant inflation factor that approaches 1 as $\kappa\to\infty$. 

Finally, consider the case where all $m$ control group sizes are distinct.  Here,  $C(\bm n) = \sum_{j=1}^m \nmin / n_j > \Lambda (\bm n)$. 
If $\Lambda(\bm n) \le1 $, then $\mathfrak h_{\Lambda(\bm n )} \le \mathfrak h _1 =1\le C(\bm n)$; if $\Lambda(\bm n)  > 1 $, then since $t \mapsto \mathfrak{h}_t$ is concave and its derivative is $\pi^2/6 - 1 < 1$ at $t = 1$, we see that $\mathfrak h_t \le t$ for all $t\ge 1$, so $\mathfrak h_{\Lambda(\bm n)} \le \Lambda(\bm n)  < C(\bm n )$. Hence,  $\mathfrak h_{\Lambda(\bm n)}  <  C(\bm n )$ uniformly, 
so the harmonic bound is always tighter.
\end{remark}

Prior work on dependent compound p-values has only yielded the worst-case FDR inflation factor of~$\mathfrak{h}_m$, even though a constant inflation factor is achievable under independence \citep{barber2025false}.  Theorem~\ref{thm:inflated-fdr-control} applies to compound p-values having a particular dependence structure that arises from our competitions, but reveals that much tighter FDR control is achievable (even nominal FDR control in the best, equal-size competitions case).

At first glance, the harmonic inflation factor may seem overly conservative. The following result clarifies that it is asymptotically attainable.
\begin{proposition}
    \label{prop:lower-bound}
    There exists a sequence of competition-based testing 
    problem instances satisfying the conditions of Theorem~\ref{thm:inflated-fdr-control}, indexed by $N\in \NN$, with competition sizes $\bm{n}_N$ and corresponding nominal levels $\alpha_N\in(0,1)$ such that
    \begin{equation}
        \label{eq:fdr-log-lower-bound}
        \fdr_N
        \ge
        \bigl( 1-o(1)\bigr)\,\alpha_N
        \log (1+\Lambda_N),\qquad \text{as } N\to\infty,
    \end{equation}
    where $\Lambda_N \coloneqq \Lambda(\bm{n}_N)$ and $\fdr_N \coloneqq \fdr(\hat R^\cbt_N )$.  Furthermore, $\log(1+ \Lambda_N) \rightarrow \infty$ as $N \rightarrow \infty$ for this sequence of problem instances. 
\end{proposition}
As a consequence of Proposition~\ref{prop:lower-bound}, there is no universal constant $\kappa > 0$ such that $\fdr(\hat R^\cbt) \le \kappa\alpha$.  By Proposition~\ref{prop:lower-bound} and Theorem~\ref{thm:inflated-fdr-control}, for this sequence of problem instances,
\begin{equation}
        \bigl( 1-o(1)\bigr)\,\alpha_N
        \log (1+\Lambda_N) \le \fdr_N 
        \le \alpha_N\bigl(\log(\Lambda_N) + 0.5773\bigr),
\end{equation}
so the worst-case FDR inflation factor is tightly controlled.  In fact, our numerical and real data experiments find that the FDR is often controlled at the nominal level; see Section~\ref{sec:sims}.  %

We conclude this subsection on FDR control with three variants and extensions of our competition-based testing methodology. %

\paragraph{Nominal FDR control when competitions are 
equally weighted.}

{Consider again the second setting in Remark~\ref{rem:three-different-settings}, where}
$n_1 = \cdots = n_{m-1} = n$ and $n_m= N$, with 
$n \in \NN$ and $N = \kappa n$ for some $\kappa > 0$.
{In this instance}, $C(\bm n)  = 1 + 1/\kappa$ for $\kappa >1$.  Thus, when $n$ and $N$ are close but not equal, $C(\bm n)$ is roughly~2. However, when $n = N$, we have $C(\bm n) =1$. 
This discontinuity raises the question of whether 
the proposed method achieves FDR control much closer to the 
nominal level than implied by Theorem~\ref{thm:inflated-fdr-control}
for settings where $\bm n$ are all nearly equal.

The following proposition provides a heuristic answer in 
the affirmative to the above question. 
We find that replacing the weights $1/n_j$ with $1/\nmin$ in the compound p-values of Theorem~\ref{thm:cbt-compound}, thereby enforcing equal weighting, leads to valid compound p-values 
that control the FDR at the nominal level via BH. 

\begin{proposition}
    \label{prop:cbt-pval-equal-weights}
    Suppose that the conditions of Theorem~\ref{thm:cbt-compound} hold. Define
    \begin{equation}
        \label{eq:cbt-pval-equal-weights}
        \tilde p_j^\ew
        \coloneqq
        \biggl[
        \frac{
             \frac{1}{\nmin}
        +\sum_{\ell=1}^m\frac{1}{\nmin}\sum_{k=1}^{n_\ell}
        \ind\{W_{\ell,k}=1,\ S_{\ell,k}\ge S_{j,0}\}
        +m\ind\{W_{j,0}=0\}
        }{m}
        \biggr]\wedge 1,
        \qquad j\in[m].
    \end{equation} 
    Then $\tilde p_1^\ew,\dots,\tilde p_m^\ew$ are compound p-values.  
    Moreover, suppose that the competition-winning scores are distinct almost surely. Then,  for every $\alpha\in(0,1)$, 
    \begin{equation}
        \fdr \bigl( R^\bh(\tilde p_1^\ew,\dots,\tilde p_m^\ew) \bigr) \le \alpha.
    \end{equation}
\end{proposition}

The compound validity of $\tilde p_1^\ew,\dots,\tilde p_m^\ew$ follows directly
from Theorem~\ref{thm:cbt-compound}, since $\tilde p_j^\ew  \ge \tilde p_j(1/\nmin )$ for all $j\in [m]$. When all values in $\bm n$ are large and close to $\nmin$,
we have $\tilde p_j^\ew  \approx \tilde p_j(1/\nmin ) $ and consequently expect 
$\fdr \bigl( R^\bh(\tilde p_1^\ew,\dots,\tilde p_m^\ew) \bigr) \approx \fdr \bigl( R^\bh(\tilde p_1(1/\nmin ),\dots,\tilde p_m(1/\nmin )) \bigr)$.
On the other hand, when $\nmin$ is not large, or when $\bm n$ 
varies widely, we expect $\tilde p_j^\ew$ to 
induce a more conservative procedure; this is verified by an empirical comparison of the two 
competition-based compound p-values in Appendix~\ref{sec:appd-deferred-sims}. 
We maintain our initial proposal as the primary 
approach and provide Proposition~\ref{prop:cbt-pval-equal-weights}
as a heuristic for its nominal level FDR-controlling behavior 
when the competition sizes are nearly equal.

A related approach is to subsample $\nmin$ controls 
from each individual control group via independent randomness, then compute
the SSS+ compound p-values~\eqref{eq:sss-compound-pval} on the ensuing equal-size competitions; conditional on 
the randomness in the subsampling step, the resulting p-values are compound valid.
We detail both approaches and prove their nominal FDR control in Appendix~\ref{sec:appd-equal-weights}.

\paragraph{FDR inflation bounds for general $c_0$.} 
The previous FDR bounds are given for the default offset $c_0 = 1/\nmin$. 
Here, we briefly discuss a separate FDR upper bound for  general $c_0 \ge 1/\nmin$.

\begin{proposition}
    \label{prop:kr-inflated-fdr-control-common-offset}
    Suppose that the conditions of Theorem~\ref{thm:cbt-compound} hold.
    For any $c_0\ge 1/\nmin$ and prespecified level
    $\alpha\in(0,1)$,
    \begin{equation}
        \fdr\bigl(\hat R^\bh(\tilde p_1(c_0),\dots,\tilde p_m(c_0))\bigr)
        \le
        \alpha c_0\nmin^2
        \int_0^1
        (1-z)^{c_0\nmin-1}
        \frac{\log\bigl(1/(1-z)\bigr)}{\log(1+\nmin z)}
        \,dz.
    \end{equation}
    In the case $c_0 = 1/\nmin$, the inflation factor becomes $\int_0^1 \frac{ \nmin \log(1/(1-z)) }{\log(1 + \nmin z)}\,dz $; moreover, the inflation factor converges to $1$ as $c_0 \to \infty$.
\end{proposition}
Unfortunately, when $c_0 = 1/\nmin$, 
this bound is  much  looser than that of 
Theorem~\ref{thm:inflated-fdr-control},
and thus the practicality of Proposition~\ref{prop:kr-inflated-fdr-control-common-offset} is currently limited. 
In Section~\ref{sec:shared_cbt}, we will see that the 
analogous compound p-values for competitions with shared controls 
naturally require much larger common offsets. Therefore,
we state Proposition~\ref{prop:kr-inflated-fdr-control-common-offset}  as a precursor to its more practical counterpart 
Proposition~\ref{prop:shared-cbt-ville-fdr-control}
in Section~\ref{sec:shared_cbt}. 
We provide the proof of 
Proposition~\ref{prop:kr-inflated-fdr-control-common-offset} and discuss its connection to the high-probability FDP bounds studied by~\citet{katsevich2020simultaneous} in Appendix~\ref{sec:appd-deferred-fdr-bounds}.

\paragraph{Establishing FDR control with compound e-values.}
As an alternative way of establishing nominal FDR control, we can instead turn to compound e-values and the e-BH procedure~\citep{wang2022false}. 
As a preliminary to the compound e-value construction, we first define a perturbed version of the threshold $T$ in~\eqref{eq:cbt-rejection-set}  with the default offset $c_0=1/\nmin$. For each $j\in[m]$,  write $S_{j,(0)}\coloneqq\max_{0\le k\le n_j}S_{j,k}$ and define 
\begin{equation}
    \label{eq:cbt-fdp-plus}
        T_j^+
        \coloneqq
        \min\bigg\{t\in\cS:
        \frac{
        \frac{1}{\nmin}
        +\frac{1}{n_j}\ind\{S_{j,(0)}\ge t\}
        +\sum_{\ell\ne j}\frac{1}{n_\ell}\sum_{k=1}^{n_\ell}
        \ind\{W_{\ell,k}=1,\ S_{\ell,k}\ge t\}
        }{
        1\vee\sum_{\ell\ne j}
        \ind\{W_{\ell,0}=1,\ S_{\ell,0}\ge t\}
        }
        \le\alpha\bigg\}.
\end{equation} 
The perturbed FDP estimator in the definition of $T_j^+$ treats the $j$th 
competition as always having a control winner. This inflates the 
estimated FDP and induces a more conservative rejection rule $T_j^+$, 
in the sense that $T_j^+ \ge  T$ almost surely, with equality  on the event 
that $S_{j,0}$ is not the winner of competition~$j$. 
Using these hypothesis-specific 
thresholds, we define the following collection of compound e-values. 

\begin{proposition}
    \label{prop:cbt-compound-evalues}
    Suppose that the conditions of Theorem~\ref{thm:cbt-compound} hold. 
    For a prespecified level $\alpha\in(0,1)$, define
    \begin{equation}
        \label{eq:cbt-compound-evalue}
        \tilde e_j
        \coloneqq
        \frac{
        m\ind\{W_{j,0}=1,\ S_{j,0}\ge T_j^+\}
        }{
        \frac{1}{n_j}
        +\sum_{\ell=1}^m\frac{1}{n_\ell}\sum_{k=1}^{n_\ell}
        \ind\{W_{\ell,k}=1,\ S_{\ell,k}\ge T_j^+\}
        }
    \end{equation}
    for $j \in [m]$.  Then $\tilde e_1,\dots,\tilde e_m$ are compound e-values.  
     Further, the rejection set of  e-BH applied to $\tilde e_1, \dots, \tilde  e_m$ is a 
subset of $\hat R^\bh\bigl( (\tilde p_1(1/\nmin), \dots, \tilde p_m(1/\nmin)) \bigr)$.
\end{proposition}

The incorporation of $T_j^+$ rather than $T$ is the key to establishing the compound validity of the e-values.  The fact that $\tilde{e}_1,\ldots,\tilde{e}_m$ are compound e-values means that applying e-BH controls the FDR at the nominal level \citep{wang2022false}.  On the other hand, the rejection set inclusion in the final part may be strict, as
$T_j^+ > T$ can happen with positive probability. 
In our simulation settings, we find that  the e-BH approach attains virtually identical power to the BH approach; see Appendix~\ref{sec:appd-deferred-sims}. 

\subsection{Relaxing the winner prerequisite}
\label{sec:qfinalist}

Thus far, we have only considered compound p-values built upon
competition winners.  On the one hand, if the test score beats all its control competitors, then this suggests evidence against the null hypothesis.  On the other hand, there are certain scenarios where this test winner prerequisite for 
rejecting $H_j$ can be detrimental to the power. The clearest case is 
when $n_j$ is very large; in this setting, 
there is a ``large-competition'' phenomenon where 
it is extremely likely that some negative control scores 
higher than the test score by pure chance, even when $H_j$ is false.  This phenomenon has been observed elsewhere in the competition-based testing literature; see, e.g.,  the ``TopGap'' 
phenomenon of~\citet{gui2026adaptive}.

To address this issue, prior work on competition-based testing~\citep[e.g.][]{emery2019controlling}
has allowed for more general competition outcomes that 
track the placement of the test score to 
measure its evidence against its null. Importantly, 
this approach enables the rejection of $H_j$ when 
$S_{j,0}$ is large relative to most of its negative controls, but is not necessarily the competition winner.
We generalize these ideas by introducing a {\em $q$-finalist} condition 
when comparing each test score to its negative controls. 
Specifically, we call $S_{j,k}$ a $q_j$-finalist
if it scores in the top $q_j \in [n_j+1]$ in competition $j$, i.e.~if $1 + \sum_{a=0}^{n_j} \mathbbm{1}\{S_{j,a} > S_{j,k}\} \leq q_j$.  In order to construct corresponding compound p-values, it is convenient to introduce $q$-finalist indicator variables $W_{j,k}^{(q_j)} \coloneqq \ind\bigl\{1+\sum_{a=0}^{n_j}
\ind\{S_{j,a}>S_{j,k}\}\le q_j\bigr\}$.  
\begin{proposition}
    \label{prop:qfinalist-cbt-compound}
    For $j\in[m]$, let $q_j\in[n_j+1]$ and $c_j \ge 1/n_j$.
    Denote $\bm q \coloneqq (q_1,\dots,q_m)$ and suppose that the conditions of Theorem~\ref{thm:cbt-compound} hold. Define
    \begin{equation}
        \label{eq:qfinalist-cbt-pval}
        \tilde p_j^{(\bm q)}(c_j)
        \coloneqq
        \biggl[
        \frac{
            c_j
        +\sum_{\ell=1}^m\frac{1}{n_\ell}\sum_{k=1}^{n_\ell}
        \ind\{W_{\ell,k}^{(q_\ell)}=1,\ S_{\ell,k}\ge S_{j,0}\}
        +m\ind\{W_{j,0}^{(q_j)}=0\}
        }{m}
        \biggr]\wedge 1,
        \qquad j\in[m].
    \end{equation}
    Then $\tilde p_1^{(\bm q)}(c_1),\dots,\tilde p_m^{(\bm q)}(c_m)$ are compound p-values.
\end{proposition}
Similar to the winner-only construction, the contribution of the event $\{W_{\ell,k}^{(q_\ell)}=1, S_{\ell, k} \ge S_{j,0}\}$ to the numerator of our p-values is weighted by $1/n_\ell$, due to the symmetry under $H_\ell$ of the test and control units that achieve $q_\ell$-finalist status. 
The ability to choose $q_j$ provides flexibility in the number of high-scoring competitors employed in the second stage of the competition-based testing procedure (Figure~\ref{fig:negative-control-paradigms}).
One straightforward choice is to take $q_j \propto n_j$, which selects scores that rank within 
some top percentile of their competition; another is to  take $q_j =1$ for all $j\in[m]$, coinciding 
with the initial proposal.

We conclude our treatment of $q$-finalist compound p-values by providing a bound on the FDR of BH for equal-size competitions. %

\begin{proposition}
    \label{prop:fdr-qfinalist}
    Suppose that the conditions of Proposition~\ref{prop:qfinalist-cbt-compound} hold with $n_1=\cdots=n_m=n$ and $q_1=\cdots=q_m=q\in[n]$, and denote $\tilde p_j^{(q)}$ as the compound p-values in~\eqref{eq:qfinalist-cbt-pval} with  $c_j  =1/n$.
    Then, for any prespecified level $\alpha \in (0,1)$,  
    \begin{equation}
        \fdr\bigl(\hat R^\bh(\tilde p_1^{(q)},\dots,\tilde p_m^{(q)})\bigr)
        \le
        \alpha\frac{n}{n-q+1}.
    \end{equation}
\end{proposition}

Proposition~\ref{prop:fdr-qfinalist} ensures that when $q = \lfloor\gamma n\rfloor$ (the top $\gamma$ fraction of scores per competition), the FDR inflation factor is at most $1/(1-\gamma)$.
This conclusion is similar to the FDR upper bounds of \citet{dohler2018new} and \citet{emery2019controlling} (though these works do not concern compound p-values).

%% file: sections/shared_cbt.tex

%
%
%
%
%

In Section~\ref{sec:cbt}, we exclusively considered competitions 
based on individual control groups whose scores satisfy 
Assumption~\ref{assu:individual}.
If we were to try to form individual control groups from a shared negative control set that has 
fewer than $m$ scores, then we would be unable to 
allocate each hypothesis its own control score.   
A natural solution is to assign to certain test units, 
whose distributions are plausibly identical or similar, 
an allocation of negative controls to share 
among themselves.  This relaxes the requirement of per-hypothesis assignment and allows for powerful multiple testing to be conducted 
even when the number of controls is significantly outnumbered by the number of hypotheses.

As a preliminary to introducing this new {\em block-specific} negative control paradigm, let $B_1,\dots,B_G$ be a partition
of $[m]$, with $b_g\coloneqq|B_g|$, and write $g(j)=g$ when $j\in B_g$. 
Each block $B_g$ is assigned a shared control group 
$C_g\coloneqq(S_{g,1},\dots,S_{g,n_g})$, so that the competition for 
each $H_j$ with $j\in B_g$ compares $S_{j,0}$ against the 
$n_g$ control scores in $C_g$.
\begin{assumption}[Block-specific exchangeability]
    \label{assu:block-specific}
    For each $g \in [G]$, conditional on
    \begin{equation}
        \cG_{-g}^{\blk}
        \coloneqq
        \sigma\bigl(
        \{S_{j,0}:j\in[m]\setminus(B_g\cap\cH_0)\}
        \cup
        \{S_{h,k}:h\in[G]\setminus\{g\},\ k\in[n_h]\}
        \bigr),
    \end{equation}
    the concatenated  null test and control scores $ (S_{j,0})_{j\in B_g\cap\cH_0}\cup(S_{g,k})_{k=1}^{n_g}$ are exchangeable.
\end{assumption}

Under this  assumption, we can construct compound p-values for the 
block-specific competition setting. These p-values follow a similar 
form to the individual compound p-values proposed in Section~\ref{sec:cbt}.
For each $g\in[G]$ and  $j\in B_g$, first define the block-specific winner indicators:
\begin{equation}
    \label{eq:block-winner}
    W_{j,k}^{\blk}\coloneqq 
    \begin{cases}
        \ind\{S_{j,0}=\max(S_{j,0},S_{g,1},\dots,S_{g,n_g})\}, & k=0\\
        \ind\{S_{g,k}=\max(S_{j,0},S_{g,1},\dots,S_{g,n_g})\},
         & k\in[n_g].
    \end{cases}
\end{equation}
Intuitively, each $H_j$ defines its own competition using the shared controls 
and $S_{j,0}$, with $W_{j,k}^\blk$ indicating winner status of this 
specific competition. Although the competitors themselves are shared, 
the competitions are distinct across hypotheses: test units are not 
competing against each other, instead entering their 
own exclusive competition with these shared controls.
We now provide the compound p-value construction below.

\begin{theorem}
    \label{thm:shared-cbt-compound}
    Suppose that Assumption~\ref{assu:block-specific} holds and there are no ties among the scores within each competition. For each $j\in[m]$, define
    \begin{equation}
        \label{eq:shared-cbt-pval}
        \tilde p_j^{\blk}
        \coloneqq
        \left[
        \frac{
        \frac{b_{g(j)}}{n_{g(j)}}
        +\sum_{h=1}^G\sum_{\ell\in B_h}\frac{1}{n_h}\sum_{k=1}^{n_h}
        \ind\{W_{\ell,k}^{\blk}=1,\ S_{h,k}\ge S_{j,0}\}
        +m\ind\{W_{j,0}^{\blk}=0\}
        }{m}
        \right]\wedge 1.
    \end{equation}
    Then $\tilde p_1^{\blk},\dots,\tilde p_m^{\blk}$ are compound p-values.
\end{theorem}
The block-specific 
$\tilde p_j^\blk$ reduces to the individual 
$\tilde p_j$~\eqref{eq:cbt-pval-thm} in the case where 
every block of hypotheses is a singleton. In that setting, the block-specific 
winner indicators coincide with individual competition winner indicators,
and the offset term also matches as $b_{g(j)}=1$.

\subsection{FDR control for block-specific competition-based testing}
\label{sec:shared-cbt-fdr}

In Section~\ref{sec:fdr-cbt}, 
Proposition~\ref{prop:kr-inflated-fdr-control-common-offset}
provides an upper bound on the FDR inflation when a common offset $c_0$ is used, exhibiting only a small inflation when $c_0$ is large.  The block-specific setting is motivated by scenarios where control units
are in short supply relative to test units; in such cases, the analogous
default choice of common offset $\max_{ g\in [G]} b_g/n_g$
will typically be large. For example, in the virtual ST data analysis 
conducted in Section~\ref{sec:application-virtual-st}, 
this common offset ranges approximately from 19.5 to 216.4, depending 
on the algorithm used to form the blocks.
Therefore, we expect a similar FDR bound to be useful 
in the block-specific setting, and provide a similar result below.

\begin{proposition}
    \label{prop:shared-cbt-ville-fdr-control}
    Suppose that the conditions of Theorem~\ref{thm:shared-cbt-compound} hold. 
    Denote
    \[
        \rmax\coloneqq\max_{g\in[G]}\frac{b_g}{n_g}
    \]
    and define the p-value
    \begin{equation}
        \label{eq:shared-cbt-common-offset-pval}
        \tilde p_j^{\blk}(\rmax)
        \coloneqq
        \biggl[
        \frac{
        \rmax
        +\sum_{h=1}^G\sum_{\ell\in B_h}\frac{1}{n_h}\sum_{k=1}^{n_h}
        \ind\{W_{\ell,k}^{\blk}=1,\ S_{h,k}\ge S_{j,0}\}
        +m\ind\{W_{j,0}^{\blk}=0\}
        }{m}
        \biggr]\wedge 1.
    \end{equation}
    Then $\tilde p_1^{\blk}(\rmax),\dots,\tilde p_m^{\blk}(\rmax)$ are compound p-values. Moreover, for any prespecified level $\alpha\in(0,1)$, 
    \begin{equation} \fdr\bigl(\hat R^{\bh} (\tilde p_1^{\blk}(\rmax),\dots,\tilde p_m^{\blk}(\rmax))\bigr)\\
            \le
            \alpha\int_0^1
            \frac{\log\bigl(1/(1-z)\bigr)}{\rmax\log(1+z/\rmax)}
            \,dz. 
    \end{equation}
\end{proposition}

\begin{corollary}
    \label{coro:limsup-fdr-pi2alpha}
        Under the conditions of Theorem~\ref{thm:shared-cbt-compound},
        \[
            \limsup_{\rmax\to\infty}
            \fdr\bigl(\hat R^{\bh} (\tilde p_1^{\blk}(\rmax),\dots,\tilde p_m^{\blk}(\rmax))\bigr)
            \le \frac{\pi^2}{6}\alpha.
        \]  
\end{corollary}
To verify the utility of the above bound and gain intuition about the rate of convergence to $\pi^2/6\approx 1.645$ in Corollary~\ref{coro:limsup-fdr-pi2alpha}, we may compute that the upper bound on the inflation factor in Proposition~\ref{prop:shared-cbt-ville-fdr-control}  is $2.099$, $1.882$, and $1.743$ when $\rmax = 1$, $2$, and $5$ respectively.  
In general, when $\rmax$ is large, we may achieve nearly nominal FDR control by applying BH to  
$(\tilde p_j^{\blk}(\rmax))_{j\in [m]}$  at the corrected 
level $6\alpha/\pi^2$.  In fact, we find empirically that the 
uncorrected BH procedure generally reports an FDR of 
at most~$\alpha$ in our numerical experiments; see Section~\ref{sec:sims}.  
In our virtual ST and CRISPR applications (Section~\ref{sec:real_data}), the range of inflation factors 
is approximately $(1.647, 1.670)$ and $(1.657,1.766)$, respectively. 
For the virtual ST application where we have access to ground truth gene expression levels, we also find empirically that the FDR is controlled at the nominal level. 

Similar to our results in Section~\ref{sec:cbt}, we note that the compound p-values $\bigl(\tilde p_j^{\blk}(\rmax)\bigr)_{j\in [m]}$ remain 
compound valid if $\rmax$ were instead defined to be a fixed common offset 
greater than or equal to $\max_{g} b_g / n_g$; 
the FDR inflation factor presented in 
Proposition~\ref{prop:shared-cbt-ville-fdr-control} decreases to $\pi^2/6$ as $\rmax\to\infty$.  

%
%

%

%% file: sections/sims.tex

\subsection{Setup and data-generating process}
\label{sec:dgp}

We examine the empirical performance of our proposed methodology 
under a hierarchical Gaussian mixture model inspired by 
the conformal outlier detection 
literature~\citep{bates2023testing, lee2025full}.
Gaussian data are generated in different clusters, with each mixture component corresponding to a null distribution.
The goal is to identify outliers with respect to the 
mixture null distribution, using the negative control 
data points across all clusters.  
The following setup is common to all experiments; more details can be found in Appendix~\ref{sec:appd-details-on-dgp}.

Before any replications, we sample 25 independent cluster  
centers $K_g\sim\unif([-3,3]^{25})$, $g\in[25]$, and write 
$\cK\coloneqq\{K_1,\ldots,K_{25}\}$.  We then sample a $\cN_{25}(0,I)$ random vector, before centering and rescaling it to have mean zero and unit  variance and denoting it by $\bm u\coloneqq(u_1,\ldots,u_{25})$.  Both $\cK$ and $\bm u$ are 
sampled once and held fixed across all experiments. 
For each test unit $j \in [m]$ with cluster label $G_j \in [25]$, we generate
\begin{equation}
    \label{eq:the-dgp}
    X_{j,0} \coloneqq \begin{cases}
        K_{G_j} + \exp(\tau u_{G_j}) Z_{j,0}  & \text{under }H_j,\\
        K_{G_j} + \exp(\tau u_{G_j})\sqrt\eta\, Z_{j,0}& \text{under the alternative},
    \end{cases}
\end{equation}
where $\tau > 0$, $\eta > 1$ and $Z_{1,0},\ldots,Z_{m,0} \stackrel{\iid}{\sim} \cN_{25}(0,I)$. 
The sampling mechanism of $G_j$ varies across simulation settings.  We also generate $Z_{j,1},\ldots,Z_{j,n_j} \stackrel{\iid}{\sim} \cN_{25}(0,I)$ and compute the $k$th negative control for $j$ as for $X_{j,0}$ under $H_j$, but with $Z_{j,0}$ replaced with $Z_{j,k}$.  Thus the distribution of $X_{j,k}$ only depends on $j$ through its cluster label $G_j$. 

Similarly to~\citet{bates2023testing}, we fit a OneClassSVM with default arguments, using the Python package \texttt{scikit-learn}~\citep{scikit-learn}, on an independent held-out training dataset of 2,000 independent observations drawn from the uniform-weighted null mixture distribution. The resulting fitted model   assigns higher non-conformity scores to inputs that are more likely to be anomalous---we denote the score of the unit $X_{j,k}$ as $S_{j,k}$. 

Unless otherwise stated, we set $\tau=0.25$, vary $\eta\in\{2,2.5,\ldots,6\}$
and report the average metrics over 1,000 Monte Carlo replications  
(with error bands of two standard errors displayed where applicable).
Finally, we always apply BH at $\alpha=0.1$ and Bonferroni at $\alpha=0.2$, even when the former only has inflated FDR control guarantees 
as in Sections~\ref{sec:sims-heterog} and~\ref{sec:sims-block}.
For BH, we use the default common-offset version of the competition 
p-values $\bigl(\tilde p_j(1/\nmin)\bigr)_{j\in [m]}$ from~\eqref{eq:cbt-pval-thm} and  $\bigl(\tilde p_j^\blk(\rmax)\bigr)_{j\in [m]}$ from~\eqref{eq:shared-cbt-common-offset-pval}; 
for Bonferroni, we use the individual-offset competition p-values $\bigl(\tilde p_j(1/n_j)\bigr)_{j\in [m]}$ from~\eqref{eq:cbt-pval-thm}.

\paragraph{Baselines.} We compare the proposed competition-based testing (CBT)
methodology with the following existing methods:
\begin{itemize}
    \item {\em Individual empirical p-values.} Each $S_{j,0}$ uses 
    its own individual control group to construct the empirical p-values, 
    as in~\eqref{eq:indiv-pval-defn}.
    \item {\em Pooled empirical p-values.} We pool all negative controls 
    and use the resulting eCDF to construct the empirical p-values, 
    as in~\eqref{eq:pooled-pval-defn}.
    \item {\em Block empirical  p-values.} For competitions 
    with shared controls, 
    we can construct pooled empirical p-values {\em within} each block:
    \begin{equation}
        \label{eq:block-pval-defn}
        p_j^{\blk} = \frac{1 + \sum_{k=1}^{n_{g(j)}} \ind\{  S_{g(j),k}\ge S_{j, 0} \} }{n_{g(j)}+1}.
    \end{equation}
    Under Assumption~\ref{assu:block-specific}, these are superuniform under $H_j$. If we further assume no ties between scores in 
    the same block, then $(p_j^\blk)_{j=1}^m$ satisfies the PRDS condition,
    implying nominal FDR control via BH; see 
    Appendix~\ref{sec:appd-block-pvalues-prds} for further details.
    \item {\em Compound empirical p-values.} When hypotheses 
    have their individual control groups, one can construct 
    the following compound p-values, which are an extension of those proposed by~\citet[][Section~3.4]{barber2025false}:
    \begin{equation}
        \label{eq:rb-pval-defn}
        \tilde p_j^\rb 
        \coloneqq \frac 1 m \sum_{\ell=1}^m
        \biggl(
        \frac{
        \ind\{S_{\ell,0}\ge S_{j,0}\}
        +\sum_{k=1}^{n_\ell}\ind\{S_{\ell,k}\ge S_{j,0}\}
        }{n_\ell+1}
        \biggr).
    \end{equation}
    By examining the original proof (which assumes  $n_1= \cdots = n_m$), we see that these are independent 
    compound p-values (conditional on all hypotheses' eCDFs); 
    thus, BH applied to $(\tilde p_j^\rb)_{j\in[m]}$ attains an FDR of at 
    most $1.93\alpha$. 
\end{itemize}

\subsection{Simulations for equal-size competitions}
\label{sec:sims-homog}

First, we conduct experiments in the equal-size competition setting. 

\paragraph{Additional setup.}

We sample $m=500$ test units, of which $20\%$ are nonnull. 
Each test unit has $n_j=25$ null controls sampled from its own cluster. 
Cluster labels are generated independently, with
null cluster labels distributed as $\unif([25])$ and nonnull labels as either 
$\unif([25])$ or sampled via $\PP(G=g)\propto\exp(u_g)$, 
corresponding to the ``Uniform $K$'' and ``Weighted $K$'' columns.

\paragraph{Results.}
Figure~\ref{fig:homog-n25-unif_and_weighted} presents the FDR, FWER and power curves of the methods above.  All methods control the FDR and FWER at the nominal level (as proved for CBT in Theorem~\ref{thm:inflated-fdr-control} and Proposition~\ref{prop:fwer-compound}), but the CBT compound p-values have greatest power at all values of $\eta$.  The closest competitor in terms of power is the pooled empirical approach, with the compound empirical approach next.  The individual empirical p-values suffer from poor granularity and are almost entirely powerless.

\subsection{Simulations for varying-size competitions}
\label{sec:sims-heterog}

\paragraph{Additional setup.}
We sample $m=500$ tests with $20\%$ nonnulls, drawing cluster labels independently from $\unif([25])$ under both the null and alternative. Each test's control units are sampled from its own cluster. We have two simulation settings for the $n_j$'s: {\em binary}, where $n_{j} \in \{ n_{\text{low}}, n_{\text{high}} \}$, and {\em smooth}, where $n_j$ is distributed as $ n_{\text{low}} + \Bin(n_{\text{high}}-n_{\text{low}}, \rho_g) $ for a cluster-specific probability $\rho_g$.

We take $n_{\mathrm{low}}=10$ and $n_{\mathrm{high}}=50$ and consider two cases:
\begin{enumerate}[1.]
    \item {\em Higher variance clusters have more controls.} In the binary setting, $n_j=n_{\mathrm{high}}$ when $u_g$ is at or above $\mathrm{median}(\bm u)$ and $n_j=n_{\mathrm{low}}$ otherwise; in the smooth setting, $\rho_g=\operatorname{logit}^{-1}(u_g)$.
    \item {\em Lower variance clusters have more controls.} The assignments are reversed in the binary setting; in the smooth setting, $\rho_g=\operatorname{logit}^{-1}(-u_g)$ instead.
\end{enumerate}
We call these Cases~1 and~2 respectively in our figures.   For Case 2, $\bm n$  is randomly generated per replicate; however, we still enjoy existing validity guarantees conditional on $
\bm n$ (and hence unconditionally). Case~1 does not have random $\bm n$.

\begin{figure}[tbp!]
    \centering
    \includegraphics[width=\textwidth,clip]{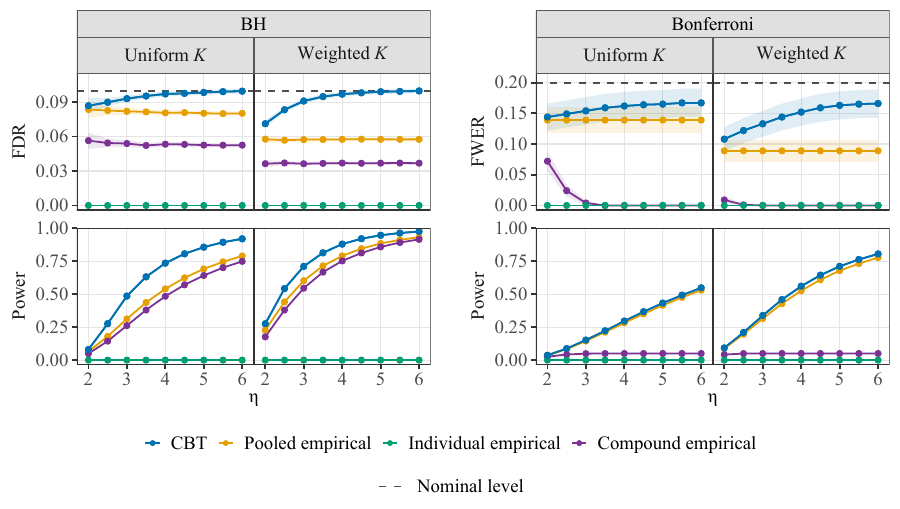}
    \caption{Power, FDR and FWER curves for the equal-size competition experiments.  The left plots apply the BH procedure to the relevant compound (or strict) p-values, while the right plots apply the Bonferroni procedure.  ``Uniform K'' and ``Weighted K'' refer to the cluster label generating mechanisms described in the text.}
    \label{fig:homog-n25-unif_and_weighted}
\end{figure}

\paragraph{Results.}
Figure~\ref{fig:heterog-binary-10-50} visualizes the FDR, FWER and power curves in the binary $n_j\in \{10, 50\}$ setting, while the smooth $n_j \in [10, 50]$ results (which are qualitatively similar) are shown in Figure~\ref{fig:heterog-smooth-10-50} in Appendix~\ref{sec:appd-binom-heterog}.  In Case~2, the pool of controls overrepresents low variance datapoints
and makes higher-variance nulls appear anomalous; this inflates 
the FDR and FWER of the pooled approach well above the nominal level.  The three remaining methods continue to enjoy Type~I error control (both FDR and FWER), but CBT achieves much greater power at all values of $\eta$ than the other two approaches.  On the other hand, in Case~1, the pooling approach is conservative, as the overrepresentation of high-variance nulls in the pooled control set means the nonnull test points from lower-variance clusters seem less 
anomalous. In this case, all four methods have nominal Type~I error control, and again  CBT is most powerful at all values of $\eta$.

\begin{figure}[tbp!]
    \centering
    \includegraphics[width=\textwidth,clip]{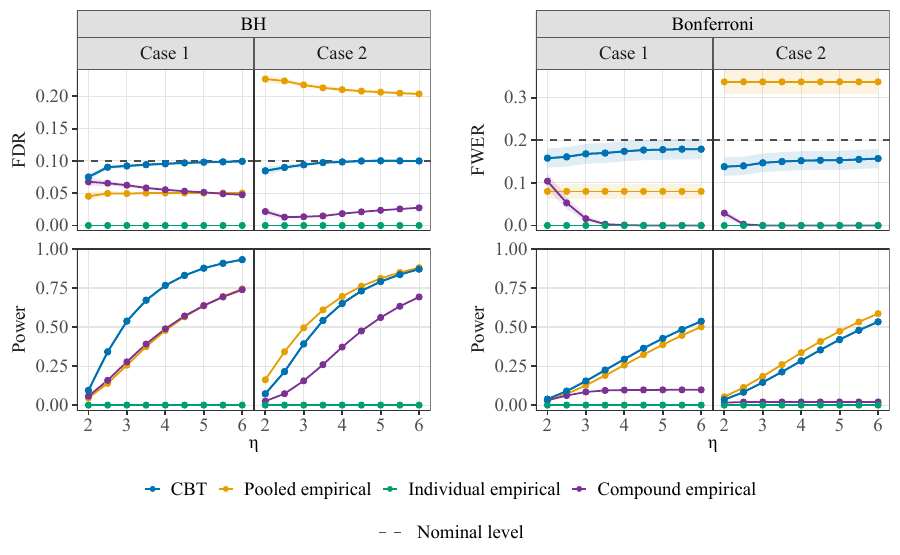}
    \caption{Power, FDR and FWER curves for the varying-size competition experiments, where 
    $n_j\in \{10, 50\}$ for all $j\in[m]$. Cases~1 and~2 refer to scenarios where higher-variance clusters have more controls and lower variance clusters have more controls, as described in the text.}
    \label{fig:heterog-binary-10-50}
\end{figure}

\subsection{Simulations for competitions with shared controls}
\label{sec:sims-block}

We now consider a variation of the data-generating process 
in which the negative controls are collected all at once, 
instead of specifically sampling a collection for each test unit. 
This setup is more similar to  the numerical 
experiments conducted in prior work on conformal outlier detection~\citep{bates2023testing}.

\paragraph{Additional setup.} 
For each replicate, we sample $m=1000$ test units with $10\%$ nonnulls, drawing cluster labels independently from $ \unif([25])$ under both the null and alternative; the data are then generated following~\eqref{eq:the-dgp}.
For the controls, we sample $N_{\mathrm{total}}\in\{2000,4000\}$ 
units, again with uniformly chosen cluster labels, 
and follow the null process in~\eqref{eq:the-dgp}. 
The shared control group $C_g$ consists of the 
units that are generated from the $g$th component of the 
Gaussian mixture model; every test unit from cluster $g$ 
shares $C_g$. 
We compare our block-specific CBT p-values~\eqref{eq:shared-cbt-pval} 
against the pooled and block empirical p-values.

We run two separate experiments: one with oracle blocks and the other with  estimated blocks. 
Oracle blocks are formed with knowledge of the cluster labels $G$. 
Estimated blocks are learned by running the 
$K$-means algorithm (with $K=25$) via \texttt{scikit-learn} on the 
independent hold-out training set used to train the SVM. 
We then assign the test and control units to the closest 
learned cluster center (with some post-processing steps as described in Appendix~\ref{sec:appd-details-on-dgp}) 
to form the block-specific competition format. 

\paragraph{Results.}
In Figure~\ref{fig:shared}, the results for oracle and 
estimated blocks are indistinguishable.  All three methods have nominal Type~I error control, but for both choices of $N_{\mathrm{total}}$, CBT attains 
higher power than the baseline methods at all values of $\eta$. For $N_{\mathrm{total}} = 4000$, the 
block empirical p-values achieve power approaching that of CBT 
for mid-to-high signal strengths---due to adequate p-value granularity---but 
continue to struggle with weaker signals.

\begin{figure}[b]
    \centering
    \includegraphics[width=\textwidth,clip]{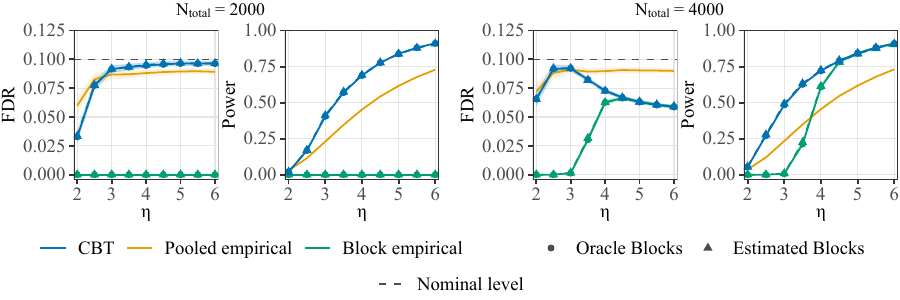}
    \caption{Power and FDR curves for competitions with shared control groups of size $N_{\mathrm{total}} = 2000$ (left) and $N_{\mathrm{total}} = 4000$ (right).}
    \label{fig:shared}
\end{figure}

\subsection{Simulations for $q$-finalists: the effect of $q$}

Finally we turn to the $q$-finalist variation of CBT to investigate the effect of $q$ on performance. We use the equal-size setting in Section~\ref{sec:sims-homog}, restricted to independent, uniform cluster-label sampling and a $20\%$ nonnull proportion. We compare $q\in\{1,2,5,10,25\}$ for $(m,n)\in\{(500,50),(50,500)\}$.

\paragraph{Results.}
Figure~\ref{fig:q-finalist-bh} illustrates that even though Proposition~\ref{prop:fdr-qfinalist} leaves open the possibility of mild FDR inflation, empirically we see that the FDR is controlled at the nominal level in all cases.  Interestingly, $q=1$ is not always the best choice for the BH procedure in terms of power.  When $(m,n) = (50, 500)$, intermediate values of $q$ outperform either $q=1$ or $25$; this agrees 
with the large-competition phenomenon discussed in 
Section~\ref{sec:qfinalist}. On the other hand, 
when $(m,n) = (500, 50)$, we find that $q=1$ is the best choice, with power typically decreasing for fixed $\eta$ as $q$ grows.

\begin{figure}[tbp!]
    \centering
    \includegraphics[width=\textwidth,clip]{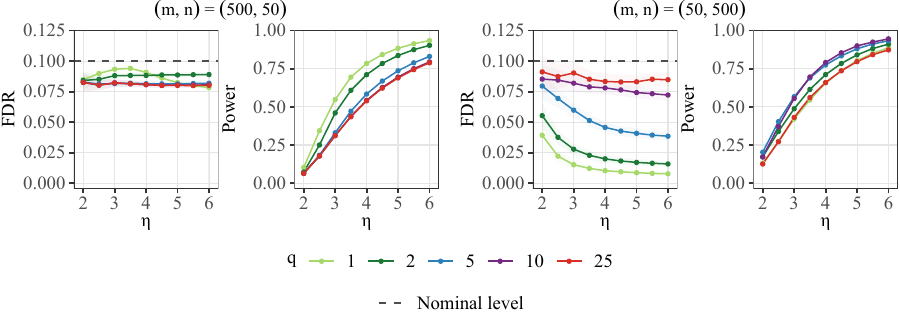}
    \caption{Power and FDR curves for the  $q$-finalist experiments  with equal-size competitions, where $q$ is varied and $(m,n) = (500,50)$ or $(m,n) = (50,500)$.} 
    \label{fig:q-finalist-bh}
\end{figure}

%% file: sections/real_data.tex

%

We examine the performance of our methods 
on publicly available datasets for 
two real-data applications: the virtual ST example 
previewed in the introduction 
and a CRISPR screen analysis that employs negative control genes.

\subsection{Application: virtual ST on gastric cancer tissue data}
\label{sec:application-virtual-st}

The public data for the gastric cancer virtual ST example from Section~\ref{sec:virtual-st-preview} are provided online by~\citet{jin2026multiscale}. Our goal is to locate areas (e.g., pixels) in 
the tissue sample with tertiary lymphoid structures (TLSs), which are associated with heightened anti-tumor immunity and are of clinical research interest.
The corresponding genetic target is the expression level of a TLS ``meta-gene'',
a curated collection of relevant genes. Virtual ST trains a blackbox model to predict expression levels across the tissue by taking a 
H\&E-stained histology slide as input (Figure~\ref{sec:virtual-st-preview}).

~\citet{jin2026multiscale} reformulate the scientific challenge as a multiple testing problem for the purpose of uncertainty quantification.  The $j$th pixel (or superpixel, depending on the resolution) in the image has a covariate--response tuple $(\bm X_{j,0}, Y_{j,0})$, where $\bm X_{j,0}$ represents a histology feature vector of the pixel's surrounding area  and $Y_{j,0}$ is the desired but unknown true meta-gene expression. Given the predicted regression function $\hat f$, the $j$th pixel has an associated predicted expression $\hat Y_{j,0} \coloneqq \hat f(\bm X_{j,0})$.  The baseline expression level is $c=0.05$, so similar to the conformal selection framework of \citet{jin2023selection}, the (random) hypothesis for pixel $j$ is
\begin{equation}
    H_j \colon Y_{j,0} \leq c;
\end{equation}
this encodes the belief that location $j$ expresses the meta-gene only negligibly, if at all.
The fact that~$H_j$ is a random hypothesis means that the setting is not explicitly covered by the theoretical results of this paper.  However, the relevant methods for this application can all be extended to this setting; we defer a rigorous treatment  to Appendix~\ref{sec:appd-random-compound}.

To test $H_j$ in a nonparametric manner,~\citet{jin2026multiscale} 
propose UTOPIA, which employs negative control locations with data 
$(\bm X_i, Y_i, \hat Y_i)$, collected from two small ROIs 
where ground truths are available from previous ST measurements. 
As we know the ground truth $Y_i$, one can learn the empirical null distribution of $\hat Y_{j,0}$ from the pool 
of null negative controls $(\hat Y_i)_{i:Y_i\le c}$.
However, as we discussed in Section~\ref{sec:virtual-st-preview},
the controls are extremely heterogeneous, as the tissue
can be classified into 15 distinct histological clusters. 
UTOPIA foregoes pooling the controls due to the heterogeneity, 
instead opting for a cluster-by-cluster analysis, with 
the block empirical p-values~\eqref{eq:block-pval-defn} made for each test score belonging to 
cluster~$g$ with null negative control scores also from cluster $g$.  Nevertheless, the controls of cluster $g$ still exhibit some heterogeneity, requiring  
the scores to be altered via an ad hoc, data-driven correction for the resulting p-values to appear superuniform. In addition, the 
resolution of the p-values from clusters with small numbers of control scores can be insufficient 
for rejection of the null via BH; hence, UTOPIA chooses to run BH for each cluster, rather 
than globally.

We implement CBT  with shared controls on the 
gastric cancer tissue data, using
smaller blocks  arising from a further covariate-based partition of the histological clusters. These smaller blocks are designed 
to assign only the most relevant control scores $(C_g)_{g\in[G]}$ to 
groups of hypotheses $(B_g)_{g\in [G]}$, providing a more 
plausible exchangeability guarantee across each $C_g$ and $B_g$.

\paragraph{Methods.}
To obtain the test and control groups, we implement two tailored algorithms: adaptive recursive coarsening (ARC), inspired by coarsened exact matching (CEM;~\citealp{iacus2012causal}); and greedy grouping (GG), inspired by optimal full matching and the \texttt{optmatch} software~\citep{hansen2007optmatch}. 
Details of these algorithms are provided in Appendix~\ref{sec:appd-matching-algorithms}.
Both return partitions that can reasonably be assumed to satisfy 
block-specific exchangeability (Assumption~\ref{assu:block-specific}), 
but we do not promote one over the other. Instead, we showcase how  CBT  can be implemented in conjunction with either partition to 
achieve good performance, acting as
a wrapper around thoughtfully populated competitions.

In the random hypothesis setting, we 
must slightly alter the compound p-value 
in~\eqref{eq:shared-cbt-common-offset-pval}. 
For $\widehat Y_g^{c,*}\coloneqq
\max\{\widehat Y_{g,k}:Y_{g,k}\leq c,\ k\in[n_g]\}$, 
with $\max\varnothing \coloneqq -\infty$ and $\rmax\coloneqq\max_{g\in[G]}b_g/n_g$, 
we define
\begin{equation}
    \widetilde p_j^{\blk, \rh}(\rmax)
    \coloneqq
    \biggl[
        \frac{
            \rmax
            + \sum_{g=1}^G \frac{b_g}{n_g}
                \ind\{\widehat Y_g^{c,*}\geq \widehat Y_{j,0}\}
            + m\ind\{\widehat Y_{j,0}<\widehat Y_{g(j)}^{c,*}\}
        }{m}
    \biggr]\wedge 1,
    \qquad j\in[m].
\end{equation}
These are compound p-values in a random sense: 
\begin{equation}
    \sum_{j=1}^m \PP\bigl( \tilde p_j^{\blk,\rh}(\rmax)\le t, Y_{j,0}\le c\bigr) \le mt,\qquad \forall t\in[0,1].
\end{equation} 
Furthermore, denoting $\hat R \coloneqq \hat R^\bh\bigl( \tilde p_1^{\blk,\rh}(\rmax), \dots,  \tilde p_m^{\blk,\rh}(\rmax)\bigr)$, we find that the analogue of the FDR result in 
Proposition~\ref{prop:shared-cbt-ville-fdr-control} and Corollary~\ref{coro:limsup-fdr-pi2alpha} still holds:
\begin{equation} 
    \label{eq:fdr-integral-bound-real-data}
    \fdr(\hat R)
    \coloneqq
    \EE\biggl[
        \frac{\sum_{j=1}^m
            \ind\{j\in\widehat R,\ Y_{j,0}\leq c\}}
        {|\widehat R|\vee1}
    \biggr]
    \leq
    \alpha\int_0^1
    \frac{-\log(1-z)}{\rmax\log(1+z/\rmax)}
    \,dz \xrightarrow[]{\rmax\to\infty} \frac{\pi^2}{6}\alpha;
\end{equation} 
see Appendix~\ref{sec:shared_cbt-rh}. 

As a baseline, we also implement the empirical p-values 
constructed per cluster by~\citet{jin2026multiscale}. 
These are a random hypothesis variant of the block empirical values~\eqref{eq:block-pval-defn} (with each cluster forming its own block of tests and controls), which are then
deflated using data-driven batch effect corrections.
We report the results of both the original UTOPIA procedure, 
which runs BH separately within each cluster and returns the union of all 
rejection sets, and a global variant that applies BH over the entire 
collection of p-values. Finally, we include the version~\eqref{eq:sss-compound-pval} of CBT that enforces  the assignment of exactly  one control to each test 
unit. As we have fewer control scores 
than test scores, we discard all hypotheses with no assigned 
controls prior to running BH, so we expect it to underperform 
in terms of power. 
Following~\citet{jin2026multiscale}, we evaluate the methods on a grid of 
$\alpha$ values from $0.005$ to $0.5$ in increments of $0.005$, 
returning the FDP, number of rejections, and number of true positives 
rejected at each $\alpha$. 

\begin{figure}[tbp!]
    \centering
    \includegraphics[width=\textwidth,clip]{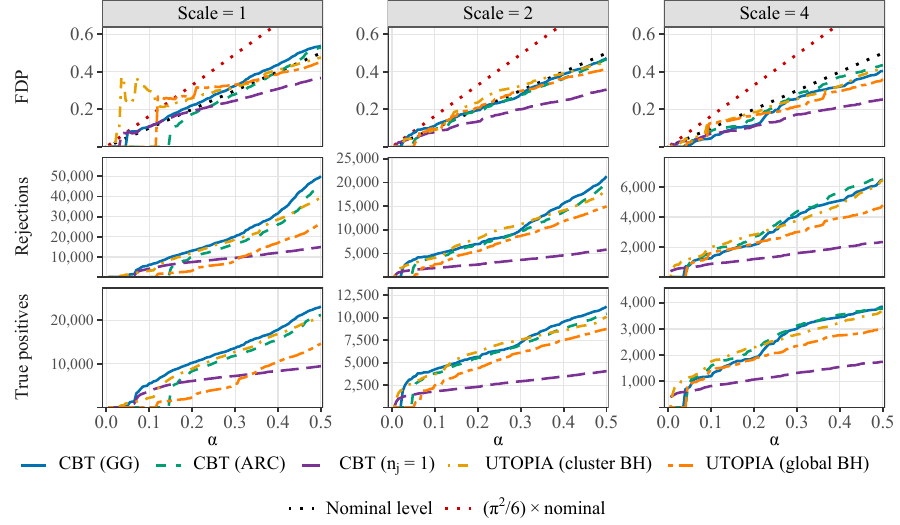}
    \caption{
    For the virtual ST data, we show the empirical performance of the CBT methods and baselines across three spatial aggregation scales (resolutions). In the top row, the black and red dotted lines signify control at $\alpha$ and $(\pi^2/6)\alpha$, respectively. We report the FDP, number of rejections, and number of true positives rejected.}
    \label{fig:utopia-tls-metagene}
\end{figure}

\paragraph{Results.}
Figure~\ref{fig:utopia-tls-metagene} shows the results of 
our data analysis for multiple spatial resolutions of the tissue sample.  \citet{jin2026multiscale} mention that all methods may perform 
quite poorly at fine resolutions (lower scale values);
this is showcased by the fact that 
both UTOPIA methods attain FDP significantly above $\alpha$ at scale 1. 
However, all CBT methods track the nominal level at all scales, 
with both shared-competitor variants rejecting more than the global 
UTOPIA baseline. Even compared to the cluster UTOPIA method 
(which, again, does not target global FDR control), 
the shared CBT methods reject a comparable number of hypotheses.  
Finally, as we have ground truth expression levels available, 
we demonstrate that a similar story holds for the number of true 
positive locations identified. Throughout each experiment, 
the $n_j=1$ version of CBT lags behind, as expected: requiring every hypotheses
to have its individual control group can prove to be a stringent 
condition, which is avoided by our shared-competitor paradigm.

In the FDP plots, we also highlight that the CBT methods have FDPs that are well below $(\pi^2/6)\alpha \approx 1.645\alpha$. 
At scales $1$, $2$, and $4$, respectively, the common offsets $\rmax$ are (approximately) $216.4$, $74.1$, and $31.7$ for ARC, giving FDR inflation factor bounds of approximately $1.647$, $1.652$, and $1.661$; the corresponding GG offsets are $19.5$, $29.7$, and $26.2$, giving bounds of approximately $1.670$, $1.662$, and $1.664$. The value of $\rmax$ also affects the effective resolution 
of the CBT p-values, so we see that ARC-based CBT, with its higher 
$\rmax$ values, can suffer in power at the lower FDR targets. Meanwhile, 
the GG-based CBT performs reasonably due to smaller $\rmax$.
For lower $\alpha$, GG-based CBT also outperforms global UTOPIA: as the 
plots demonstrate, it starts to achieve nonnegligible power starting 
at smaller $\alpha$ levels for scales 1 and 2.

\subsection{Application: CRISPR screen analysis with multiple cell lines}
\label{sec:application-crispr}

We next turn to a second application based on large-scale, pooled CRISPR screens
for determining whether knocking out a specific gene affects 
the viability of a certain cell.  The technology first infects cells of a certain type 
(called a cell line) with single-guide RNAs to disrupt the genes of interest and then measures the effects of the gene disruption on the cell population over time.  In particular,~\citet{dempster2025controlling} apply the 
Chronos module~\citep{dempster2021chronos} to assign 
gene-effect scores to each gene per cell line based on 
the outcome of the CRISPR screen.

Statistically, we test for 
each gene $j \in [m]$ and cell line $c \in [C]$ the null hypothesis:
\begin{equation}
    H_{j,c} \colon \text{knockout of gene }j \text{ has no effect on viability for cell line }c.
\end{equation}
We follow~\citet{dempster2025controlling}, 
who analyze five separate Avana screens~\citep{depmap202424q4}, 
each of which considers a different cell line.  Each cell line screen also provides access to nonessential genes, known not to affect viability when knocked out, and these are pooled to form a negative control group. Each gene is assigned a Chronos gene-effect score by the Chronos module, where higher values signify 
stronger evidence for presence of a viability effect.  These scores are then used to construct pooled empirical p-values per cell line screen; similar to UTOPIA for virtual ST, 
one then runs BH within each screen,\footnote{
    To be precise,~\citet{dempster2025controlling} actually run a variation of the BH procedure---known as the two-stage BH procedure~\citep{benjamini2006adaptive}---within each cell line. In our analysis, we use their p-values but run the original BH procedure per cell line as well as globally. By fixing the role of BH in all methods, we are able to visualize the advantages of our competition-based testing p-values.
} foregoing global FDR control over all $ H_{j,c}$.

\paragraph{Methods.} 
As in the virtual ST application, we use 
covariates to learn refined ARC and GG partitions of the data. 
With the resulting partitions
$(B_g)_{g\in [G]}$ of $[m]$ and $(C_g)_{g \in [m]}$ of the 
negative control gene-effect scores, we implement the 
CBT method with shared competitors 
(using~\eqref{eq:shared-cbt-common-offset-pval} 
as the hypotheses are deterministic) and apply BH. 
For the empirical p-values proposed by~\citet{dempster2025controlling}, 
we test two methods: the first runs BH within each cell line dataset and returns the union of rejection sets, and the second applies BH over the 
global collection of p-values.
The common offset values range from $4.0$ to
$40.5$ (including both ARC and GG), which corresponds to 
inflation factor bounds ranging from approximately $1.657$ to $1.766$. 

As stated by~\citet{dempster2025controlling}, there are some index pairs $(j,c)$ for which the ground truths for~$H_{j,c}$ are not known.  To assess FDR and power, we 
run two types of experiment: a global null setting and 
an unrestricted setting.
For the former, as any rejections would be 
false positives significant numbers of rejections would 
indicate a lack of FDR control.
In the latter setting, only some of the hypotheses' ground truths are known. 
We report the number of rejections as a general measure of discovery power, 
supplemented by the number of rejections whose ground truth is 
known to be nonnull (``Labeled true positives'')  as well as the 
proportion of rejections whose ground truth is known to be null 
(``Labeled FDP'').  
We evaluate the methods on a grid of $\alpha$ values from $0.005$ to $0.2$ 
in increments of $0.005$, returning the above metrics at each $\alpha$ level.

\begin{figure}[tbp!]
    \centering
    \includegraphics[width=\textwidth,clip]{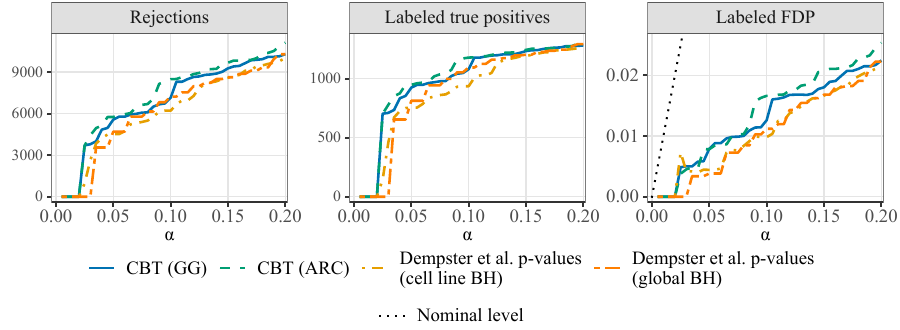}
    \caption{
        CRISPR application results when the hypotheses include both nulls and nonnulls.  We report the total number of rejections, the number of labeled true positives rejected, and the labeled FDP across nominal levels $\alpha$, signified in the right panel by the black dotted line.}
    \label{fig:crispr-all-genes}
\end{figure}

\paragraph{Results.}
In the global null experiment, all methods reported zero rejections at every value of $\alpha$, so we do not provide a plot.  Figure~\ref{fig:crispr-all-genes} summarizes the results of 
the unrestricted setting (where  all genes are included in the hypothesis set). 
We find that CBT with either ARC or GG partitions leads to more rejections 
than the cell line p-value approaches of~\citet{dempster2025controlling}, with an especially large power gap for small to moderate $\alpha$.  Further, we find that CBT exhibits a higher, but still comparable, 
labeled FDP. This empirical FDP is still much lower than the nominal level, 
but the labeled FDP is only a lower bound for the real 
FDP, as some unlabeled hypotheses rejected by the methods could be true nulls.

Unlike in the virtual ST example, taking the union of 
each cell line's rejection sets is not more powerful than 
running BH globally in this example. Due to the sparsity of nonnulls for certain cell lines, 
some hypotheses may not be rejected in the cell line-specific analysis. 
These hypotheses, in turn, can be rejected globally if the other cell lines can contribute many more p-values of similar magnitude.

%% file: sections/discussion.tex

We have proposed powerful multiple testing procedures based on compound p-values 
for a broad framework of competition-based testing where the 
competitions can vary in size and and share competitors
across hypotheses. Our theoretical results establish finite-sample FDR control guarantees (with inflation factor bounds based on the control group sizes $\bm n$)
for BH applied to these compound p-values. For the case where
every hypothesis has their own individual control group, 
we show tight upper and lower bounds on the worst-case inflation 
factor $\fdr/\alpha$. 
Moreover, we propose new competition structures where 
control units are shared by multiple test units (potentially in blocks) and provide 
compound p-values and corresponding inflated FDR control guarantees.
We then demonstrate the Type~I error and power 
of our methodology on simulated and real-data experiments, 
empirically verifying that our proposals control the FWER and FDR 
at the nominal level while achieving higher power than baseline 
methods.  

%% file: sections_appendix/appd_random_compound.tex

In this section, we generalize our competition-based testing methodology to the conformal selection framework
introduced by~\citet{jin2023selection}. 
We give preliminaries on the conformal selection setting, 
define the analogous random-compound p-value property, 
and prove extensions of our results. 

\subsection{Preliminaries}

We follow the treatment of~\citet{jin2023selection} and \citet{gui2025acs}, where each unit is a random tuple $(S, U, \cC)$ taking values in $\RR \times \RR \times \cP(\RR)$ and  $\cP(\RR)$ denotes the power set of $\RR$.  Conceptually, $U$ is a random outcome that we call ``interesting'' if $U \notin \cC$, so we call $\cC$ the ``dull set''.  For example, analysts may find values of $U$ outside some range interesting, and would seek to select such units.  In this setting,~$U$ can be unobserved, but we always observe $S$, a random variable that acts as a prediction or proxy for~$U$.  The straightforward approach is to select units based on $S \in \cC$. However, the uncertainty in $S$ as a prediction for $U$ is left unchecked, which necessitates   rigorous uncertainty quantification. 

Conformal selection~\citep{jin2023selection} formulates the selection problem as multiple testing with random hypotheses. In our language, in the $j$th competition, $S_{j,0}$ denotes our proxy for the unobserved $U_{j,0}$ and $\cC_{j, 0}$ is the dull set.  We can then define the random null hypothesis 
\begin{equation}
    \label{eq:random-hyp}
    H_{j} \colon U_{j,0} \in \cC_{j,0}.
\end{equation}
Conformal selection assumes access to fully-observed negative controls $(S_i, U_i, \cC_i)_{i=1}^N$ and further assumes that their concatenation with $( S_{j,0}, U_{j,0}, \cC_{j,0} )_{j=1}^m$ is a jointly exchangeable collection of $m+N$ tuples.  The method then computes a standard conformal p-value for $H_j$ using the scores $S_{j,0}$ and $(S_i)_{i:U_i \in \cC_i}$; thus, only the negative controls with dull $U_i$ act as reference units for $S_{j,0}$. This results in a computable quantity
\begin{equation}
    \label{eq:pooled-pval-rh}
    p_j^\pooled
    \coloneqq 
    \frac{
    1+\sum_{i=1}^N
    \ind\{S_i\ge S_{j,0}, U_i\in\cC_i\}
    }{
    N+1
    }.
\end{equation}
Despite the seemingly inflated  denominator, $p_j^\pooled$ is a strict p-value in the random-hypothesis sense:
\begin{equation}
    \label{eq:strict-pval-defn-rh}
    \PP(p_j^\pooled \le t, U_{j,0} \in \cC_{j,0}) \le t
\end{equation}
for each $t \in [0,1]$.  Often, $\cC_i$ is treated as deterministic: \citet{jin2023selection} originally formulated this problem by taking $\cC_i$ to be the deterministic one-sided intervals $(-\infty, c_i]$ for some constant $c_i$; hence,~$U_i$ is interesting if $U_i > c_i$. For ease of presentation, we will also consider such dull sets; extending our methods to general dull sets is straightforward.

To conclude the preliminaries, 
we extend the notions of FWER, FDR, and compound validity 
to the conformal selection setting. For $c_{1,0},\ldots,c_{m,0} \in \RR$ and a rejection set $\hat R$:
\begin{align}
    \label{eq:random-fwer}
    \fwer(\hat R)
    &\coloneqq
    \PP\bigg(
    \sum_{j=1}^m\ind\{j\in\hat R, U_{j,0}\le c_{j,0}\}\ge1
    \bigg);\\
    \label{eq:random-fdr}
    \fdr(\hat R)
    &\coloneqq
    \EE\bigg[
    \frac{\sum_{j=1}^m\ind\{j\in\hat R,  U _{j,0}\le c_{j,0}\}}
    {|\hat R|\vee1}
    \bigg].
\end{align}
Compared with the deterministic $H_j$ setting, we replace the sum over a fixed $\cH_0$ with a sum over all~$m$ hypotheses, 
but with terms randomly included based on $U_{j,0}$. 
Similarly, we extend the notion of strict p-value validity in~\eqref{eq:strict-pval-defn-rh} to a definition of compound p-values for random hypotheses.
\begin{definition}
    \label{def:random-compound-validity}
    We say $(\tilde p_1,\dots,\tilde p_m)$ are \emph{random-compound valid} if
    \begin{equation} 
        \sum_{j=1}^m
        \PP(\tilde p_j\le t, U_{j,0} \le c_{j,0})
        \le mt,
        \qquad t\in[0,1].
    \end{equation}
\end{definition}

\subsection{Extensions of Section~\ref{sec:cbt}}

We construct random-compound p-values for the varying-size setting based on $q$-finalists. To do so, we first rephrase the competition setting in the conformal selection language. First, for each $j\in[m]$, we have a test unit $(S_{j,0}, U_{j,0})$ as well as negative controls $(S_{j,k}, U_{j,k})_{k=1}^{n_j}$ and constants $(c_{j,k})_{k=0}^{n_j}$.  The analogue of Assumption~\ref{assu:individual} is the following:
\begin{assumption}[Individual exchangeability in the conformal selection setting]
    \label{assu:individual-rh}
    Conditional on
    \begin{equation}
        \cF_{-j}^{\rh}
        \coloneqq
        \sigma\bigl(
        \{(S_{\ell,k},\ind\{U_{\ell,k}\le c_{\ell,k}\}):
        \ell\in[m]\setminus\{j\},\ 0\le k\le n_\ell\}
        \bigr),
    \end{equation}
    the score and null status pairs
    \begin{equation}
        \bigl(S_{j,k},\ind\{U_{j,k}\le c_{j,k}\}\bigr)_{k=0}^{n_j}
    \end{equation}
    are exchangeable for each $j\in[m]$.
\end{assumption}

The next step is to define the $j$th competition. As in the conformal selection procedure, we filter out control units with the property $U_{j,k } > c_{j,k}$; the units that are left,
indexed by $\cN_j\coloneqq \{k \in [n_j]\colon U_{j,k} \le c_{j,k}\}$, are well suited to construct an empirical null distribution for $S_{j,0}$. The $q_j$-finalist condition, therefore, should be based on the score comparisons only among these null units.  For $a\in[|\cN_j|]$, let $S_{j,(a)}^{\cN}$ denote the $a$th largest score among $\{S_{j,k}:k\in\cN_j\}$, and define the finalist cutoff
\begin{equation}
    \tau_j^{(q_j)}
    \coloneqq
    \begin{cases}
        S_{j,(q_j)}^{\cN}, & \text{for } q_j\le |\cN_j|,\\
        -\infty, & \text{for }  q_j>|\cN_j|.
    \end{cases}
\end{equation}
We then define the $q_j$-finalist indicators as follows:
\begin{equation}
    V_{j,k}^{(q_j)}
    \coloneqq
    \begin{cases} 
        \ind\{S_{j,0}\ge\tau_j^{(q_j)}\},
        & \text{for } k=0,\\[0.5em]
        \displaystyle
        \ind\{k\in\cN_j,\ S_{j,k}\ge\tau_j^{(q_j)}\},
        & \text{for } k\in[n_j].
    \end{cases}
\end{equation}
These indicators are computing something similar, but not quite the same, as the $W^{(q_j)}_{j,k}$ used in the fixed-hypothesis setting.  The quantity $V_{j,0}^{(q_j)}$ computes whether $S_{j,0}$ is at least the $q_j$th largest null control score---not whether it placed in the top $q_j$ in a competition {\em that included itself}.  Similarly, when $k \in [n_j]$ (corresponding to control units), $V_{j,k}^{(q_j)}$ indicates if $(j,k)$ indexes a null control {\em and} its score placed in the top $q_j$ of the null controls. 

Restricting our attention to the $q_j =1$ case, we observe that whenever $|\cN_j | \ge 1$,   $V_{j,0}^{(1)}$ {\em and} the highest-scoring null control $V_{j,k}^{(1)}$ can both be equal to 1. This was previously impossible in the fixed hypothesis setting.

We are now in a position to present the random-compound p-values.
\begin{proposition}
    \label{prop:qfinalist-cbt-compound-rh}
    For each $j\in[m]$, fix $q_j\in[n_j]$, and write $\bm q=(q_1,\dots,q_m)$. Suppose that Assumption~\ref{assu:individual-rh} holds and there are no ties among the scores within each competition. For each $j\in[m]$, define
    \begin{equation}
        \label{eq:qfinalist-cbt-pval-rh}
        \tilde p_j^{(\bm q),\rh}
        \coloneqq
        \Biggl[
        \frac{
        \frac{1}{n_j}
        +\sum_{\ell=1}^m\frac{1}{n_\ell}
        \sum_{k=1}^{n_\ell}
        \ind\{V_{\ell,k}^{(q_\ell)}=1,\ S_{\ell,k}\ge S_{j,0}\}
        +m\ind\{V_{j,0}^{(q_j)}=0\}
        }{m}
        \Biggr]\wedge1.
    \end{equation}
    Then $\tilde p_1^{(\bm q),\rh},\dots,\tilde p_m^{(\bm q),\rh}$ are random-compound valid.
\end{proposition}

\subsection{Extensions of Section~\ref{sec:shared_cbt}}
\label{sec:shared_cbt-rh}

In addition to the setting where controls are assigned to individual hypotheses, we can find similar extensions to our compound p-values for block-specific competitions in the conformal selection setting. As in Section~\ref{sec:shared_cbt}, we have blocks of test and shared control unit tuples (rather than simply containing scores). Furthermore, we have the analogue of Assumption~\ref{assu:block-specific}:
\begin{assumption}[Block-specific exchangeability in the conformal selection setting]
    \label{assu:block-specific-rh}
    Conditional on the block partition and on
    \begin{equation}
        \begin{aligned}
            \cF_{-g}^{\blk,\rh}
            \coloneqq
            \sigma\Bigl(
            &\{(S_{j,0},\ind\{U_{j,0}\le c_{j,0}\}):j\in[m]\setminus B_g\}\\
            &\cup\{(S_{h,k},\ind\{U_{h,k}\le c_{h,k}\}):
            h\in[G]\setminus\{g\},\ k\in[n_h]\}
            \Bigr),
        \end{aligned}
    \end{equation}
    the score and null status pairs
    \begin{equation}
        \bigl((S_{j,0},\ind\{U_{j,0}\le c_{j,0}\}):j\in B_g\bigr)
        \cup
        \bigl((S_{g,k},\ind\{U_{g,k}\le c_{g,k}\}):k\in[n_g]\bigr)
    \end{equation}
    are exchangeable for each $g\in[G]$.
\end{assumption}

We  now  adjust our methodology to handle random hypotheses. First, fix $g\in[G]$ and  define the indices and max score of the null controls in group $g$ by
\begin{equation}
    \cN_g
    \coloneqq
    \{k\in[n_g]:U_{g,k}\le c_{g,k}\},\qquad
    S_g^{\cN,*}
    \coloneqq
    \max\{S_{g,k}:k\in\cN_g\},
\end{equation}
where $S_g^{\cN,*} \coloneqq -\infty$ if $\cN_g=\varnothing$. Analogously to the block-specific winner indicators in~\eqref{eq:block-winner}, for each $j\in B_g$, define
\begin{equation}
    \label{eq:block-winner-rh}
    V_{j,k}^{\blk}
    \coloneqq
    \begin{cases}
        \ind\{S_{j,0}=\max(S_{j,0},S_g^{\cN,*})\},
        & \text{for } k=0,\\
        \ind\{k\in\cN_g,\ S_{g,k}=S_g^{\cN,*}\},
        & \text{for } k\in[n_g].
    \end{cases}
\end{equation}
As in the previous subsection, these winner indicators are slightly different from their fixed-hypothesis predecessors. 

\begin{proposition}
    \label{prop:block-cbt-compound-rh}
    Suppose that Assumption~\ref{assu:block-specific-rh} holds and there are no ties among the test and control scores within each block. For each $j\in[m]$, define
    \begin{equation}
        \label{eq:block-cbt-pval-rh}
        \tilde p_j^{\blk,\rh}
        \coloneqq
        \Biggl[
        \frac{
        \frac{b_{g(j)}}{n_{g(j)}}
        +\sum_{h=1}^G\sum_{\ell\in B_h}\frac{1}{n_h}
        \sum_{k=1}^{n_h}
        \ind\{V_{\ell,k}^{\blk}=1,\ S_{h,k}\ge S_{j,0}\}
        +m\ind\{V_{j,0}^{\blk}=0\}
        }{m}
        \Biggr]\wedge1.
    \end{equation}
    Then $\tilde p_1^{\blk,\rh},\dots,\tilde p_m^{\blk,\rh}$ are random-compound valid.
\end{proposition}

We conclude with the analogue of Proposition~\ref{prop:shared-cbt-ville-fdr-control}, which we reference in Section~\ref{sec:application-virtual-st} when discussing the empirical performance of the BH procedure on $(\tilde p_j^{\blk,\rh})_{j\in[m]}$. 

\begin{proposition}
    \label{prop:abstract-random-block-cbt-ville-fdr-control}
    Suppose that the conditions of Proposition~\ref{prop:block-cbt-compound-rh} hold.
    Define $\rmax \coloneqq \max_{g\in[G]} b_g/n_g$ and
    \begin{equation}
        \label{eq:abstract-random-block-cbt-common-offset-pval}
        \tilde p_j^{\blk,\rh}(\rmax)
        \coloneqq
        \Biggl[
        \frac{
        \rmax
        +\sum_{h=1}^G\sum_{\ell\in B_h}\frac{1}{n_h}
        \sum_{k=1}^{n_h}
        \ind\{V_{\ell,k}^{\blk}=1,\ S_{h,k}\ge S_{j,0}\}
        +m\ind\{V_{j,0}^{\blk}=0\}
        }{m}
        \Biggr]\wedge1.
    \end{equation}
    Then 
    $\tilde p_1^{\blk,\rh}(\rmax),\dots,\tilde p_m^{\blk,\rh}(\rmax)$
    are random-compound valid. Moreover, for every $\alpha\in(0,1)$,
    \begin{equation}  \fdr\Bigl(
            \hat R^\bh\bigl(
            \tilde p_1^{\blk,\rh}(\rmax),\dots,
            \tilde p_m^{\blk,\rh}(\rmax)
            \bigr)
            \Bigr) \le
            \alpha\EE\biggl[
            \int_0^1
            \frac{\log\bigl(1/(1-z)\bigr)}{\rmax\log(1+z/\rmax)}
            \,dz
            \biggr]. 
    \end{equation}
    Moreover, 
    \[
        \limsup_{\rmax\to\infty}
            \fdr\Bigl(
                \hat R^\bh\bigl(
                \tilde p_1^{\blk,\rh}(\rmax),\dots,
                \tilde p_m^{\blk,\rh}(\rmax)
                    \bigr)
            \Bigr)
            \le \frac{\pi^2}{6}\alpha.
    \]
\end{proposition}
%
%
%

%% file: sections_appendix/appd_technical_proofs.tex

%
%

\section{Proofs of compound validity}
\label{sec:appd-proofs-compoundness}

Proposition~\ref{prop:qfinalist-cbt-compound} establishes Theorem~\ref{thm:cbt-compound} and Proposition~\ref{prop:sss-compound-bh} (the compound validity portion) as special cases.  The former follows when $q_j = 1$ for all $j\in[m]$; the latter when additionally $n_j = n$ for all $ j\in [m]$.

\subsection{Proof of Proposition~\ref{prop:qfinalist-cbt-compound}}
\label{sec:appd-qfinalist-compoundness-proof}

Since $c_j \mapsto \tilde p_j^{(\bm q)}(c_j)$ is increasing, it suffices to prove the result for $c_j=1/n_j$ for every $j\in[m]$. We write $\tilde p_j^{(\bm q)}$ for $\tilde p_j^{(\bm q)}(1/n_j)$ throughout the proof.

    Fix $t\in(0,1)$. For $j\in[m]$ and $0\le k\le n_j$, write
    \[
        \rho_{j,k}
        \coloneqq
        \sum_{a=0}^{n_j}\ind\{S_{j,a}\ge S_{j,k}\}
    \]
    for the rank of $S_{j,k}$ within its competition, so that
    $W_{j,k}^{(q_j)}=\ind\{\rho_{j,k}\le q_j\}$. 
    Also define the weighted number of control finalists with score at least $s \in \RR$ and its leave-$j$th-competition-out version by
    \begin{align}
        C_{\bm q}(s)
        &\coloneqq
        \sum_{\ell=1}^m\frac{1}{n_\ell}\sum_{k=1}^{n_\ell}
        \ind\{W_{\ell,k}^{(q_\ell)}=1,\ S_{\ell,k}\ge s\},
        \\
        C_{\bm q}^{(-j)}(s)
        &\coloneqq
        \sum_{\ell\ne j}\frac{1}{n_\ell}\sum_{k=1}^{n_\ell}
        \ind\{W_{\ell,k}^{(q_\ell)}=1,\ S_{\ell,k}\ge s\}.
    \end{align}
If $W_{j,0}^{(q_j)}=0$, then $\tilde p_j^{(\bm q)}=1$. On the other hand, if $W_{j,0}^{(q_j)}=1$, then precisely $\rho_{j,0}-1$ controls in competition $j$ have scores greater than $S_{j,0}$, and each of them is also a $q_j$-finalist. Consequently,
    \[
        C_{\bm q}(S_{j,0})
        =
        C_{\bm q}^{(-j)}(S_{j,0})
        +\frac{\rho_{j,0}-1}{n_j}.
    \]
It follows that
    \begin{equation}
        \label{eq:qfinalist-pvalue-event}
        \{\tilde p_j^{(\bm q)}\le t\}
        =
        \biggl\{
        W_{j,0}^{(q_j)}=1,\ 
        \frac{\rho_{j,0}}{n_j}
        +C_{\bm q}^{(-j)}(S_{j,0})
        \le mt
        \biggr\}.
    \end{equation}

    Now fix $j\in\cH_0$. 
    Conditional on $\cG_{-j}$ and under Assumption~\ref{assu:individual}, 
    the joint distribution of $(S_{j,k})_{k=0}^{n_j}$ is invariant under any permutation of $\{0, \dots, n_j\}$.
    Moreover, the function $C_{\bm q}^{(-j)}(\cdot)$ does not depend on $(S_{j,k})_{k=0}^{n_j}$. Hence, from~\eqref{eq:qfinalist-pvalue-event}, 
    \begin{align}
    \begin{split}
    \label{eq:qfinalist-label-balance}
\PP(\tilde p_j^{(\bm q)}\le t ) = \mathbb{E}\bigl[\PP(\tilde p_j^{(\bm q)}\le t \given \mathcal{G}_{-j})\bigr] &= \EE\biggl[ \PP\biggl(
        W_{j,0}^{(q_j)}=1,\ 
        \frac{\rho_{j,0}}{n_j}
        +C_{\bm q}^{(-j)}(S_{j,0})
        \le mt
         \bigggiven \mathcal{G}_{-j}\biggr)\biggr]\\
&= \EE\biggl[
            \frac{1}{n_j}\sum_{k=1}^{n_j}
            \ind\Bigl\{
            W_{j,k}^{(q_j)}=1,\ 
            \frac{\rho_{j,k}}{n_j}
            +C_{\bm q}^{(-j)}(S_{j,k})
            \le mt
            \Bigr\}
            \biggr],
            \end{split}
    \end{align}
where the final equality holds because replacing the zero subscripts with $k \in [n_j]$ leaves the conditional probability unchanged.
    Let
    \begin{equation}
        \label{eq:A_t-def-indiv-fixed}
        \cA_t
        \coloneqq
        \biggl\{
        (j,k) \in [m] \times [n_j]: W_{j,k}^{(q_j)}=1,\ 
        \frac{\rho_{j,k}}{n_j}
        +C_{\bm q}^{(-j)}(S_{j,k})
        \le mt
        \biggr\}. 
    \end{equation}
    If $\cA_t$ is non-empty, then choose $(j^*,k^*)\in\cA_t$ with $S_{j^*,k^*} = \min\{S_{j,k}:(j,k) \in \cA_t\}$, and write $s^*\coloneqq S_{j^*,k^*}$ and $r^*\coloneqq\rho_{j^*,k^*}$.  Every index pair of a control in $\cA_t$ outside competition $j^*$ has corresponding score at least $s^*$ and is therefore counted by $C_{\bm q}^{(-j^*)}(s^*)$.  Moreover, for any $k \in [n_{j^*}]$ such that $(j^*,k) \in \cA_t$, we have
    \[
    \sum_{a=0}^{n_{j^*}} \ind\{S_{j^*,a} \geq S_{j^*,k}\} \leq \sum_{a=0}^{n_{j^*}} \ind\{S_{j^*,a} \geq s^*\} \leq r^*.
    \]
    Hence, by~\eqref{eq:qfinalist-label-balance},
    \begin{align}
        \label{eq:qfinalist-pathwise-bound}
        \begin{split}
            \sum_{j\in\cH_0}\PP(\tilde p_j^{(\bm q)}\le t) &= \EE\biggl[\sum_{j\in\cH_0}\frac{1}{n_j}\sum_{k=1}^{n_j}
            \ind\Bigl\{
            W_{j,k}^{(q_j)}=1,\ 
            \frac{\rho_{j,k}}{n_j}
            +C_{\bm q}^{(-j)}(S_{j,k})
            \le mt
            \Bigr\}\biggr] \\
            &\le
            \EE\biggl[
            \sum_{j=1}^m\frac{1}{n_j}\sum_{k=1}^{n_j}
            \ind\{(j,k)\in\cA_t\}\biggr]  \le
            \EE\biggl[\frac{r^*}{n_{j^*}}+C_{\bm q}^{(-j^*)}(s^*)\biggr] 
            \le mt.
        \end{split}
    \end{align}
    The cases $t=0$ and $t=1$ are immediate, so the result follows.
    \hfill $\qed$

\begin{remark}[Assumptions in Theorem~\ref{thm:cbt-compound} and Proposition~\ref{prop:qfinalist-cbt-compound}]
    \label{rem:looser-assumptions}
    Assumption~\ref{assu:individual} is a convenient sufficient condition, 
    but Proposition~\ref{prop:qfinalist-cbt-compound} remains true under a weaker condition on the null competitions.  Specifically, for each $j\in\cH_0$, define the $\sigma$-algebra 
    \begin{equation}
        \label{eq:individual-invariant-sigma-algebra}
        \cI_j \coloneqq \sigma\bigl(\Lbag(S_{j,k})_{k=0}^{n_j}\Rbag,\,\cG_{-j}\bigr),
    \end{equation}
    where $\Lbag A\Rbag$ denotes the unordered version of a collection $A$ (a multiset). 
    To show the last line in~\eqref{eq:qfinalist-label-balance}, it is sufficient that, for every $j\in\cH_0$,
    \begin{equation}
        \label{eq:qfinalist-conditional-rank-uniformity}
        \PP(\rho_{j,0}=r\given\cI_j)
        =
        \frac{1}{n_j+1},
        \qquad r\in[n_j+1].
    \end{equation}
    This condition is implied by Assumption~\ref{assu:individual}. 
    Indeed, we then have
    \begin{align*}
            \EE\biggl[
            \ind\biggl\{
            W_{j,0}^{(q_j)}=1,\;
            \frac{\rho_{j,0}}{n_j}
            +C_{\bm q}^{(-j)}(S_{j,0})\le mt
            \biggr\}
            \bigggiven \cI_j
            \biggr] &=
            \frac{1}{n_j+1}\sum_{k=0}^{n_j}
            \ind\biggl\{
            W_{j,k}^{(q_j)}=1,\;
            \frac{\rho_{j,k}}{n_j}
            +C_{\bm q}^{(-j)}(S_{j,k})\le mt
            \biggr\}\\
            &\hspace{-0.5cm}=
            \frac{1}{n_j}
            \EE\biggl[
            \sum_{k=1}^{n_j}
            \ind\biggl\{
            W_{j,k}^{(q_j)}=1,\;
            \frac{\rho_{j,k}}{n_j}
            +C_{\bm q}^{(-j)}(S_{j,k})\le mt
            \biggr\}
            \bigggiven \cI_j
            \biggr].
    \end{align*}
For the $q_j=1$ case covered by Theorem~\ref{thm:cbt-compound}, even~\eqref{eq:qfinalist-conditional-rank-uniformity} is stronger than necessary. 
Since $S_{j,(0)} \coloneqq \max_{k=0, \dots, n_j} S_{j,k}$
is $\cI_j$-measurable, it is enough to assume that
    \begin{equation}
        \label{eq:cbt-conditional-win-probability}
        \PP(W_{j,0}=1\given\cI_j)
        =
        \frac{1}{n_j+1},
        \qquad j\in\cH_0.
    \end{equation}
Indeed, by our assumption that each competition has a unique winner, we then have
    \begin{equation}
        \label{eq:cbt-conditional-winner-balance}
        \EE[W_{j,0}\given\cI_j]
        =\frac{1}{n_j+1}
        =\frac{1}{n_j}\EE[{1-W_{j,0}}\given\cI_j]=\frac{1}{n_j}\EE\biggl[{\sum_{k=1}^{n_j}W_{j,k}}\given\cI_j\biggr],
        \qquad j\in\cH_0.
    \end{equation}
    Hence
    \begin{align*}
            \EE\biggl[
            \ind\biggl\{
            W_{j,0} =1,\;
            \frac{1}{n_j}
            +C_{\bm q} ^{(-j)}(S_{j,0})\le mt
            \biggr\}
            \bigggiven \cI_j
            \biggr] &= 
            \EE\biggl[W_{j,0} \cdot
            \ind\biggl\{
            \frac{1}{n_j}
            +C_{\bm q} ^{(-j)}(S_{j,(0)})\le mt
            \biggr\}
            \bigggiven \cI_j
            \biggr]\\
            &= 
            \ind\biggl\{\frac{1}{n_j}
            +C_{\bm q}^{(-j)}(S_{j,(0)})\le mt\biggr\}
            \EE [
            W_{j,0}
            \given \cI_j
            ]\\
            &= 
            \ind\biggl\{\frac{1}{n_j}
            +C_{\bm q}^{(-j)}(S_{j,(0)})\le mt\biggr\}\cdot \frac{1}{n_j}\EE\biggl[{\sum_{k=1}^{n_j}W_{j,k}}\given\cI_j\biggr]\\
            &= 
            \EE\biggl[
                \frac{1}{n_j}
                \sum_{k=1}^{n_j}
                \ind\biggl\{W_{j,k}=1,  \frac{1}{n_j}
            +C_{\bm q}^{(-j)}(S_{j,(0)})\le mt \biggr\}
            \bigggiven \cI_j
            \biggr]\\
            &=  \EE\biggl[
                \frac{1}{n_j}
                \sum_{k=1}^{n_j}
                \ind\biggl\{W_{j,k}=1,  \frac{1}{n_j}
            +C_{\bm q}^{(-j)}(S_{j,k})\le mt \biggr\}
            \bigggiven \cI_j
            \biggr],
    \end{align*}
    which establishes the last line in~\eqref{eq:qfinalist-label-balance} for this case.
\end{remark}
\subsection{Proof of Theorem~\ref{thm:shared-cbt-compound}}
\label{sec:appd-proof-shared-cbt-compound}

    Fix $t\in(0,1)$.  For $s \in \RR$, define the weighted number of control winners that are at least $s$, the leave-one-block-out version excluding block $g\in [G]$, and
    the number of test winners that are at least $s$ as
    \begin{align}
        C^{\blk}(s)
        &\coloneqq
        \sum_{h=1}^G\sum_{\ell\in B_h}\frac{1}{n_h}\sum_{k=1}^{n_h}
        \ind\{W_{\ell,k}^{\blk}=1,\ S_{h,k}\ge s\}, 
        \\
        C^{\blk,(-g)}(s)
        &\coloneqq
        \sum_{h\ne g}\sum_{\ell\in B_h}\frac{1}{n_h}\sum_{k=1}^{n_h}
        \ind\{W_{\ell,k}^{\blk}=1,\ S_{h,k}\ge s\}, 
        \\
        R^{\blk}(s)
        &\coloneqq
        \sum_{\ell=1}^m
        \ind\{W_{\ell,0}^{\blk}=1,\ S_{\ell,0}\ge s\}, 
    \end{align}

    Fix $g\in [G]$ and $j\in B_g \cap \cH_0$. 
    If $W_{j,0}^{\blk}=0$, then $\tilde p_j^{\blk}=1$. On the other hand, if $W_{j,0}^{\blk}=1$, then by our no ties assumption, $S_{j,0}$ is larger each of $S_{g,1}, \dots, S_{g,n_g}$. Consequently, the contribution to $C^{\blk}(S_{j,0})$ from $B_g$ is zero, and hence $ C^{\blk}(S_{j,0}) =  C^{\blk,(-g)}(S_{j,0})$
    when $S_{j,0}$ is a winner.
    It follows that
    \begin{equation}
        \label{eq:shared-cbt-pvalue-event}
        \{\tilde p_j^{\blk}\le t\}
        =
        \biggl\{W_{j,0}^{\blk}=1,\ \frac{b_g}{n_g}+C^{\blk,(-g)}(S_{j,0})\le mt\biggr\},
        \qquad j\in B_g.
    \end{equation}

    The function ${b_g}/{n_g}+C^{\blk,(-g)}(\cdot)$ depends only on scores outside block $g$. Conditional on $\cG_{-g}^\blk$ and under Assumption~\ref{assu:block-specific}, 
    the joint distribution of the scores $(S_{j,0}, S_{g,1}, \dots, S_{g,n_g})$ is again invariant under any  permutation  of $\{0, \dots, n_g \}$. Furthermore, 
    $C^{\blk,(-g)}(\cdot)$ does not depend on 
    any test or control scores from block $g$.
    Therefore, by~\eqref{eq:shared-cbt-pvalue-event} and arguing as in Remark~\ref{rem:looser-assumptions},
    \begin{align}
        \begin{split}
        \label{eq:shared-cbt-label-balance}
    \PP(\tilde p_j^{\blk}\le t ) = \mathbb{E}\bigl[\PP(\tilde p_j^{\blk}\le t \given \cG^{\blk}_{-g})\bigr] &= \EE\biggl[ \PP\biggl(W_{j,0}^\blk=1, \frac{b_g}{n_g}+C^{\blk,(-g)}(S_{j,0})\le mt
             \bigggiven \cG^{\blk}_{-g}\biggr)\biggr]\\
            &= \EE\biggl[
            \frac{1}{n_g}\sum_{k=1}^{n_g}
            \ind\Bigl\{
            W_{j,k}^\blk=1, \frac{b_g}{n_g}+C^{\blk,(-g)}(S_{g,k})\le mt
            \Bigr\}
            \biggr].
        \end{split}
    \end{align}

    Let
    \begin{equation}
        \label{eq:A_t-def-block-fixed}
        \cA_t
        \coloneqq
        \bigcup_{g=1}^G
        \biggl\{
        (j,k)\in B_g\times[n_g]:
        W_{j,k}^{\blk}=1,\ 
        \frac{b_g}{n_g}
        +C^{\blk,(-g)}(S_{g,k})
        \le mt
        \biggr\}. 
    \end{equation} 
    If $\cA_t$ is non-empty, then choose
    $(j^*,k^*)\in\cA_t$ with
    $S_{g(j^*),k^*}  = \min\{S_{g(j),k}:(j,k)\in\cA_t\}$,
    and write $g^*\coloneqq g(j^*)$ and
    $s^*\coloneqq S_{g^*,k^*}$.
    Every index pair $(j,k) \in \cA_t$ for which  $j \notin B_{g^*}$ has $S_{g(j),k} \geq s^*$ and is therefore counted by
    $C^{\blk,(-g^*)}(s^*)$.
    Moreover, each competition indexed by $B_{g^*}$ has at most
    one control winner by assumption, so
    \[
        \sum_{j\in B_{g^*}}\sum_{k=1}^{n_{g^*}}
        \ind\{(j,k)\in\cA_t\}
        \le
        \sum_{j\in B_{g^*}}\sum_{k=1}^{n_{g^*}}
        W_{j,k}^{\blk}
        \le b_{g^*}.
    \]
    Hence, by~\eqref{eq:shared-cbt-label-balance},
    \begin{equation}
    \begin{aligned}
        \sum_{j\in\cH_0}\PP(\tilde p_j^{\blk}\le t)
        &=
        \EE\biggl[
        \sum_{g=1}^G\sum_{j\in B_g\cap\cH_0}
        \frac{1}{n_g}\sum_{k=1}^{n_g}
        \ind\Bigl\{
        W_{j,k}^{\blk}=1,\ 
        \frac{b_g}{n_g}
        +C^{\blk,(-g)}(S_{g,k})\le mt
        \Bigr\}
        \biggr] \\
        &\le
        \EE\biggl[
        \sum_{g=1}^G\sum_{j\in B_g}
        \frac{1}{n_g}\sum_{k=1}^{n_g}
        \ind\{(j,k)\in\cA_t\}
        \biggr] \le
        \EE\biggl[ 
            \frac{b_{g^*}}{n_{g^*}}
            +C^{\blk,(-g^*)}(s^*) \biggr] 
        \le mt,
    \end{aligned} 
    \end{equation}
    where the last inequality follows from the fact that
    $(j^*,k^*)\in\cA_t$ whenever $\cA_t$ is non-empty.
    The cases $t=0$ and $t=1$ are immediate, so the result follows.
    \hfill $\qed$

\subsection{Proofs of random-compound validity results}

\subsubsection{Proof of Proposition~\ref{prop:qfinalist-cbt-compound-rh}}
\label{sec:appd-proof-qfinalist-cbt-compound-rh}

For $j \in [m]$ and $k \in [n_j]$, write $I_{j,k}\coloneqq\ind\{U_{j,k}\le c_{j,k}\} $
    for the null status indicator of each unit
    $(j,k)$. We call
    $(S_{j,k}, I_{j,k})$
    a  ``score--null-status'' pair.

    Fix $t\in(0,1)$.  Define the rank of $S_{j,k}$
    relative to null controls as 
    \begin{equation}
        \rho_{j,k}^{\rh}
        \coloneqq
        1+\sum_{a=1}^{n_j}
        I_{j,a}\ind\{S_{j,a}>S_{j,k}\}.
    \end{equation}
    Then the test and control 
    $q_j$-finalist indicators can be rewritten as $V_{j,0}^{(q_j)}  =\ind\{\rho_{j,0}^{\rh}\le q_j\}$
    and $V_{j,k}^{(q_j)}  =I_{j,k}\ind\{\rho_{j,k}^{\rh}\le q_j\}$,
    respectively. Define $C_{\bm q}^{\rh}$ and $C_{\bm q}^{\rh,(-j)}$
    as the corresponding counts in the proof of
    Proposition~\ref{prop:qfinalist-cbt-compound}, with each
    $W_{\ell,k}^{(q_\ell)}$ replaced by $V_{\ell,k}^{(q_\ell)}$.
    The calculation leading to~\eqref{eq:qfinalist-pvalue-event},
    now counting only the controls with null status $I_{j,k} = 1$, gives
    \begin{equation}
        \label{eq:qfinalist-rh-pvalue-event}
        \{\tilde p_j^{(\bm q),\rh}\le t,\ I_{j,0}=1\}
        =
        \biggl\{
        I_{j,0}=1,\ 
        \rho_{j,0}^{\rh}\le q_j,\ 
        \frac{\rho_{j,0}^{\rh}}{n_j}
        +C_{\bm q}^{\rh,(-j)}(S_{j,0})\le mt
        \biggr\}.
    \end{equation}

    Now fix $j \in [m]$. 
    For each $k\in[n_j]$, 
    write $\rho_{j,0}^{\rh;(j,k)}$ for the rank of the 
    test unit when  $(S_{j,0},I_{j,0})$ and
    $(S_{j,k},I_{j,k})$ are swapped (so that the latter is the new test unit for competition $j$).
    The original test becomes a control post-swap, so
    \begin{equation}
        \label{eq:qfinalist-rh-rank-comparison}
        \rho_{j,0}^{\rh;(j,k)}
        =
        \rho_{j,k}^{\rh}
        +I_{j,0}\ind\{S_{j,0}>S_{j,k}\}
        \ge \rho_{j,k}^{\rh}.
    \end{equation}
    Conditional on $\cF_{-j}^{\rh}$ and under
    Assumption~\ref{assu:individual-rh}, the joint
    distribution of $(S_{j,k}, I_{j,k})_{0\le k\le n_j}$ is invariant under any permutation of $\{0, \dots, n_j \}$, while
    $C_{\bm q}^{\rh,(-j)}(\cdot)$ remains fixed.
    Applying the same exchangeability argument as
    in~\eqref{eq:qfinalist-label-balance}, followed by
    \eqref{eq:qfinalist-rh-rank-comparison}, yields
    \begin{equation}
        \label{eq:qfinalist-rh-label-balance}
        \begin{aligned}
            \PP(\tilde p_j^{(\bm q),\rh}\le t,\ I_{j,0}=1)
            &=
            \EE\biggl[
            \frac{1}{n_j}\sum_{k=1}^{n_j}
            I_{j,k} \cdot 
            \ind\Bigl\{
            \rho_{j,0}^{\rh;(j,k)}\le q_j,\ 
            \frac{\rho_{j,0}^{\rh;(j,k)}}{n_j}
            +C_{\bm q}^{\rh,(-j)}(S_{j,k})\le mt
            \Bigr\}
            \biggr]\\
            &\le
            \EE\biggl[
            \frac{1}{n_j}\sum_{k=1}^{n_j}
            \ind\Bigl\{
            V_{j,k}^{(q_j)}=1,\ 
            \frac{\rho_{j,k}^{\rh}}{n_j}
            +C_{\bm q}^{\rh,(-j)}(S_{j,k})\le mt
            \Bigr\}
            \biggr].
        \end{aligned}
    \end{equation}

    Define $\cA_t^{\rh}$ as $\cA_t$ in \eqref{eq:A_t-def-indiv-fixed} (in the proof of
    Proposition~\ref{prop:qfinalist-cbt-compound}), but with
    $W_{j,k}^{(q_j)}$, $\rho_{j,k}$, and
    $C_{\bm q}^{(-j)}$ replaced by their
    random-hypothesis counterparts 
    $V_{j,k}^{(q_j)}$, $\rho_{j,k}^{\rh}$, and
    $C_{\bm q}^{\rh,(-j)}$, respectively.
    The same argument via the minimum 
    score of $\cA_t^\rh$ can then be repeated;
    denote this minimum score as $s^*$ and 
    its index pair as $(j^*, k^*)$.
    We can then write
    \begin{equation}
        \label{eq:qfinalist-rh-pathwise-bound}
        \sum_{j=1}^m
        \PP(\tilde p_j^{(\bm q), \rh}\le t,\ I_{j,0}=1)
        \le
        \EE\biggl[
        \sum_{j=1}^m\frac{1}{n_j}\sum_{k=1}^{n_j}
        \ind\{(j,k)\in\cA_t^{\rh}\}
        \biggr]
        \le
            \EE\biggl[\frac{\rho_{j^*,k^*}^{\rh}}{n_{j^*}}+C_{\bm q}^{\rh,(-j^*)}(s^*)\biggr] 
        \le mt.
    \end{equation}
    As for the fixed-hypothesis compound validity proofs, the cases $t=0$ and $t=1$ are immediate, so the result follows.
    \hfill $\qed$

\subsubsection{Proof of Proposition~\ref{prop:block-cbt-compound-rh}}
\label{sec:appd-proof-block-cbt-compound-rh}

    Fix $t\in(0,1)$ and $g \in [G]$. 
    Define $C^{\blk,\rh}$ and $C^{\blk,\rh,(-g)}$
    as the corresponding counts in the proof of
    Theorem~\ref{thm:shared-cbt-compound}, with each
    $W_{\ell,k}^{\blk}$ replaced by $V_{\ell,k}^{\blk}$. 
    Fix $j\in B_g$.
    If $V_{j,0}^{\blk}=0$, then $\tilde p_j^{\blk}=1$.
    On the other hand, if $V_{j,0}^{\blk}=1$, then
    $S_{j,0}$ is larger than every control score in block $g$
    with $I_{g,k}=1$.
    Consequently, the contribution to
    $C^{\blk,\rh}(S_{j,0})$ from $B_g$ is zero, and hence 
       $ C^{\blk,\rh}(S_{j,0})
        =
        C^{\blk,\rh,(-g)}(S_{j,0}) $
    when~$S_{j,0}$ is a winner. It follows that
    \begin{equation}
        \label{eq:block-rh-pvalue-event}
        \{\tilde p_j^{\blk}\le t,\ I_{j,0}=1\}
        =
        \biggl\{
        I_{j,0}=1,\ 
        V_{j,0}^{\blk}=1,\ 
        \frac{b_g}{n_g}
        +C^{\blk,\rh,(-g)}(S_{j,0})\le mt
        \biggr\}.
    \end{equation}

    For $k\in[n_g]$, write $(V_{j,0}^{\blk})^{(j,k)}$ for the winner indicator when $(S_{j,0},I_{j,0})$
    and $(S_{g,k},I_{g,k})$ are swapped (so that the latter is the new test unit for competition $j$).  The original test becomes a control, so 
    \begin{equation}
        \label{eq:block-rh-winner-comparison}
            I_{g,k}\cdot (V_{j,0}^{\blk})^{(j,k)}
            =
            V_{j,k}^{\blk}\bigl(1 - 
            I_{j,0}\ind\{S_{j,0}>S_{g,k}\}\bigr) \le V_{j,k}^{\blk}. 
    \end{equation}
    Conditional on $\cF_{-g}^{\blk,\rh}$, the joint
    distribution of 
    $\bigl(
    (S_{j,0},I_{j,0}),
    (S_{g,1},I_{g,1}),\dots,
    (S_{g,n_g},I_{g,n_g})
    \bigr)$ is 
    invariant under any permutation of $\{0, \dots, n_g \}$ by 
    Assumption~\ref{assu:block-specific-rh}, while
    $C^{\blk,\rh,(-g)}(\cdot)$ remains fixed.
    Applying the same exchangeability argument as
    in~\eqref{eq:shared-cbt-label-balance}, followed by
    \eqref{eq:block-rh-winner-comparison}, yields
    \begin{equation}
        \label{eq:block-rh-label-balance}
        \begin{aligned}
            \PP(\tilde p_j^{\blk}\le t,\ I_{j,0}=1)& =
            \EE\biggl[
            \frac{1}{n_g}\sum_{k=1}^{n_g}
            I_{g,k} \cdot 
            \ind\Bigl\{
            (V_{j,0}^{\blk})^{(j,k)}=1,\ 
            \frac{b_g}{n_g}
            +C^{\blk,\rh,(-g)}(S_{g,k})\le mt
            \Bigr\}
            \biggr]\\
            &\le
            \EE\biggl[
            \frac{1}{n_g}\sum_{k=1}^{n_g}
            \ind\Bigl\{
            V_{j,k}^{\blk}=1,\ 
            \frac{b_g}{n_g}
            +C^{\blk,\rh,(-g)}(S_{g,k})\le mt
            \Bigr\}
            \biggr].
        \end{aligned}
    \end{equation}

    Define $\cA_t^{\blk,\rh}$ as $\cA_t$ in~\eqref{eq:A_t-def-block-fixed} (in the proof of
    Theorem~\ref{thm:shared-cbt-compound}), with
    $W_{j,k}^{\blk}$ and $C^{\blk,(-g)}$ replaced by their
    random-hypothesis counterparts
    $V_{j,k}^{\blk}$ and $C^{\blk,\rh,(-g)}$, respectively. 
    Choose
    $(j^*,k^*)\in\cA_t^{\blk,\rh}$ with
    \[
        S_{g(j^*),k^*}
        =
        \min\bigl\{
        S_{g(j),k}:(j,k)\in\cA_t^{\blk,\rh}
        \bigr\},
    \]
    and write $g^*\coloneqq g(j^*)$ and
    $s^*\coloneqq S_{g^*,k^*}$.
    Every control winner in $\cA_t^{\blk,\rh}$ outside
    block $g^*$ has score at least~$s^*$ and is therefore
    counted by $C^{\blk,\rh,(-g^*)}(s^*)$.
    Moreover, each competition in block $g^*$ has at most
    one null control winner, so
    \[
        \sum_{j\in B_{g^*}}\sum_{k=1}^{n_{g^*}}
        \ind\{(j,k)\in\cA_t^{\blk,\rh}\}
        \le
        \sum_{j\in B_{g^*}}\sum_{k=1}^{n_{g^*}}
        V_{j,k}^{\blk}
        \le b_{g^*}.
    \]
    Hence, by~\eqref{eq:shared-cbt-label-balance},
    \begin{equation}
    \begin{aligned}
        \sum_{j=1}^m \PP(\tilde p_j^{\blk}\le t, I_{j,0}=1)
        &=
        \EE\biggl[
        \sum_{g=1}^G\sum_{j\in B_g\ }
        \frac{1}{n_g}\sum_{k=1}^{n_g}
            \ind\Bigl\{
            V_{j,k}^{\blk}=1,\ 
            \frac{b_g}{n_g}
            +C^{\blk,\rh,(-g)}(S_{g,k})\le mt
        \Bigr\}
        \biggr] \\
        &\le
        \EE\biggl[
        \sum_{g=1}^G\sum_{j\in B_g}
        \frac{1}{n_g}\sum_{k=1}^{n_g}
        \ind\{(j,k)\in\cA_t\}
        \biggr] \le
        \EE\biggl[ 
            \frac{b_{g^*}}{n_{g^*}}
            +C^{\blk,\rh,(-g^*)}(s^*) \biggr] 
        \le mt.
    \end{aligned} 
    \end{equation}
    where the last inequality follows from
    $(j^*,k^*)\in\cA_t^{\blk,\rh}$.
    The 
    cases $t=0$ and $t=1$ are immediate, so 
    the result follows.
    \hfill $\qed$

\section{Proofs for results on BH and FDR control}
\label{sec:appd-proofs-fdr}

\subsection{Proof of Proposition~\ref{prop:cbt-bh-equivalence}}
\label{sec:appd-cbt-equivalence-proof}

    Fix $c_0\ge0$ and define the weighted number of control winners and the number of test winners with scores at least $t \in \RR$ by
    \begin{align}
        C(t)
        &\coloneqq
        \sum_{\ell=1}^m\frac{1}{n_\ell}\sum_{k=1}^{n_\ell}
        \ind\{W_{\ell,k}=1,\ S_{\ell,k}\ge t\},
        \\
        R(t)
        &\coloneqq
        \sum_{\ell=1}^m
        \ind\{W_{\ell,0}=1,\ S_{\ell,0}\ge t\}.
    \end{align}
    Thus, the threshold in~\eqref{eq:cbt-rejection-set} can be written as
    \begin{equation}
        \label{eq:rewrite-T}
        T
        =
        \min\biggl\{
        t\in\cS:
        \frac{c_0+C(t)}{1\vee R(t)}\le\alpha
        \biggr\}.
    \end{equation} 

    Define $\hat k \coloneqq  | \hat R^\bh |$.
    On the event $\{\hat k > 0\}$,
    write $j_* \coloneqq \mathrm{sargmin}\{S_{j,0}:j \in \hat{R}^{\bh}\}$, and define $j_* \coloneqq 1$ otherwise. Since $t\mapsto C(t)$ is decreasing, we have for every $j\in \hat R^\bh\bigl(\tilde p_1(c_0),\dots,\tilde p_m(c_0)\bigr)$ that
    \begin{equation}
        \label{eq:compound-p-bh-threshold}
        \tilde p_j(c_0) = \frac{c_0  + C(S_{j,0})}{m} \leq \frac{c_0  + C(S_{j_*,0})}{m} = \tilde p_{j_*}(c_0) \le \frac{\alpha \hat k}{m}.
    \end{equation} 
    Since $\alpha \hat k /m  < 1$, we have $W_{j,0} = 1$ for $j \in \hat{R}^{\mathrm{BH}}$. Hence 
    \begin{equation}
        \label{eq:bh-rej-set-subseteq}
        \hat{R}^\bh\bigl(\tilde p_1(c_0),\dots,\tilde p_m(c_0)\bigr) \subseteq \{ j\in [m] \colon W_{j,0} = 1, S_{j,0} \ge S_{j_*,0} \}.
    \end{equation}
    In particular, comparing the cardinalities of these sets yields $c_0 + C(S_{j_*,0}) \le \alpha \hat k\le \alpha  R(S_{j_*,0})$, so $T \le S_{j_*,0}$ by~\eqref{eq:rewrite-T}.  Combining this with~\eqref{eq:bh-rej-set-subseteq}, we conclude that $\hat R^\bh \subseteq \hat R^\cbt$, where we note that this inclusion holds trivially if $\hat R^\bh = \emptyset$, i.e. $\hat{k} = 0$.
    
    For the other direction, suppose again that $\hat{k} > 0$ and that $j \in \hat R^\cbt$, so that $W_{j,0} = 1$ and $S_{j,0} \ge T$.  By the   monotonicity of $C(\cdot)$ again, together with the definition of $T$,
    \begin{equation}
        \label{eq:cbt-pval-bh-thresh}
        \tilde p_j(c_0)
        \le
        \frac{c_0+C(T)}{m}
        \le
        \frac{\alpha R(T)}{m};
    \end{equation} 
    this is analogous to the self-consistency property~\citep{blanchard2008two} of the BH procedure for strict p-values. 
    Now $|\{j:W_{j,0}  =1, S_{j,0} \ge T\}| = R(T)$, so $\max\{k:\tilde p_j(c_0) \leq \alpha k/m\} \geq R(T)$.  Hence, BH will reject these $R(T)$ hypotheses, 
    so $\hat R^\cbt \subseteq \hat R^\bh$ on $\{\hat k > 0\}$.  
    
    Finally, suppose that $\hat k=0$ but for a contradiction that $R(T) = |\hat R^\cbt|>0$.  Then $|\{j:\tilde p_j(c_0) \leq \alpha R(T)/m\}| \geq R(T)$, so $\hat{k} > 0$, which establishes our contradiction.  We conclude for every $c_0 \ge 0$ that  
    $ \hat R^\cbt  =   \hat R^\bh$.
    \hfill $\qed$

\subsection{Proof of Proposition~\ref{prop:sss-compound-bh}}
\label{sec:appd-proof-sss-compound-bh}

Taking $q_j=1$, $n_j=n$, and $c_j=1/n$ for every $j\in[m]$ in Proposition~\ref{prop:qfinalist-cbt-compound}, we see that $\tilde{p}_1,\ldots,\tilde{p}_m$ are compound valid.  For the second part, we specialize Proposition~\ref{prop:cbt-bh-equivalence} to $n_j=n$ and $c_0=1/n$. The rejection set and threshold in~\eqref{eq:cbt-rejection-set} are then exactly the SSS+ rejection set and threshold in~\eqref{eq:multi-kn}. Therefore, $\hat R^\sss=\hat R^\bh(\tilde p_1,\dots,\tilde p_m)$. \hfill $\qed$

\subsection{Proof of Theorem~\ref{thm:inflated-fdr-control}} 
\label{sec:appd-proof-inflated-fdr-control}

The theorem follows by combining the bounds in Propositions~\ref{prop:cbt-logarithmic-fdr-bound} and~\ref{prop:cbt-counting-fdr-bound} below.  We first introduce notation shared by the proofs of both bound components; this notation is similar to that of the proof of the nominal FDR control of SSS+ by~\citet{barber2015controlling}, which uses a rank-based rather than score-based rejection rule. 

Throughout this subsection, take $c_0=1/\nmin$   and write 
$S_{j,(0)}
\coloneqq
\max_{0\le k\le n_j}S_{j,k}.$ 
Since each competition has a unique winner,
$\sum_{k=1}^{n_j}W_{j,k}
=
1-W_{j,0}$.  We rank the competitions (both null and nonnull) by their winning scores:
\[
    \rho_j
    \coloneqq
    1+\sum_{\ell=1}^m
    \ind\{S_{\ell,(0)}>S_{j,(0)}\},
    \qquad j\in[m].
\]
Since we assume the scores of each competition winner are all distinct,  $(\rho_j)_{j\in[m]}$ is a permutation of $[m]$. For $s\in\{0,\dots,m\}$, define the number of test winners, 
the number of test winners from null competitions, and 
the weighted number of control winners (plus the offset) among the first $s \in [m]$ competitions (in descending order of their winning score) by
\begin{equation}
    \begin{aligned}
        \label{eq:cbt-rank-fdp-process}
        R(s)
        &\coloneqq
        \sum_{j=1}^m W_{j,0}\ind\{\rho_j\le s\}, \\
        V(s)
        &\coloneqq
        \sum_{j\in\cH_0}W_{j,0}\ind\{\rho_j\le s\},
        \\
        \hat V(s)
        &\coloneqq
        \frac{1}{\nmin}
        +\sum_{j=1}^m\frac{1}{n_j} \sum_{k=1}^{n_j} W_{j,k}\ind\{\rho_j\le s\}
        =
        \frac{1}{\nmin}
        +\sum_{j=1}^m\frac{1}{n_j}(1-W_{j,0})\ind\{\rho_j\le s\}.
    \end{aligned} 
\end{equation}
We can now define
\begin{equation}
    \label{eq:tau}
    \tau
    \coloneqq
    \max\biggl(
    \biggl\{
    s\in[m]:
    \frac{\widehat V(s)}{R(s)\vee1}\le\alpha
    \biggr\}
    \cup\{0\}
    \biggr),
\end{equation}
The FDP bound of Lemma~\ref{lem:duality-of-cbt} below is a major ingredient of the proof for both the harmonic and counting inflation bounds.
\begin{lemma}
    \label{lem:duality-of-cbt}
    We have $\tau = \sum_{j=1}^m \ind\{S_{j, (0)} \ge T\}$. Furthermore,
    $\hat R^\cbt  =
    \{j\in[m]:W_{j,0}=1,\ \rho_j\le\tau\}$ and $
    \fdp(\hat R^\cbt) \le  \alpha 
    {V(\tau)} / {\widehat V(\tau)}$.
\end{lemma}
\begin{proof}%
For the first claim, observe that 
\begin{equation}
    \begin{split}
        \hat\fdp(t)\coloneqq 
        \frac{
        \frac{1}{\nmin}
        +\sum_{j=1}^m\frac{1}{n_j}\sum_{k=1}^{n_j}
        \ind\{W_{j,k}=1,\ S_{j,k}\ge t\}
        }{
        1\vee\sum_{j=1}^m
        \ind\{W_{j,0}=1,\ S_{j,0}\ge t\}
        }
        &= \frac{
        \frac{1}{\nmin}
        +\sum_{j=1}^m\frac{1}{n_j}\sum_{k=1}^{n_j} W_{j,k}
        \ind\{ S_{j,(0)}\ge t\}
        }{
        1\vee\sum_{j=1}^m W_{j,0}
        \ind\{ S_{j,(0)}\ge t\}
        }
        \\
        &=
        \frac{\hat V\bigl(\sum_{j=1}^m\ind\{S_{j,(0)}\ge t\}\bigr)}{1\vee R\bigl(\sum_{j=1}^m\ind\{S_{j,(0)}\ge t\}\bigr)},
    \end{split}
\end{equation} 
so since $\hat\fdp(T) \leq \alpha$, we have $\sum_{j=1}^m \ind\{S_{j, (0)} \ge T\} \leq \tau$.  Now suppose for a contradiction that  $\sum_{j=1}^m \ind\{S_{j, (0)} \ge T\} < \tau$, so we could find $T' < T$ such that $\tau =\sum_{j=1}^m \ind\{S_{j, (0)} \ge T'\} $. Since $\hat V(\tau) / \bigl(1 \vee R(\tau)\bigr) \le\alpha$, we have $\hat\fdp (T') \le\alpha$. But this contradicts the minimality of $T$ in~\eqref{eq:cbt-rejection-set} and therefore establishes our first claim.  For the second, we can now note that  
\[
    \{W_{j,0} = 1, S_{j,0} \ge T\} = 
    \{W_{j,0} = 1, S_{j,(0)} \ge T\}
    =
    \{W_{j,0} = 1, \rho_j \le \tau\},
\]
so $\hat R^\cbt  =
\{j\in[m]:W_{j,0}=1,\ \rho_j\le\tau\}$.  Finally, observe that $
\fdp(\hat R^\cbt)
=
{V(\tau)}/\bigl({R(\tau)\vee1}\bigr)$. 
When $V(\tau) > 0$, we have $R(\tau) \geq 1$, and by definition of $\tau$ we deduce that $\hat V(\tau) \le \alpha R(\tau)$. Hence
\begin{equation}
    \label{eq:cbt-basic-fdp-ratio-bound}
    \fdp(\hat R^\cbt)
    =\frac{V(\tau)}{R(\tau)\vee1}
    \le\alpha\frac{V(\tau)}{\widehat V(\tau)},
\end{equation}
as required.
\end{proof}

\begin{proposition}[Harmonic inflation bound]
    \label{prop:cbt-logarithmic-fdr-bound}
    Under the conditions of Theorem~\ref{thm:inflated-fdr-control}, define 
        \[
        \Lambda_0(\bm n)
        \coloneqq
        \sum_{j\in\cH_0}\frac{\nmin}{n_j+1}.
        \]
        Then 
       $ \fdr(\hat R^\cbt)
        \le
        \alpha\mathfrak{h}_{\Lambda_0(\bm n)}
        \le
        \alpha\mathfrak{h}_{\Lambda (\bm n)}$. 
\end{proposition}
\begin{proof}
    Let $m_0\coloneqq|\cH_0|$ and assume that $m_0\ge 1$ (as the FDR is 0 otherwise).
    List the null indices as $j_1,\dots,j_{m_0}$ in increasing order of $\rho_j$, or equivalently in decreasing order of winning score $S_{j, (0)}$.  We first record the conditional distribution of the null winner labels. Define the $\sigma$-algebra of competition-winning scores
    $ \cG
    \coloneqq
    \sigma\bigl(
    (S_{j,(0)})_{j\in[m]}
    \bigr)$. 
    As in Remark~\ref{rem:looser-assumptions}, 
    conditional on $\cG$, we have
    \begin{equation}
        \label{eq:cbt-log-conditional-label-law}
        W_{j_r,0}\mid\cG
        \sim
        \Bern\biggl( \frac{1}{n_{j_r}+1}\biggr ),
        \qquad r\in[m_0].
    \end{equation}
    Furthermore, we claim that conditional on $\mathcal{G}$, the variables $W_{j_1,0},\dots,W_{j_{m_0},0}$ are independent. 
    To see this, for each $j\in\cH_0$,  define 
    $\cE^{-j} \coloneqq \sigma\bigl(\cG, \bigl\{ W_{i,0} \colon i \in \cH_{0}\setminus \{j\} \bigr\}\bigr) \subseteq \cI_j$.
    Then 
    $\PP( W_{j,0} = 1 \given \cE^{-j}) = 1/(n_j + 1)$.
    Moreover, for any $(b_i)_{i \in \mathcal{H}_0}$ with $b_i \in \{0,1\}$, we have 
    \[
        \begin{aligned}
            \PP(\cap_{i \in \cH_0} \{W_{i,0} = b_i \}  \given \cG ) &= 
            \EE[\PP(  (\cap_{i \in \cH_0} \{W_{i,0} = b_i \} \given \cE^{-j}  ) \given \cG  ]\\
            &= \EE\Big[ \ind\bigl\{\cap_{i \in \cH_0\setminus \{j\}} \{W_{i,0} = b_i \} \bigr\} \cdot 
            \PP(W_{j,0}=b_j\given\cE^{-j}) \biggiven \cG \Big]\\
            &=  \PP(\cap_{i \in \cH_0\setminus \{j\}} \{W_{i,0} = b_i \}  \given \cG   )\cdot \PP(W_{j,0}=b_j\given\cG).
        \end{aligned}
    \]
    Repeating this process for the indices in $\cH_0\setminus \{j\}$ proves the claim of conditional independence of the null test winner indicators.

    Define
    \begin{equation}
        \label{eq:cbt-null-control-weight-process}
        A_0\coloneqq0,
        \qquad
        A_r
        \coloneqq
        \sum_{u=1}^r \frac{1}{n_{j_u}}(1-W_{j_u,0}),
        \qquad r\in[m_0].
    \end{equation} 
    The null competitions with indices $j_1,\ldots,j_r$ are precisely those null
    competitions whose ranks are at most $\rho_{j_r}$.
    Consequently, for every $r\in[m_0]$,
    \[ 
        \frac{1}{\nmin}+A_r
        =
        \frac{1}{\nmin}
        +\sum_{j\in\cH_0}
        \frac{1}{n_j}(1-W_{j,0})
        \ind\{\rho_j\le\rho_{j_r}\} 
        \le
        \hat V(\rho_{j_r}), 
    \]
    where the inequality follows from the fact that $\hat V(\cdot)$ includes nonnegative contributions
    from nonnull competitions. 
    By Lemma~\ref{lem:duality-of-cbt} and using the fact that $s \mapsto \hat V(s)$ is increasing,
    \begin{equation}
        \label{eq:cbt-log-pathwise}
        \begin{aligned}
            \fdp(\hat R^\cbt)
            \le
            \alpha\frac{V(\tau)}{\widehat V(\tau)}
            \le
            \alpha\sum_{r=1}^{m_0}
            \frac{W_{j_r,0}\ind\{\rho_{j_r}\le\tau\}}
            {\widehat V(\rho_{j_r})}
            \le
            \alpha\sum_{r=1}^{m_0}
            \frac{W_{j_r,0}}
            {1/\nmin+A_r}.
        \end{aligned}
    \end{equation}

    For $r\in\{0,\dots,m_0\}$, let
    \[
        \cF_r
        \coloneqq
        \sigma(\cG,W_{j_1,0},\dots,W_{j_r,0}).
    \]
    Since $A_r=A_{r-1}+(1-W_{j_r,0})/ {n_{j_r}}$, we have by~\eqref{eq:cbt-log-conditional-label-law} that
    \begin{equation}
        \label{eq:cbt-log-one-step}
        \EE\biggl[
        \frac{W_{j_r,0}}{1/\nmin+A_r}
        \bigggiven
        \cF_{r-1}
        \biggr] =
        \frac{{1}/({n_{j_r}+1})}{1/\nmin+A_{r-1}} 
        = \biggl( \frac{n_{j_r}}{n_{j_r}+1}\biggr) \cdot 
        \frac{\nmin / n_{j_r}}{1+\nmin A_{r-1}}.
    \end{equation}
    Using the fact that $t\mapsto (1-x^t) / (1-x)$ is concave on $[0,\infty)$ for each $x\in(0,1)$, we see that  $H(t)\coloneqq \mathfrak{h}_t$ is concave.  From~\eqref{eq:cbt-log-one-step} and applying  Lemma~\ref{lem:harmonic-increment}, we deduce that  
    \begin{align}
        \label{eq:conversion-to-harmonic-telescope-1}
    \EE\biggl[
        \frac{W_{j_r,0}}{1/\nmin+A_r}
        \bigggiven
        \cF_{r-1}
        \biggr] &= 
        \biggl( \frac{n_{j_r}}{n_{j_r}+1}\biggr) \cdot 
        \frac{\nmin / n_{j_r}}{1+\nmin A_{r-1}} \nonumber \\
        &\le \biggl( \frac{n_{j_r}}{n_{j_r}+1}\biggr) \bigl\{ H(\nmin A_{r-1} + \nmin / n _{j_r} ) - H(\nmin A_{r-1} ) \bigr\}.
    \end{align}
    Since $\PP(W_{j_r,0} = 0 \given \cF_{r-1}) = \PP(W_{j_r,0} = 0 \given \cG) = n_{j_r} / ({n_{j_r}+1}) $
    by~\eqref{eq:cbt-log-conditional-label-law}
    and the conditional independence of $(W_{j_r})_{r \in [m_0]}$ given $\cG$, 
    we find that
    \begin{equation}
        \label{eq:conversion-to-harmonic-telescope-2}
        \biggl( \frac{n_{j_r}}{n_{j_r}+1}\biggr) \bigl\{ H(\nmin A_{r-1} + \nmin / n _{j_r} ) - H(\nmin A_{r-1} ) \bigr\}
        =
        \EE\bigl[ H(\nmin A_r) - H(\nmin A_{r-1}) \given \cF_{r-1} \bigr]
    \end{equation}
    since $A_r = A_{r-1}$ when $W_{j_r,0} = 1$ and 
    $A_r = A_{r-1} + 1 / n_{j_r}$ when $W_{j_r,0}  = 0$. 

    It follows from~\eqref{eq:cbt-log-pathwise},~\eqref{eq:conversion-to-harmonic-telescope-1},~\eqref{eq:conversion-to-harmonic-telescope-2} and Jensen's inequality that
    \begin{align}
        \fdr(\hat{R}^{\cbt}) &\leq   \alpha \EE \biggl[\sum_{r=1}^{m_0}
        \frac{W_{j_r,0}}{1/\nmin+A_r}\biggr] = \alpha \EE \biggl[\sum_{r=1}^{m_0}
        \EE\Bigl[\frac{W_{j_r,0}}{1/\nmin+A_r}\bigggiven \cF_{r-1}\Bigr]\biggr]\\
        & \le \alpha \EE\biggl[\sum_{r=1}^{m_0} \EE\bigl[ H(\nmin A_r) - H(\nmin A_{r-1}) \given \cF_{r-1} \bigr]\biggr]  = \alpha  \EE\bigl[ H(\nmin A_{m_0}) - H(\nmin A_{0})  \bigr] \\
        &= \alpha  \EE\bigl[ H(\nmin A_{m_0})\bigr] \le \alpha  H \bigl( \nmin  \EE[A_{m_0}]\bigr)= \alpha H\biggl( \nmin \sum_{j\in\cH_0} \frac{1}{n_j} \EE[1-W_{j,0}] \biggr)
         \\
         &= \alpha H\biggl( \sum_{j \in\cH_0} \frac{\nmin }{n_j + 1} \biggr) =\alpha  H \bigl( \Lambda_0(\bm{n})\bigr) \leq \alpha H \bigl( \Lambda(\bm{n})\bigr),
         \end{align}
         as required.
    \end{proof}

\begin{lemma}
        \label{lem:harmonic-increment}
        For all $z\ge 0$ and $y \in [0,1]$, we have $H(z+y)  -H(z) \ge y / (1+z)$.
    \end{lemma}

\begin{proof} %
    By concavity,
    \begin{align*}
        H(z+y) - H(z) &\ge (1-y) H(z) + y H(z+1) - H(z) = y\bigl(H(z+1)-H(z)\bigr) = y\int_0^1 x^z \, dz = \frac{y}{1+z},
    \end{align*}
    as required.
\end{proof}

\begin{proposition}[Counting inflation bound]
    \label{prop:cbt-counting-fdr-bound}
    Under the conditions of Theorem~\ref{thm:inflated-fdr-control},
    \[
        \fdr(\hat R^\cbt)
        \le
        \alpha C(\bm n).
    \]
\end{proposition} 
\begin{proof}
    Fix $j\in\cH_0$. Define the $j$th-perturbed versions of the processes $R(\cdot)$ and $\hat V(\cdot)$ in~\eqref{eq:cbt-rank-fdp-process} by
    \begin{align*}
        R^{j,-}(s)
        &\coloneqq
        \ind\{\rho_j\le s\}
        +\sum_{\ell\ne j}W_{\ell,0}\ind\{\rho_\ell\le s\},\\
        \widehat V^{j,-}(s)
        &\coloneqq
        \frac{1}{\nmin}
        +\sum_{\ell\ne j}
        \frac{1}{n_\ell}(1-W_{\ell,0})\ind\{\rho_\ell\le s\},
    \end{align*}
    and let
    \begin{equation}
        \label{eq:tau-perturbed}
        \tau_j^-
        \coloneqq
        \max\biggl(
        \biggl\{
        s\in[m]:
        \frac{\widehat V^{j,-}(s)}{R^{j,-}(s)\vee1}
        \le\alpha
        \biggr\}
        \cup\{0\}
        \biggr),
        \qquad
        \widehat V_{j,-}
        \coloneqq
        \widehat V^{j,-}(\tau_j^-).
    \end{equation} 
    These processes coincide with corresponding processes for their original versions (without the ``$j,-$'' superscript)  
    on the event that competition $j$ has a test winner.  
    On the event that competition $j$ has a control winner, $R^{j, -}(s) = R(s) + \ind\{\rho_j \leq s\}$, while $\hat V^{j,-}(s) = \hat V(s) - \ind\{ \rho_j \le s \}/n_j $, so for these processes depart from their original versions when $S_{j,(0)}$ ranks in the top $s$ of all competition-winning scores.  Moreover, the random variables $\rho_j$, $\tau_j^-$, and $\widehat V_{j,-}$ are $\cI_j$-measurable, so by~\eqref{eq:cbt-conditional-winner-balance} 
    \begin{align}
        \label{eq:fdr-contribution-loo}
        \EE\biggl[
        \frac{W_{j,0}\ind\{\rho_j\le\tau\}}
        {\widehat V(\tau)}
        \biggr]
        &= 
        \EE\biggl[
        \frac{W_{j,0}
        \ind\{\rho_j\le\tau_j^-\}}
        {\widehat V_{j,-}}
        \biggr] 
        =
        \EE\biggl[
        \frac{ 
        \ind\{\rho_j\le\tau_j^-\}}
        {\widehat V_{j,-}} \EE[W_{j,0}\given \cI_j]
        \biggr] \\
        &= \EE\biggl[
        \frac{ 
        \ind\{\rho_j\le\tau_j^-\}}
        {\widehat V_{j,-}} \EE\Bigl[\frac{1}{n_j}(1-W_{j,0})\given \cI_j\Bigr]
        \biggr]
        = 
        \EE\biggl[
        \frac{\frac{1}{n_j}(1-W_{j,0})
        \ind\{\rho_j\le\tau_j^-\}}
        {\widehat V_{j,-}}
        \biggr].
    \end{align} 
    Hence, by Lemma~\ref{lem:duality-of-cbt}, $\fdr(\hat R^\cbt)\le  \alpha \EE[\Gamma]$, where
    \begin{equation}
    \label{eq:gamma-defn}
        \Gamma
        \coloneqq
        \sum_{j\in\cH_0}
        \frac{\frac{1}{n_j}(1-W_{j,0})
        \ind\{\rho_j\le\tau_j^-\}}
        {\widehat V_{j,-}}.
    \end{equation} 
    
    We now provide an upper bound for $\Gamma$.
    Let $D(\bm n)\coloneqq\{n_j:j\in[m]\}$ denote the set of distinct control counts.
    For each $d\in D(\bm n)$, define
    \[
        \cJ_d
        \coloneqq
        \{j\in\cH_0:
        n_j=d,\ W_{j,0}=0,\ \rho_j\le\tau_j^-\},
        \qquad
        k_d
        \coloneqq
        |\cJ_d|.
    \]
    We first show that $\tau_j^- = \tau_{j'}^-$ for $j,j'\in\cJ_d$.  Fix $j,j'\in\cJ_d$, and suppose without loss of generality that $\rho_j<\rho_{j'}$. Since $W_{j,0}=W_{j',0}=0$ and $n_j=n_{j'}=d$, we have for every $s\ge\rho_{j'}$ that $R^{j,-}(s)=R^{j',-}(s)$ and $\hat{V}^{j,-}(s)=\hat V^{j',-}(s)$.  It follows that the $j$th and $j'$th perturbed processes agree at $s=\tau_{j'}^- \geq \rho_{j'}$, so using definition of $\tau_{j'}^-$, we have that
    \[
        \frac{\hat V^{j,-}(\tau_{j'}^-)}
        {R^{j,-}(\tau_{j'}^-)\vee1} = \frac{\hat V^{j',-}(\tau_{j'}^-)}
        {R^{j',-}(\tau_{j'}^-)\vee1} \le\alpha.
    \]
    By maximality of $\tau_j^-$, we deduce that $\tau_j^-\ge\tau_{j'}^- \geq \rho_{j'}$.  Thus the two perturbed processes also agree at $s=\tau_j^-$, and hence
    \[
        \frac{\hat V^{j',-}(\tau_j^-)}
        {R^{j',-}(\tau_j^-)\vee1}
        =
        \frac{\hat V^{j,-}(\tau_j^-)}
        {R^{j,-}(\tau_j^-)\vee1}
        \le\alpha.
    \]
    The maximality of $\tau_{j'}^-$ now gives $\tau_{j'}^-\ge\tau_j^-$, so $\tau_j^-=\tau_{j'}^-$.

    When $\cJ_d$ is nonempty, we denote the common value of $\tau_j^-$ for $j\in \cJ_d$ as $u_d$. For every $j\in\cJ_d$, we have $\hat V_{j,-} = \hat V(u_d)-{1}/{d} \ge 1/\nmin + (k_d-1)/d$.
    Therefore, 
    \begin{equation}
        \label{eq:cbt-counting-class-bound}
        \sum_{j\in\cJ_d}
        \frac{1/n_j}{\widehat V_{j,-}}
        \le
        \frac{k_d}{k_d-1+d/\nmin} \leq \frac{h(d)}{{h(d)-1+d/\nmin}},
    \end{equation}
    where in the final inequality, we used the fact that for $d\ge\nmin$, the function 
    $k \mapsto  {k}/({k-1+d/\nmin})$ 
    is increasing on $\mathbb{N}_0$ and $k_d\le | \{j \in [m]\colon n_j = d\}| = h(d)$.  We conclude that
    \[
        \Gamma \le \sum_{d\in D(\bm n)} \sum_{j\in\cJ_d}
        \frac{1/n_j}
        {\widehat V_{j,-}} 
        \le
        \sum_{d\in D(\bm n)}
        \frac{h(d)}{h(d)-1+d/\nmin}
        = \sum_{j=1}^m 
        \frac{1}{h(d)-1+n_j/\nmin} = 
        C(\bm n),
    \]
    as required.
\end{proof}

\subsection{Proof of Proposition~\ref{prop:lower-bound}}
\label{sec:appd-proof-lower-bound}

    We use the rank-based notation  from the common setup in
    Appendix~\ref{sec:appd-proof-inflated-fdr-control},
    with $c_0=1/n_{\min}$.
    In particular, $R(s)$, $\widehat V(s)$, $V(s)$, and $\tau$ are defined in~\eqref{eq:cbt-rank-fdp-process} and~\eqref{eq:tau}.  
    For $N\in\NN$, let
    \[
        L_N
        \coloneqq
        \lfloor N^{1/2}\rfloor,
        \qquad
        \alpha_N
        \coloneqq
        \frac{1}{4N^2}.
    \]
    When there is no confusion, we omit the $N$ subscript for brevity. 

\paragraph{Constructing the competition problem instance.}
    Our construction has $L$ null competitions. For $r\in[L]$, define
    \[
        w_r
        \coloneqq
        \frac{1}{N-r+1},
        \qquad
        \overline A_r
        \coloneqq
        \sum_{u=1}^r w_u,
    \]
    with $\overline A_0\coloneqq0$.  
    The null competition corresponding to the index $r$ will have $N-r+1$ controls. All nonnull competitions introduced below will have $N^4$ controls, so $n_{\min} = N-L+1$.  We stratify the competitions into $L$ stages. Define
    \[
        D_r
        \coloneqq
        \bigg\lceil
        \frac{1/n_{\min}+\overline A_{r-1}}{\alpha}
        \bigg\rceil-1,
        \qquad r\in[L],
    \]
    with $D_0\coloneqq0$; further, define the differences $d_r\coloneqq D_r-D_{r-1}$. Since $\overline A_r$ is increasing, $d_r\ge0$.  Stage~$r$ consists of one null competition followed by $d_r$ nonnull competitions, so the total number of competitions is $m_N = L+D_L$.  Denote the index of the sole null competition in stage $r$ by
        $j_r
        \coloneqq
        r+D_{r-1}$, 
    and the index of the final competition in stage $r$ by
       $ s_r
        \coloneqq
        r+D_r$. 
    In particular, $\cH_0=\{j_1,\dots,j_L\}$.

        We next populate the competitions with scores consistent with this
    ordering. For competition $j \in [m_N]$ and for $k\in \{0, \dots, n_j\}$, define $c_{j,k}
        \coloneqq
        ({n_j-k})/(4n_j+1)$, 
    so that $c_{j,0}>c_{j,1}>\cdots>c_{j,n_j}=0$.  For $r \in [L]$, draw independent uniform permutations $\Pi_r$ of $\{0,\dots,n_{j_r}\}$ and set
    \[
        S_{j_r,k}
        \coloneqq
        -j_r+c_{j_r,\Pi_r(k)}
    \]
    for each $k\in \{0, \dots, n_{j_r}\}$. 
    Then 
    \[
        \pi_r
        \coloneqq
        \PP(W_{j_r,0}=1)
        =
        \PP\big(\Pi_r(0)=0\big)
        =
        \frac{1}{n_{j_r}+1}
        =
        \frac{1}{N-r+2}.
    \]
    Moreover, $W_{j_1,0},\dots,W_{j_L,0}$ are independent because
    the permutations are independent.

    For each \emph{nonnull} competition $j\notin\cH_0$ and for $k\in \{0, 1,\dots, n_j\}$, set
    \[
        S_{j,k}
        \coloneqq
        -j+c_{j,k}.
    \]
    Thus $W_{j,0}=1$ with probability 1.
    For both null and nonnull $j$, the scores within competition $j$ are distinct
    and lie in $[-j,-j+1/4)$, so there are no ties.  This construction therefore satisfies all the conditions of
    Theorem~\ref{thm:inflated-fdr-control}.   Additionally, the  scores across different $j$ are mutually independent (and the nonnull competitions are not random at all).  The winning scores satisfy
    \[
        S_{j,(0)}
        =
        -j+\frac{n_j}{4n_j+1},
    \]
    so $S_{1,(0)}>S_{2,(0)}>\cdots>S_{m_N,(0)}$.  
    Consequently, 
    \begin{equation}
        \label{eq:cbt-lower-deterministic-ranks}
        \rho_j
        =
        1+\sum_{\ell=1}^{m_N}
        \ind\{S_{\ell,(0)}>S_{j,(0)}\}
        =
        j,
        \qquad j\in[m_N].
    \end{equation}

\paragraph{A lower bound for the attained FDR.}
    For $r\in[L]$, consider the event
    \[
        E_r
        \coloneqq
        \{W_{j_r,0}=1\}
        \cap
        \bigcap_{u \in [L] \setminus\{ r \} }\{W_{j_u,0}=0\}.
    \]
    Thus, on $E_r$, the null competition in stage $r$ is the unique null competition with a test winner. We claim that, for all $N \in\NN$ and on the event $E_r$,  we have $\tau = s_r$ and
    \begin{equation}
        \label{eq:cbt-lower-stage-behavior}
        \fdp(\hat R^\cbt)
        =
        \frac{1}{D_r+1}.
    \end{equation}
To prove the claim, fix $r\in[L]$ and work on $E_r$. At the endpoint $s_k$ of stage $k$, the first $s_k$ competitions contain the first $k$ null competitions and $D_k$ nonnull competitions. Since every nonnull competition has a test winner, by~\eqref{eq:cbt-rank-fdp-process} and~\eqref{eq:cbt-lower-deterministic-ranks},
    \[
        \bigl(R(s_k),\widehat V(s_k)\bigr)
        =
        \begin{cases}
            (D_k,\ 1/n_{\min}+\overline A_k),
            &k \leq r-1,\\
            (D_k+1,\ 1/n_{\min}+\overline A_k-w_r),
            &k\ge r.
        \end{cases}
    \]
    In particular,
    \[
        \frac{\widehat V(s_r)}{R(s_r)}
        =
        \frac{1/n_{\min}+\overline A_{r-1}}{D_r+1}
        \le
        \alpha,
    \]
    where the inequality follows from the definition of $D_r$.

    Now fix $k>r$.  Then $\alpha(D_k+1)
    \le
    {1}/{n_{\min}}+\overline A_{k-1}+\alpha$ and 
    \[
        \widehat V(s_k)
        =
        \frac{1}{n_{\min}}+\overline A_{k-1}+w_k-w_r \geq \frac{1}{n_{\min}}+\overline A_{k-1}+ w_k-w_{k-1} > \frac{1}{n_{\min}}+\overline A_{k-1} + \frac{1}{N^2} > \frac{1}{n_{\min}}+\overline A_{k-1} + 
        \alpha.
    \]
    It follows that $\widehat V(s_k)/R(s_k) > \alpha$.  More generally, every $s$ satisfying $s_{k-1}<s\le s_k$ has $\widehat V(s)=\widehat V(s_k)$ and $R(s)\le R(s_k)$.  Hence
    \[
        \frac{\widehat V(s)}{R(s)\vee1}
        \ge
        \frac{\widehat V(s_k)}{R(s_k)}
        >
        \alpha.
    \]
    Thus, $\tau=s_r$ on $E_r$.  Moreover, if rejecting at this threshold, then exactly one of the $D_r+1$ rejections is null, which proves~\eqref{eq:cbt-lower-stage-behavior}.

    Now define  
    $\lambda_{0,N}
    \coloneqq
    \sum_{r=1}^L\pi_r$.  The events $E_1,\dots,E_L$ are disjoint, and independence of the null competitions gives that 
\begin{align*}
\PP(E_r)
    =
    \pi_r\prod_{u \in [L] \setminus  \{r\}}(1-\pi_u) \geq \pi_r\biggl(1-\sum_{u \in [L] \setminus  \{r\}}\pi_u\biggr) \geq (1-\lambda_{0,N})\pi_r. 
\end{align*}
    Combining these observations with~\eqref{eq:cbt-lower-stage-behavior} yields
    \begin{equation}
        \label{eq:cbt-lower-fdr-sum}
        \fdr_N = \EE\left[ \fdp(\hat R^\cbt) \right] \ge \sum_{r=1}^L \left(\PP(E_r)\cdot \frac{1}{D_r+1}\right)
        \ge
        (1-\lambda_{0,N})\alpha
        \sum_{r=1}^L
        \frac{\pi_r}{1/n_{\min}+\overline A_{r-1}+\alpha}.
    \end{equation}
    Now, for $r\in[L_N]$,
    \[
        \frac{\pi_r}{w_r}
        =
        \frac{N-r+1}{N-r+2}
        =
        1-  \frac{1}{N-r+2} \ge 1- \frac{1}{N - L_N +2}, %
    \]
    and 
    \[
        \frac{1}{1/n_{\min}+\overline A_{r-1}+\alpha}
        \ge
        \frac{1-\alpha n_{\min}}
        {1/n_{\min}+\overline A_{r-1}}.
    \] 
    Further, with the fact that $\overline A_r=\overline A_{r-1}+w_r$ and
    $\overline A_0=0$, we invoke the inequality $\log(1+x)\le x$ for $x\ge0$ to obtain 
    \begin{equation}
        \label{eq:cbt-lower-log-telescope}
        \begin{aligned}
            \sum_{r=1}^L
            \frac{w_r}{1/n_{\min}+\overline A_{r-1}}
            &\ge
            \sum_{r=1}^L
            \log\biggl(
            1+\frac{w_r}{1/n_{\min}+\overline A_{r-1}}
            \biggr)\\
            &=
            \sum_{r=1}^L
            \bigl[
            \log(1+n_{\min}\overline A_r)
            -\log(1+n_{\min}\overline A_{r-1})
            \bigr]=
            \log(1+n_{\min}\overline A_L).
        \end{aligned}
    \end{equation}
    Applying the preceding bounds 
    to~\eqref{eq:cbt-lower-fdr-sum} gives %
    \begin{equation}
        \label{eq:cbt-lower-harmonic-sum}
        \begin{aligned}
            \fdr_N
            &\ge
            (1-\lambda_{0,N})
            \biggl(1-\frac{1}{N-L_N+2}\biggr)
            (1-\alpha n_{\min})\alpha\log(1+n_{\min}\overline A_L)
            \\
        \end{aligned}
    \end{equation}
Now, using the fact that 
\[
 D_L =
            \bigg\lceil
            \frac{1/n_{\min}+\overline A_{L-1}}{\alpha}
            \bigg\rceil\le
            \frac{L}{\alpha n_{\min}} =
            \frac{4N^2L}{N-L+1},
\]
we obtain
\begin{align}
\label{Eq:ALLambdaN}
|n_{\min}\overline A_L - \Lambda_N| = \biggl|n_{\min}\overline A_L - n_{\min} \sum_{r=1}^L \pi_r - \frac{n_{\min}D_L}{N^4+1}\biggr| &\leq n_{\min}\biggl(\frac{1}{N-L+1} - \frac{1}{N+1}\biggr) + \frac{n_{\min}D_L}{N^4+1} \nonumber \\
&\leq \frac{L}{N+1} + \frac{4N^2L}{N^4+1} \rightarrow 0.
\end{align}
Since $\lambda_{0,N} \leq L/(N+1) \rightarrow 0$, $\alpha n_{\min} = (N-L+1)/(4N^2) \rightarrow 0$ and $n_{\min} \overline A_L \geq L(N-L+1)/N \rightarrow \infty$, we conclude from~\eqref{eq:cbt-lower-harmonic-sum} that
\[
\fdr_N \geq \{1-o(1)\}\alpha_N\log(1+\Lambda_N),
\]
which proves~\eqref{eq:fdr-log-lower-bound}. 

Finally, from~\eqref{Eq:ALLambdaN} and the fact that $n_{\min} \overline A_L \rightarrow \infty$, we have that $\log(1+\Lambda_N)\to\infty$.

\subsection{Proof of Proposition~\ref{prop:fdr-qfinalist}}

Throughout the proof, take $n_j = n$ and $q_j = q$ for every $j\in [m]$, and use the within-competition ranks $\rho_{j,k}$  and the processes $C_{\bm q}(\cdot)$ and $C^{(-j)}_{\bm q} (\cdot)$ introduced in the proof of Proposition~\ref{prop:qfinalist-cbt-compound}.  For $t\in \RR$, define %
\[
    R_{\bm q}(t) \coloneqq \sum_{\ell =1}^m \ind\{ W^{(q)}_{\ell,0} = 1, S_{\ell, 0} \ge t \}
\]
and let $\cS_q \coloneqq \{ S_{j,k} \colon W^{(q)}_{j,k} = 1, j \in [m], 0 \le k \le n \}$ 
denote the set of all $q$-finalist scores in all competitions. 
Define
\begin{equation}
    T^{(\bm q)} \coloneqq \min \biggl\{
        t \in \cS_q \colon \frac{1/n + C_{\bm q} (t)}{R_{\bm q}(t) \vee 1} \le \alpha
    \biggr\},
\end{equation}
with $\min \varnothing \coloneqq \infty$. The same argument used in the proof of Proposition~\ref{prop:cbt-bh-equivalence} can be applied to prove that
\begin{equation}
    \label{eq:cbt-bh-qfin-equivalence}
    \hat R^\bh\coloneqq \hat R^\bh ( \tilde p_1^{(q)} , \dots, \tilde p_m^{(q)}  ) = \{ j \in [m]\colon W_{j,0}^{(q)}  = 1, S_{j,0} \ge T^{(\bm q)} \}. 
\end{equation} 
Similarly to~\eqref{eq:cbt-basic-fdp-ratio-bound}, 
\begin{equation}
    \label{eq:qfin-fdp-bound}
        \fdp(\hat R^\bh) = \frac{\sum_{j\in \cH_0} \ind\{ W_{j,0}^{(q)} = 1, S_{j,0} \ge T^{(\bm q)} \}}{R_{\bm q}(T^{(\bm q)}) \vee 1} \le \alpha \cdot n \frac{\sum_{j\in \cH_0} \ind\{ W_{j,0}^{(q)} = 1, S_{j, 0} \ge T^{(\bm q)}\}}{1 + n C_{\bm q}(T^{(\bm q)})}.
\end{equation}
For each $j \in [m]$, write the ordered scores within competition $j$ as $S_{j,(0)} > \cdots > S_{j, (n)}$.  
Define  the perturbed processes
\begin{equation}
    \begin{split}
         R^{(j)}(t)  &\coloneqq \ind\{ S_{j,(0)} \ge t\} + \sum_{\ell \neq j} \ind\{ W_{\ell, 0}^{(q)}=1 , S_{\ell, 0} \ge t \} \\
         C^{(j)}(t) &\coloneqq \sum_{r=1}^{q-1} \ind\{S_{j,(r)} \ge t\} + \sum_{\ell  \neq j} \sum_{k=1}^n \ind\{ W_{\ell, k}^{(q)} = 1, S_{\ell, k} \ge t \}
    \end{split}
\end{equation} 
and define the perturbed version of $T^{(\bm q)}$ by
\begin{equation}
    T_j  \coloneqq \min \biggl\{
        t \in \cS_q \colon \frac{1 + C^{(j)}(t) }{n \big(R^{(j)}(t) \vee 1\big)} \le \alpha
    \biggr\}.
\end{equation}
The pair $\bigl( R_{\bm q}(t), C_{\bm q}(t) \bigr)$ coincides with $\bigl(R^{(j)}(t), C^{(j)}(t) / n\bigr)$ on the event $\{ W_{j,0}^{(q)} = 1 , S_{j,0} \ge t\}$;
otherwise, $R_{\bm q}(t) \le R^{(j)}(t)$ and $C_{\bm q}(t)\ge C^{(j)}(t)/n$. Hence $T_j \le T^{(\bm q)}$ for $j\in [m]$.  Furthermore, when $j \in \hat R^\bh$, we have from~\eqref{eq:cbt-bh-qfin-equivalence} that $W_{j,0}^{(q)}  =1$ and $S_{j,0} \ge T^{(\bm q)} \ge T_j$. Thus
\[
       \frac{1/n + C_{\bm q} (T_j)}{R_{\bm q}(T_j) \vee 1} =\frac{1 + C^{(j)}(T_j) }{n \big(R^{(j)}(T_j) \vee 1\big)}\le \alpha ,
\]
so by minimality of $T^{(\bm q)}$ we conclude that $T^{(\bm q)}\le T_j$.  We deduce that $T_j=T^{(\bm q)}$ for $j\in \hat R^\bh$ and hence
\begin{equation}
    \label{eq:qfin-frac-bound-1}
    \sum_{j\in\cH_0} \frac{\ind\{j \in \hat{R}^{\mathrm{BH}}\}}{1 + n C_{\bm q}(T^{(\bm q)})} = \sum_{j\in\cH_0} \frac{\ind\{ W_{j,0}^{(q)} = 1, S_{j, 0} \ge T^{(\bm q)}\}}{1 + n C_{\bm q}(T^{(\bm q)})}
    \le \sum_{j\in\cH_0} \frac{\ind\{ W_{j,0}^{(q)} = 1, S_{j, 0} \ge T_j \}}{1 + n C_{\bm q}(T_j)}.
\end{equation}
Now, for $j \in \cH_0$, we have that $T_j$ and $S_{j,(0)},\ldots,S_{j,(n)}$ are $\cI_j$-measurable.  On the event $\{W_{j,0}^{(q)} = 1, S_{j,0} \ge T_j\}$, the test score is a $q$-finalist in competition $j$, so 
\[ 
n C_{\bm q}(T_j)  = \biggl(\sum_{r=0}^{q-1} \ind\{ S_{j, (r)}\ge T_j \} -1\bigg)+ n C_{\bm q}^{(-j)} (T_j), 
\]
and  
\[
    \PP(W_{j,0}^{(q)} = 1, S_{j,0} \ge T_j  \given \cI_j )  = \sum_{r=0}^{q-1} \PP( S_{j,0} = S_{j,(r)}, S_{j,(r)} \ge T_j \given \cI_j ) = \frac{1}{n+1}  \sum_{r=0}^{q-1} \ind\{ S_{j, (r)}\ge T_j \}.
\]
On the other hand, on the event $\{W_{j,0}^{(q)} = 0\}$, all of the $q$-finalist scores in competition $j$ are control scores, so
\[
\sum_{k=1}^n \ind\{ W_{j,k}^{(q)} = 1, S_{j,k} \ge T_j \} = \sum_{r=0}^{q-1}\ind\{ S_{j,(r)} \ge T_j \}, \quad 
n C_{\bm q}(T_j)  = \sum_{r=0}^{q-1} \ind\{ S_{j, (r)}\ge T_j \} + n C_{\bm q}^{(-j)} (T_j). 
\]
Since $\mathbb{P}\bigl(W_{j,0}^{(q)} = 0 \biggiven \cI_j\bigr) = (n-q+1)/ (n+1)$, it follows that for $j \in \cH_0$,
\begin{align} 
    \label{eq:qfin-frac-bound-2}
        \EE \biggl[\frac{\ind\{ W^{(q)}_{j,0} = 1, S_{j, 0} \ge T_j \}}{1 + n C_{\bm q}(T_j)}  \biggr] &= \EE\biggl[ \EE \biggl[\frac{\ind\{ W_{j,0}^{(q)} = 1, S_{j, 0} \ge T_j \}}{1 + n C_{\bm q}(T_j)} \bigggiven \cI_j \biggr] \biggr] \nonumber \\
        &= \EE\biggl[ \frac{\frac{1}{n+1}  \sum_{r=0}^{q-1} \ind\{ S_{j, (r)}\ge T_j \} }{\sum_{r=0}^{q-1} \ind\{ S_{j, (r)} \ge T_j \} + n C_{\bm q}^{(-j)}(T_j)} \biggr] \nonumber \\ 
        &=  \EE\biggl[ \frac{\EE[ \ind\{W_{j,0}^{(q)} =0 \} \given \cI_j]}{n-q+1} \cdot  \frac{  \sum_{r=0}^{q-1} \ind\{ S_{j, (r)}\ge T_j \} }{\sum_{r=0}^{q-1} \ind\{ S_{j, (r)} \ge T_j \} + n C_{\bm q}^{(-j)}(T_j)}   \biggr] \nonumber \\
        &=  \EE\biggl[ \frac{1}{n-q+1} \EE\biggl[ \ind\{ W_{j, 0}^{(q)} = 0 \} \frac{  \sum_{r=0}^{q-1} \ind\{ S_{j, (r)}\ge T_j \} }{\sum_{r=0}^{q-1} \ind\{ S_{j, (r)} \ge T_j \} + n C_{\bm q}^{(-j)}(T_j)} \bigggiven \cI_j  \biggr]  \biggr] \nonumber \\
        &= \EE\biggl[ \frac{1}{n-q+1}\EE\biggl[ 
        \frac{ \ind\{ W_{j,0}^{(q)} = 0 \} \sum_{k=1}^n \ind\{W_{j,k}^{(q)}=1, S_{j,k} \ge T_j \} }{n C_{\bm q} (T_j)}
        \bigggiven\cI_j\biggr] \biggr] \nonumber \\
        &=    \frac{1}{n-q+1}\EE\biggl[ 
        \frac{ \ind\{ W_{j,0}^{(q)} = 0 \} \sum_{k=1}^n \ind\{W_{j,k}^{(q)}=1, S_{j,k} \ge T_j \} }{n C_{\bm q} (T_j)}
         \biggr]. 
\end{align}
Combining the results of~\eqref{eq:qfin-fdp-bound},~\eqref{eq:qfin-frac-bound-1}, and~\eqref{eq:qfin-frac-bound-2}, we have
\begin{equation}
    \fdr(\hat R^\bh) = \EE[ \fdp(\hat R^\bh) ] \le \alpha\frac{n}{n - q +1} \sum_{j\in\cH_0}\EE\biggl[ 
        \frac{ \ind\{ W_{j,0}^{(q)} = 0 \} \sum_{k=1}^n \ind\{W_{j,k}^{(q)}=1, S_{j,k} \ge T_j \} }{n C_{\bm q} (T_j)}
         \biggr]. 
\end{equation}
Now let
\begin{equation}
    \cJ \coloneqq \biggl\{ j \in \cH_0\colon W_{j,0}^{(q)} = 0, \sum_{k=1}^n \ind\{ W_{j,k}^{(q)} = 1, S_{j,k} \ge T_j \} > 0 \biggr\},
\end{equation}
We claim that $T_j$ takes a common value, denoted $T_*$, for $j \in \mathcal{J}$.  To see this, we follow a similar argument to the one found in the proof of Proposition~\ref{prop:cbt-counting-fdr-bound}. Fix $j\in \cJ$, so that $S_{j,(0)} \ge T_j$.
For $j' \in \cH_0$ and $t\le T_j$, we have $R^{(j')} (t) \le R^{(j)} (t)$ and $C^{(j')} (t) \ge C^{(j)}(t)$, so 
\begin{equation}
    \frac{1 + C^{(j')} (t) }{n \bigl( R^{(j')} (t) \vee 1 \bigr)}  \ge  \frac{1 + C^{(j)} (t) }{n \bigl( R^{(j)} (t) \vee 1 \bigr)}
\end{equation}
for such $j'$ and $t$. 
By the minimality of $T_j$, we have for all $t \in \cS_q$ with $t < T_j$ that the right-hand side is strictly greater than $\alpha$, so $T_{j'} \ge T_j$.  On the other hand, if $j' \in \cJ$, we may reverse this argument, giving $T_j \ge T_{j'}$ and so $T_j = T_{j'}$. 

We conclude that
\begin{equation}
    \begin{split}
        \sum_{j\in\cH_0}\EE\biggl[ 
        \frac{ \ind\{ W_{j,0}^{(q)} = 0 \} \sum_{k=1}^n \ind\{W_{j,k}^{(q)}=1, S_{j,k} \ge T_j \} }{n C_{\bm q} (T_j)}
         \biggr] = \EE\biggl[
         \frac{\sum_{j \in \cJ} \sum_{k=1}^n \ind\{ W_{j,k}^{(q)} = 1, S_{j,k} \ge T_* \} }{\sum_{\ell =1}^m  \sum_{k=1}^n \ind\{ W_{\ell ,k}^{(q)} = 1, S_{\ell ,k} \ge T_* \}}
         \biggr] \le 1,
    \end{split}
\end{equation}
as desired.

%
%
%
%
%

%
 
%
%
%
%
%
%
%
%
%
%
%
%
%
%
%
%
%
%
%
%

%
%
%
%
%
%
%
%
%
%
%
%
%
%
%
%
%

%
%
%
%
%
%
%
%
%
%
%
%
%
%
%
%
%
%
%

%
%
%
%
%
%
%
%
%
%
%
%
%
%
%
%
%
%
%
%
%
%
%
%
%
%
%
%
%
%
%

%
%
%
%
%
%
%
%
%
%
%
%
%
%
%
%
%
%
%
%
%
%
%
%
%
%
%
%
%
%
%
%
%
%
%
%
%
%

%
%
%
%
%
%
%
%
%
%
%
%
%
%
%
%

%
%
%
%
%
%
%
%
%
%
%
%
%
%
%
%
%
%
%
%
%
%
%
%
%
%
%
%
%

%
%
%
%
%
%
%
%
%
%
%
%
%
%

%
%
%
%
%
%
%

%
%
%
%
%
%
%
%
%
%
%
%
%
%
%
%
%
%
%
%
%
%
%
%
%
%
%
%
%
%
%
%
%
%
%
%
%
%
%
%
%
%
%
%
%
%
%
%
%
%
%
%
%
%
%
%
%
%
%
%
%
%
%
%
%
%
%
%
%
%
%
%
%
%

%
%
%
%
%
%
%
%
%
%
%

%
%
%
%
%
%
%
%
%
%
%
%
%

%
%
%
%
%
%
%
%
%
%
%
%
%
%
%
%
%
%
%
%
%
%
%
%
%
%
%
%
%

%
%
%
%

%
%
%
%
%
%
%
%
%
%
%
%
%
%
%
%
%
%
%

%
%
%
%
%
%
%
%
%
%
%
%
%
%
%
%
%

%% file: sections_appendix/appd_evalues.tex

\section{Proofs and additional results on compound e-values}
\label{sec:appd-evalues}

As a preliminary, we define the e-BH procedure~\citep{wang2022false}. For random variables $e_1, \dots, e_m$, let $e_{(k)}$ denote the 
$k$th largest e-value and write
\begin{equation}
    \hat k \coloneqq \max \biggl\{ k \in [m] \colon e_{(k)} \ge \frac{m}{\alpha k} \biggr\},
\end{equation} 
where $\max\varnothing \coloneqq 0$. The e-BH procedure 
then rejects the $\hat k$ largest e-values:
\begin{equation}
    \hat R^\ebh (e_1, \dots, e_m)  \coloneqq \biggl\{ j \in [m]\colon e_j \ge \frac{m}{\alpha \hat k} \biggr\}.
\end{equation}
One can observe that $\hat{R}^\ebh(e_1, \dots, e_m) = \hat R^\bh\bigl((1/e_1) \wedge 1, \dots, (1/e_m) \wedge 1\bigr)$, with 
division by zero resulting in $\infty$ by convention.~\citet{wang2022false} proved that 
e-BH applied to strict or compound e-values controls the FDR at the nominal level, regardless of the dependence structure of the e-values.

\subsection{Proof of Proposition~\ref{prop:cbt-compound-evalues}}

We first prove the compound validity of $\tilde e_1, \dots, \tilde e_m$.
Fix $j \in \mathcal{H}_0$.  For $t \in \RR$, Define a leave-one-out version of the weighted control winner count $C(t)$ from the proof of Proposition~\ref{prop:cbt-bh-equivalence} by
\begin{align*}
    C^{(-j)}(t)=
    \sum_{\ell\ne j}\frac{1}{n_\ell}
    \sum_{k=1}^{n_\ell}
    \ind\{W_{\ell,k}=1,\ S_{\ell,k}\ge t\}.
\end{align*}
As a standard convention throughout this paper, we interpret
$0/0$ as zero for the ratios that follow.  On $\{W_{j,0}=1\}$, no control wins competition $j$, so
$C(T_j^+)=C^{(-j)}(T_j^+)$; on $\{W_{j,0}=0\}$, we have 
$\tilde e_j = 0$. Therefore,
\[
    \frac{\tilde e_j }{m}
    =
    \frac{
        W_{j,0}\ind\{S_{j,(0)}\ge T_j^+\}
    }{
        \frac{1}{n_j}+C^{(-j)}(T_j^+)
    }.
\]
Recalling the definition of the $\sigma$-algebra
$\cI_j$ from~\eqref{eq:individual-invariant-sigma-algebra}, both $T_j^+$ and $C^{(-j)}(T_j^+)$ are
$\cI_j$-measurable.  By~\eqref{eq:cbt-conditional-winner-balance} and the no ties assumption,
\begin{align*}
    \EE\Bigl[
        \frac{\tilde e_j }{m}
        \Biggiven \cI_j
    \Bigr]
    &=
    \EE\biggl[
        \frac{
            \frac1{n_j} (\sum_{k=1}^{n_j} W_{j,k})
            \ind\{S_{j,(0)}\ge T_j^+\}
        }{
            \frac{1}{n_j}+C^{(-j)}(T_j^+)
        }
        \Biggiven \cI_j
    \biggr]
    \\
    &=
    \EE\biggl[
        \frac{
            \frac{1}{n_j}\sum_{k=1}^{n_j}
            \ind\{W_{j,k}=1,\ S_{j,{(0)}}\ge T_j^+\}
        }{
            \frac{1}{n_j}+C^{(-j)}(T_j^+)
        }
        \Biggiven \cI_j
    \biggr]\\
    &=
    \EE\biggl[
        \frac{
            \frac{1}{n_j}\sum_{k=1}^{n_j}
            \ind\{W_{j,k}=1,\ S_{j,k}\ge T_j^+\}
        }{
            \frac{1}{n_j}+C^{(-j)}(T_j^+)
        }
        \Biggiven \cI_j
    \biggr].
\end{align*}

Whenever the numerator in the final ratio above is positive,
a control wins competition $j$, so $T_j^+=T$.
Moreover, in this case, the winning control score of competition $j$ must exceed this threshold, 
and hence $C(T) =  {1}/{n_j}+C^{(-j)}(T_j^+)$.  On the other hand, if the numerator is zero, then 
either the test unit won competition $j$ (meaning $W_{j,k} = 0$ for all $k\in[n_j]$);
or some control score $S_{j,k}$ won competition $j$ but $S_{j,k} < T_j^+$.
In the latter case, $T_j^+ = T$, so $\frac{1}{n_j}\sum_{k=1}^{n_j}
        \ind\{W_{j,k}=1,\ S_{j,k}\ge T\}=0$ as well. Consequently, 
\[
    \frac{
        \frac{1}{n_j}\sum_{k=1}^{n_j}
        \ind\{W_{j,k}=1,\ S_{j,k}\ge T_j^+\}
    }{
        \frac{1}{n_j}+C^{(-j)}(T_j^+)
    }
    =
    \frac{
        \frac{1}{n_j}\sum_{k=1}^{n_j}
        \ind\{W_{j,k}=1,\ S_{j,k}\ge T\}
    }{
        C(T)
    }.
\]
Taking expectations and summing over the null hypotheses
therefore yields
\begin{align*}
    \sum_{j\in\cH_0}\EE[\tilde e_j ]
    &=
    m\EE\biggl[
        \frac{ 
            \sum_{j\in\cH_0}\frac{1}{n_j}
            \sum_{k=1}^{n_j}
            \ind\{W_{j,k}=1,\ S_{j,k}\ge T\}
        }{C(T)}
    \biggr]\\
    &= m\EE\biggl[
        \frac{ 
            \sum_{j\in\cH_0}\frac{1}{n_j}
            \sum_{k=1}^{n_j}
            \ind\{W_{j,k}=1,\ S_{j,k}\ge T\}
        }{\sum_{j=1}^m \frac{1}{n_j}
            \sum_{k=1}^{n_j}
            \ind\{W_{j,k}=1,\ S_{j,k}\ge T\}}
    \biggr] \le m,
\end{align*}
so $\tilde{e}_1,\ldots,\tilde{e}_m$ are compound e-values.  

For the second claim, for any $j \in \hat R^\ebh$, we have 
$ \tilde e_j\ge
{m}/ \bigl({\alpha|\hat R^\ebh|}\bigr)
>0$. Thus $W_{j,0}=1$ 
and $ S_{j,0}\ge T_j^+\ge T$. 
Hence, by Proposition~\ref{prop:cbt-bh-equivalence}
\[
    \hat R^\ebh
    \subseteq
    \{j\in[m]:W_{j,0}=1,\ S_{j,0}\ge T\}
    =
    \hat R^\cbt = \hat R^\bh\bigl( \tilde p_1(1/\nmin), \dots, \tilde p_m(1/\nmin) \bigr),
\]
as required.
\hfill $\qed$

%
%
%
%
%
%

%
%
%
%
%
%
%
%
%
%
%
%
%
%
%
%
%
%
%
%
%
%
%
%
%
%
%
%
%
%
%
%
%
%
%
%
%

%
%
%
%
%
%
%
%
%
%
%
%
%
%
%
%
%
%
%
%
%
%
%
%
%
%
%
%
%
%
%
%
%
%
%
%
%
%
%
%
%
%
%

%
%
%
%
%
%
%
%
%
%
%

%
%
%
%
%
%
%
%
%
%
%

%
%
%
%
%
%
%
%
%
%
%
%
%
%
%
%
%
%
%
%
%
%
%

%
%
%
%
%
%
%
%
%
%
%
%
%
%
%
%
%
%
%

%
%
%
%
%
%
%
%
%
%
%
%
%
%
%
%
%
%
%
%
%
%
%
%
%
%
%
%
%
%
%
%
%
%
%
%
%
%
%
%
%
%
%

%
%
%
%
%
%
%
%
%
%
%
%
%

%
%
%
%
%
%
%
%
%
%
%
%
%
%
%
%
%
%
%
%
%
%
%
%
%
%
%
%
%
%
%
%
%
%
%
%
%
%
%
%
%
%
%
%
%
%
%
%
%
%

%
%
%
%
%
%
%
%
%
%
%
%
%
%
%
%
%
%
%
%
%
%
%
%
%
%
%
%
%

%
%
%
%
%
%
%
%
%
%
%
%
%
%

%
%
%
%
%
%
%
%
%
%
%
%
%

%% file: sections_appendix/appd_fdr_supermartingale.tex

\section{Proofs for supermartingale-based FDR bounds}
\label{sec:appd-fdr-supermartingale}
\label{sec:appd-deferred-fdr-bounds}

In this section, we prove the supermartingale-based FDR upper bounds 
presented in Propositions~\ref{prop:kr-inflated-fdr-control-common-offset},~\ref{prop:shared-cbt-ville-fdr-control},
and~\ref{prop:abstract-random-block-cbt-ville-fdr-control}. 
We first establish common notation in Appendix \ref{sec:supermartingale-notation}, then 
provide the common bound in Appendix~\ref{sec:supermartingale-integrated-fdr}, 
and finally specialize the result for the three settings corresponding to the 
aforementioned propositions in Appendix~\ref{sec:supermartingale-applications}.  

\subsection{Notation and reduction to the FDP}
\label{sec:supermartingale-notation}

Fix a common offset $c_0>0$ and a level $\alpha\in(0,1)$.  Similar in spirit to our previous notation, we make the following high-level definitions, with specifics to follow later: for $t\in \RR$, 
let $C(t)$ be the weighted number of control
winners, $R(t)$ the number of test winners, and $V(t)$ the number of null
test winners, all with scores at least $t$.
When dealing with random hypotheses, we write the (random) set of null hypothesis indices as 
$\cH_0\coloneqq\{j\in[m]:U_{j,0}\le c_{j,0}\}$, so that
$\ind\{j\in\cH_0\}=I_{j,0}$ as per the notation of the proof of
Proposition~\ref{prop:block-cbt-compound-rh}.
Let $W'_{j,0}$ denote the test winner indicator, specified below
for each setting, and define
\[
    \begin{aligned}
        R(t)&\coloneqq\sum_{j=1}^m\ind\{W'_{j,0}=1,S_{j,0}\ge t\},\\
        V(t)&\coloneqq\sum_{j=1}^m\ind\{j\in\cH_0\}
        \ind\{W'_{j,0}=1,S_{j,0}\ge t\};
    \end{aligned}
\]
we explicitly define by  convention  $R(+\infty) = V(+\infty) \coloneqq 0$. 

\paragraph{Test-specific controls.}
Using the winner indicators $W_{j,k}$ from Section~\ref{sec:cbt},  define
\[
    C(t)\coloneqq\sum_{j=1}^m\frac1{n_j}\sum_{k=1}^{n_j}
    \ind\{W_{j,k}=1,S_{j,k}\ge t\}.
\]
and take $W'_{j,0}\coloneqq W_{j,0}$.

\paragraph{Shared controls.}
For fixed hypotheses, use the winner indicators
$W_{j,k}^{\blk}$ in~\eqref{eq:block-winner}, define
\[
    C(t) \coloneqq \sum_{g=1}^G\sum_{j\in B_g}\frac1{n_g}\sum_{k=1}^{n_g}
    \ind\{W_{j,k}^{\blk}=1,S_{g,k}\ge t\}
\]
and take $W'_{j,0}=W_{j,0}^{\blk}$.  For random hypotheses, take $V_{j,k}^{\blk}$ from
\eqref{eq:block-winner-rh}, define
\[
    C(t) \coloneqq\sum_{g=1}^G\sum_{j\in B_g}\frac1{n_g}\sum_{k=1}^{n_g}
    \ind\{V_{j,k}^{\blk}=1,S_{g,k}\ge t\}
\]
and set $W'_{j,0} \coloneqq V_{j,0}^{\blk}$.  These control counts agree with $C^{\blk}(t)$ and $C^{\blk,\rh}(t)$
in the proofs of Theorem~\ref{thm:shared-cbt-compound} (Appendix~\ref{sec:appd-proof-shared-cbt-compound}) and  Proposition~\ref{prop:block-cbt-compound-rh} (Appendix~\ref{sec:appd-proof-block-cbt-compound-rh}), respectively. 

\paragraph{The FDP ratio.}
In each setting, and for $t \in \RR$, define
\begin{equation}
    \label{eq:supermartingale-fdp-ratio}
    \widehat{\fdp}_{c_0}(t) \coloneqq \frac{c_0+C(t)}{R(t)\vee1},
    \qquad
    Q_{c_0}(t) \coloneqq  \frac{V(t)}{c_0+C(t)} = \frac{\fdp\bigl( \{ j\in [m] \colon W_{j,0}' = 1, S_{j,0} \ge t\} \bigr)}{\hat{\fdp}_{c_0}(t)}.
\end{equation}
For shared controls, we use the default common-offset p-values in
Propositions~\ref{prop:shared-cbt-ville-fdr-control} and
\ref{prop:abstract-random-block-cbt-ville-fdr-control}, so we take $c_0 = \rmax$. 
By analogous reasoning to the proof of Proposition~\ref{prop:cbt-bh-equivalence} (Appendix~\ref{sec:appd-cbt-equivalence-proof}), in all three settings, 
\[
    \hat R^\bh = \{ j \in [m] \colon W'_{j,0} =1, S_{j,0} \ge t_\bh \},
\]
where $t_\bh$ is the score corresponding to the largest compound p-value rejected by BH (and $\infty$ if the rejection set is empty). 
Hence, $R(t_\bh)$ counts the number of rejections made by BH, while $V(t_\bh)$ counts the number of false rejections. Moreover, $c_0 + C(t_\bh) \le \alpha R(t_\bh)$ by similar reasoning to~\eqref{eq:cbt-pval-bh-thresh}.
Consequently, %
\begin{equation}
    \label{eq:supermartingale-pathwise-fdp}
    \fdp(\hat R^\bh)
    =\frac{V(t_{\bh})}{R(t_{\bh})\vee1}
    \le\alpha\frac{V(t_{\bh})}{c_0+C(t_{\bh})}
    \le\alpha\sup_{t\in\RR}Q_{c_0}(t).
\end{equation}
As a result, 
\begin{equation}
    \label{eq:supermartingale-bh-reduction}
    \begin{aligned}
        \fdr \le \alpha \,\EE\Bigl[
            \sup_{t\in\RR}Q_{c_0}(t)
        \Bigr].
    \end{aligned}
\end{equation} 

\subsection{A high-probability bound  on the FDP and the resulting FDR bound}
\label{sec:supermartingale-ratio-bound}
\label{sec:supermartingale-integrated-fdr}

To bound the expectation in~\eqref{eq:supermartingale-bh-reduction},
we first prove simultaneous high-probability bounds for $   \sup_{t\in\RR}Q_{c_0}(t)$.

\begin{lemma}[Based on \citet{katsevich2020simultaneous}]
    \label{lem:competition-supermartingale-tail}
    Let $(\cF_r)_{r\in \NN_0}$ be a filtration. For $r \in \NN$,  let $X_r\in\{0,1\}$ be
    $\cF_r$-measurable, and let $w_r$ be $\cF_{r-1}$-measurable.
    Define $\pi_r\coloneqq\PP(X_r=1\given \cF_{r-1})$ and 
    suppose that, for a deterministic $\bar w > 0$,
    \begin{equation}
        \label{eq:supermartingale-conditional-weights}
        \pi_r
        \le\frac{w_r}{1+w_r} \qquad \text{and} \qquad 0\le w_r\le\bar w.
    \end{equation} 
    Suppose also that  for every $t\in\RR$ there exists a
    nonnegative integer $r(t)$ such that
    \begin{equation}
        \label{eq:supermartingale-count-representation}
        V(t)=\sum_{\ell=1}^{r(t)}X_\ell,
        \qquad
        C(t)\ge\sum_{\ell=1}^{r(t)}w_\ell(1-X_\ell).
    \end{equation}
    Then, for every $c_0>0$ and $x>1$,
    \begin{equation}
        \label{eq:competition-supermartingale-tail}
        \PP\biggl(\sup_{t\in\RR}
        Q_{c_0}(t) \ge x \bigggiven \cF_0 \biggr)
        \le e^{-c_0x\theta_{\bar w}(x)},
    \end{equation}
    where $\theta_{\bar w}(x)$ is the unique positive solution in $\theta$ of
    \begin{equation}
        \label{eq:competition-supermartingale-root}
        e^{-\bar w x\theta}+\bar w e^\theta=1+\bar w.
    \end{equation}
\end{lemma}

\begin{proof}
    We apply~\citet[Lemma 1, Case 3]{katsevich2020simultaneous} under
    the conditional law given $\cF_0$.
    This lemma requires us to define a (potentially infinite) collection of 
    hypotheses, as well as an ordering of hypotheses and corresponding p-values and cutoffs. 
    Using the superscript $\mathrm{KR}$ to distinguish their proof objects
    from ours, 
    introduce one auxiliary hypothesis $H_r^{\mathrm{KR}} $ for each step $r$, with all hypotheses null, and define their p-values by
    \[
        p_r\coloneqq\begin{cases}
            \pi_r& \text{if }X_r=1,\\
            1& \text{if } X_r=0.
        \end{cases}
    \]
    Take the (auxiliary) null set, ordering, cutoffs, and functions to be
    \[\cH_0^{\mathrm{KR}} \coloneqq\NN,\qquad 
        \pi^{\mathrm{KR}}(r)\coloneqq r,
        \qquad \alpha_r^{\mathrm{KR}} \coloneqq\pi_r,
        \qquad h_r(p) \coloneqq w_r\ind\{p>\alpha_r^{\mathrm{KR}}\}, \quad p \in [0,1].
    \]
    By the definitions of $p_r, \alpha_r^{\mathrm{KR}}$, and $h_r$, 
    and since $\pi _r < 1$, we have 
    $\ind\{p_r\le\alpha_r^{\mathrm{KR}}\}=X_r$,
    and hence $ h_r(p_r)=w_r(1-X_r).$ 
    These quantities are $\cF_r$-measurable, 
    while $\pi^\mathrm{KR}(r)$ and $\cH_0^{\mathrm{KR}}$ are deterministic; the filtration requirement of
    \citet[Eqn.~(19)]{katsevich2020simultaneous} is therefore verified.
    The  functions $h_r$ are predictable with respect to the filtration $(\cF_r)_{r\ge 0}$, and the
    conditional inequalities of
    \citet[Eqn.~(20)]{katsevich2020simultaneous} hold because
    \[
        \PP(p_r\le\alpha_r^{\mathrm{KR}}\mid\cF_{r-1})
        =\pi_r=\alpha_r^{\mathrm{KR}},
        \qquad
        \EE[h_r(p_r)\mid\cF_{r-1}]
        =w_r(1-\pi_r)\ge\pi_r=\alpha_r^{\mathrm{KR}}.
    \]
    Moreover, $h_r(p)=0$ for $p\le\alpha_r^{\mathrm{KR}}$ and
    $0\le h_r(p)\le\bar w$, as required in Case~3 with $B=\bar w$.  \citet[Eqn.~(24)]{katsevich2020simultaneous} then  becomes~\eqref{eq:competition-supermartingale-root} after rearranging.

    The rejection sets are
    $\cR_r^{\mathrm{KR}} \coloneqq \{\ell\le r:p_\ell\le\alpha_\ell^{\mathrm{KR}}\}
    =\{\ell\le r:X_\ell=1\}$.
    With these choices and regularization $a=c_0$, the ratio appearing
    in \citet[Eqn.~(21)]{katsevich2020simultaneous} is 
    \[
        \frac{\sum_{\ell=1}^rX_\ell}
        {c_0+\sum_{\ell=1}^rw_\ell(1-X_\ell)},
    \]
    with the ratio interpreted as zero when the rejection set is empty.
    
    The conclusion of \citet[Lemma~1]{katsevich2020simultaneous} yields
    \begin{equation}
        \label{eq:supermartingale-sequential-bound}
        \PP\biggl(
        \sup_{r\ge0}
        \frac{\sum_{\ell=1}^rX_\ell}
        {c_0+\sum_{\ell=1}^rw_\ell(1-X_\ell)}\ge x
        \bigggiven \cF_0\biggr)
        \le e^{-c_0x\theta_{\bar w}(x)}.
    \end{equation}
    Finally, by~\eqref{eq:supermartingale-count-representation},
    \[
        \sup_{t\in\RR}Q_{c_0}(t) = \sup_{t\in\RR} \frac{V(t)}{c_0 + C(t)}
        \le 
        \sup_{t\in\RR}\frac{\sum_{\ell =1}^{r(t)} X_\ell  }{c_0  + \sum_{\ell =1}^{r(t)}  w_\ell (1- X_\ell)}
        \le
        \sup_{r\ge0}
        \frac{\sum_{\ell=1}^rX_\ell}
        {c_0+\sum_{\ell=1}^rw_\ell(1-X_\ell)},
    \]
    as required.
\end{proof}

We are now in a position to state and prove our main unifying result, in which bound the FDR by leveraging Lemma~\ref {lem:competition-supermartingale-tail}.
\begin{proposition}%
    \label{prop:supermartingale-integrated-fdr}
    Suppose that $\fdp\le\alpha\sup_{t\in\RR}Q_{c_0}(t)$ almost surely and that
    $C(t)$ and $V(t)$ satisfy the conditions of
    Lemma~\ref{lem:competition-supermartingale-tail}.
    Then
    \begin{align}
        \fdr
        &\le\alpha\biggl\{1+\int_1^\infty
        e^{-c_0x\theta_{\bar w}(x)}\,dx\biggr\}
        \label{eq:supermartingale-tail-integral}\\
        &=\alpha\frac{c_0}{\bar w^2}\int_0^1
        (1-z)^{c_0/\bar w-1}
        \frac{\log\bigl(1 / (1-z)\bigr)}{\log(1+z/\bar w)}\,dz,
        \label{eq:supermartingale-root-free-fdr}
    \end{align}
    where the integral is finite for all $\bar{w}$.  For fixed $\bar w$, the inflation factor in
    \eqref{eq:supermartingale-root-free-fdr} converges to one as
    $c_0\to\infty$.
\end{proposition}

\begin{proof}
    For the first upper bound~\eqref{eq:supermartingale-tail-integral}, we have by~\eqref{eq:supermartingale-bh-reduction} and Lemma~\ref{lem:competition-supermartingale-tail} that
    \[
        \begin{aligned}
            \fdr
             \le\alpha\,\EE\Bigl[\sup_{t\in\RR}Q_{c_0}(t)\Bigr] 
            &=\alpha\int_0^\infty
            \PP\biggl(\sup_{t\in\RR}Q_{c_0}(t)>x\biggr)\,dx\\
            &\le\alpha\biggl\{1+\int_1^\infty
            \PP\biggl(\sup_{t\in\RR}Q_{c_0}(t)\ge x\biggr)\,dx\biggr\} \le\alpha\biggl\{1+\int_1^\infty
            e^{-c_0x\theta_{\bar w}(x)} \,dx\biggr\}.
        \end{aligned}
    \] 
    For $x>1$, set $z \coloneqq \bar w(e^{\theta_{\bar w}(x)}-1) \in (0,1)$. By~\eqref{eq:competition-supermartingale-root} and our  choice of $z$, we have $e^{-c_0x\theta_{\bar w}(x)} = (1-z)^{c_0 / \bar w}$ and by inversion can solve for $x = x(z)$ where
    \[
        x(z) \coloneqq \frac{-\log(1-z)}{\bar w\log(1+z/\bar w)} = \frac{-\log(1-z) / z}{\bar w \log (1 + z/\bar w)  /z}. 
    \]
    Since $\lim_{z \searrow 0} x(z)=1$ and $\lim_{z\nearrow 1} x(z) = \infty$, and since $z \mapsto -\log(1-z)/z$ and $z \mapsto {\bar w \log (1 + z/\bar w)  /z}$ are strictly increasing and strictly decreasing functions respectively, it follows that $z \mapsto x(z)$ is a continuous, strictly increasing bijection from $(0,1)$ to $(1, \infty)$.

    Now, for $\varepsilon \in (0,1/2)$, we can integrate by parts to obtain
    \[
        \begin{aligned}
            \int_{x(\varepsilon)}^{x(1-\varepsilon)}
            e^{-c_0x\theta_{\bar w}(x)}\,dx
            &=
            \int_\varepsilon^{1-\varepsilon}
            (1-z)^{c_0/\bar w}x'(z)\,dz\\
            &=
            \varepsilon^{c_0/\bar w}x(1-\varepsilon)
            -(1-\varepsilon)^{c_0/\bar w}x(\varepsilon)
            +\frac{c_0}{\bar w}
            \int_\varepsilon^{1-\varepsilon}
            x(z)(1-z)^{c_0/\bar w-1}\,dz.
        \end{aligned}
    \]
    We have
    \[
        \varepsilon^{c_0/\bar w}x(1-\varepsilon)
        =
        \frac{\varepsilon^{c_0/\bar w}\log(1/\varepsilon)}
        {\bar w\log\{1+(1-\varepsilon)/\bar w\}}
        \to 0
    \]
    as $\varepsilon \searrow 0$, since $c_0 / \bar w > 0$; moreover $(1-\varepsilon)^{c_0 / \bar w} x(\varepsilon) \to 1$ as $\varepsilon \searrow 0$.  
    It follows by monotone convergence that 
    \begin{equation}
        \label{eq:mct-conclusion}
        1+\int_1^\infty
            e^{-c_0x\theta_{\bar w}(x)} \,dx 
        = \frac{c_0}{\bar w^2} \int_0^1 
        (1-z)^{c_0/\bar w-1}
        \frac{-\log(1-z)}{\log(1+z/\bar w)}\,dz,
    \end{equation} 
    which establishes~\eqref{eq:supermartingale-root-free-fdr}.
    
    The finiteness of this integral for all $\bar w$ is established as follows. Fixing $\bar w$, observe that since $\theta_{\bar{w}}(x) = \log\bigl(1 + z/w)$, and $x \mapsto z(x)$ is increasing, so $\theta_{\bar w}(x)\ge\theta_{\bar w}(2)>0$ for $x\ge2$. Thus 
    \begin{equation}
        \label{eq:integrability}
        \int_1^\infty e^{-c_0x\theta_{\bar w}(x)}\,dx
             \le
            1+\int_2^\infty e^{-c_0x\theta_{\bar w}(2)}\,dx 
             =
            1+\frac{e^{-2c_0\theta_{\bar w}(2)}}
            {c_0\theta_{\bar w}(2)}
            <\infty.
    \end{equation} 
    Taking $c_0=1$ also shows that
    $x\mapsto e^{-x\theta_{\bar w}(x)}$ is integrable on $(1,\infty)$.
    For every $c_0\ge1$, 
    $0\le e^{-c_0x\theta_{\bar w}(x)} \le e^{-x\theta_{\bar w}(x)} $.  Hence by dominated convergence, 
    \[
        \int_1^\infty e^{-c_0x\theta_{\bar w}(x)}\,dx \rightarrow 0
    \]
    as $c_0 \rightarrow \infty$, as required.
    \end{proof}

\begin{remark} [An explicit upper bound]
\label{rem:an-explicit-upper-bound}
As the solution $\theta_{\bar w}(x)$ to~\eqref{eq:competition-supermartingale-root} is not of closed form, we may consider loosening the FDR upper bound slightly in order to obtain a closed-form integrand. 
For $x>1$, define
\[
    \lambda_x \coloneqq \log\biggl(1+
    \frac{1-1/x}{\bar w\vee(1/x)}\biggr).
\]
We claim that $\lambda_x \le \theta_{\bar w}(x)$.  To see this, observe first that the function $z \mapsto e^{-\bar w x \cdot z}  + \bar w e^z  - (1 + \bar w)$ is strictly convex on $[0,\infty)$ and evaluates to zero at $z =0$ and $z = \theta_{\bar w} (x)$.   If $\bar w\ge1/x$, then
\[
    e^{\bar w x\lambda_x}
    =\biggl(1+\frac{1-1/x}{\bar w}\biggr)^{\bar w x}\ge 1 +  \bar w x \frac{x-1}{\bar w 
    x}  = x,
\]
so
\[
e^{-\bar w x \cdot \lambda_x}  + \bar w e^{\lambda_x}  - (1 + \bar w) \leq \frac{1}{x} + \bar w e^{\lambda_x}  - (1 + \bar w) = 0.
\]
Hence $\lambda_x \leq \theta_{\bar{w}}(x)$ in this case.   
On the other hand, if $\bar w<1/x$, then $\lambda_x=\log x$.
The function $s\mapsto e^{-xs\log x}+sx-s-1$ is convex on $[0,1/x]$
and is 0 at both endpoints, so again 
\begin{equation}
    \label{eq:conclusion-of-root-equation}
    e^{-\bar wx \cdot \lambda_x}+\bar w e^{\lambda_x} - (1+\bar w) \leq 0.
\end{equation}
This establishes our claim, and this bound will be used later in Appendix~\ref{sec:supermartingale-global-null-application}.
\end{remark}

\subsection{Proofs of the FDR bounds}
\label{sec:supermartingale-applications}

In each application, we define the filtration and then verify the two
conditions of Lemma~\ref{lem:competition-supermartingale-tail}:
the conditional probability condition
\eqref{eq:supermartingale-conditional-weights} and the conditions on
$C(t)$ and $V(t)$ in
\eqref{eq:supermartingale-count-representation}.
We then invoke Proposition~\ref{prop:supermartingale-integrated-fdr}
with the appropriate substitutions.  %

\subsubsection{Proof of Proposition~\ref{prop:kr-inflated-fdr-control-common-offset}}
\label{sec:supermartingale-nonshared-application}

    \paragraph{Defining the filtration.}
    For $j\in[m]$, define
    \[
        S_{j,(0)}\coloneqq\max_{0\le k\le n_j}S_{j,k},
        \qquad
        \cG
        \coloneqq
        \sigma\bigl(
        (S_{j,(0)})_{j\in[m]},
        (W_{j,0})_{j\notin\cH_0}
        \bigr).
    \] 
    Then $\cG$ contains every competition's winning score and the test winner indicators of each nonnull competition, but it leaves the null
    winner indicators $(W_{j,0})_{j\in\cH_0}$ unrevealed.  Write
    $m_0 \coloneqq |\cH_0|$.  If $m_0=0$, then the FDR is zero and the claim is
    immediate.  Hence, suppose that $m_0 \in [m]$.

    List the null indices $j_1,\ldots,j_{m_0}$ in decreasing order
    of $S_{j,(0)}$, breaking ties by index.  This ordering is random but
    $\cG$-measurable.  For $r\in[m_0]$, take 
    $X_r=W_{j_r,0}$ and $w_r=1/{n_{j_r}}$, 
    and set 
    \[
        \cF_0=\cG,
        \qquad
        \cF_r=\sigma( \cG, X_1,\ldots,X_r).
    \]
    Then $S_{j_r,(0)}$ and $w_r$ are $\cF_0$-measurable, and $X_r$ records
    whether or not the test score wins the $r$th null competition.

    \paragraph{Verifying the conditional probabilities in~\eqref{eq:supermartingale-conditional-weights}.}
    The argument establishing
    \eqref{eq:cbt-log-conditional-label-law} does not require distinct
    winning scores between competitions.  Conditional on $\cF_0$, the
    null winner indicators are independent and satisfy
    \[
        W_{j,0}\given \cF_0
        \sim\Bern\biggl(\frac1{n_j+1}\biggr),
        \qquad j\in\cH_0.
    \]
    Thus, for $r\in[m_0]$,
    \[
        \pi_r
        =\PP(X_r=1\given \cF_{r-1})
        =\frac1{n_{j_r}+1}
        =\frac{w_r}{1+w_r},
        \qquad
        0<w_r\le\frac1{\nmin}.
    \]
    This verifies~\eqref{eq:supermartingale-conditional-weights} with
    $\bar w=1/\nmin$.

    \paragraph{Verifying the conditions on $C(t)$ and $V(t)$ in~\eqref{eq:supermartingale-count-representation}.}
    For $t\in\RR$, let
    \[
        r(t)=\sum_{\ell=1}^{m_0}\ind\{S_{j_\ell,(0)}\ge t\}.
    \]
    Each competition has a unique winner by assumption, from which it follows that
    \begin{align*}
        \sum_{\ell=1}^{r(t)}X_\ell
        &=\sum_{j\in\cH_0}W_{j,0}\ind\{S_{j,(0)}\ge t\}
        =V(t),\\
        \sum_{\ell=1}^{r(t)}w_\ell(1-X_\ell)
        &=\sum_{j\in\cH_0}\frac1{n_j}
        \sum_{k=1}^{n_j}\ind\{W_{j,k}=1,S_{j,k}\ge t\}
        \le C(t).
    \end{align*}
    This verifies~\eqref{eq:supermartingale-count-representation}.
    
    Taking $\bar w=1/\nmin$ in
    Proposition~\ref{prop:supermartingale-integrated-fdr} yields
    \[
        \fdr\le\alpha c_0\nmin^2\int_0^1
        (1-z)^{c_0\nmin-1}
        \frac{-\log(1-z)}{\log(1+\nmin z)}\,dz.
    \]
    Setting $c_0=1/\nmin$ gives the stated expression when the default offset is chosen.
    The limit as $c_0\to\infty$ with $\nmin$ fixed follows from the
    same proposition.\hfill $\qed$

\subsubsection{Proof of Proposition~\ref{prop:shared-cbt-ville-fdr-control}
and Corollary~\ref{coro:limsup-fdr-pi2alpha}}
\label{sec:supermartingale-fixed-shared-application}

    \paragraph{Defining the filtration.}
    Write the number of true null hypotheses in block $g \in [G]$ as $h_g\coloneqq  |B_g\cap\cH_0|$ and define the collection of its null test unit indices as $\cT_g\coloneqq \bigl\{(j,0):j\in B_g\cap\cH_0\bigr\}$.
    Further write the total collection of indices in block $g$ as $\cL_g
    \coloneqq
    \cT_g\cup\bigl\{(g,k):k\in[n_g]\bigr\}$, with cardinality $d_g\coloneqq|\cL_g|=h_g+n_g$. 
    Also define the unordered multiset of scores corresponding to $\cL_g$ as 
    \[
        \cM_g
        \coloneqq
        \Lbag S_{j,0}:j\in B_g\cap\cH_0\Rbag
        \cup 
        \Lbag S_{g,k}:k\in[n_g]\Rbag.
    \]
    Define the initial $\sigma$-algebra in our filtration as 
    \[
        \cF_0
        \coloneqq
        \sigma\left(
        (S_{j,0})_{j\in[m]\setminus\cH_0},
        (\cM_g)_{g\in[G]}
        \right).
    \]
    The no-ties condition in
    Theorem~\ref{thm:shared-cbt-compound}, combined with block exchangeability, 
    precludes ties in $\cM_g$ for each $g\in[G]$. Explicitly, consider distinct $j,\ell\in B_g\cap\cH_0$ and any $k\in[n_g]$;
    block exchangeability then yields 
        $\PP(S_{j,0}=S_{\ell,0}\mid\cG_{-g}^{\blk})
        =\PP(S_{j,0}=S_{g,k}\mid\cG_{-g}^{\blk})=0$ 
    almost surely.

    Write the $d_g$ distinct scores in $\cM_g$  in decreasing order as
    $Z_{g,1}>\cdots>Z_{g,d_g}$. 
    For $g\in [G]$,
    define the unique bijection $\Pi_g:[d_g]\to\cL_g$ such that $Z_{g,q} = S_{\Pi_g(q)}$
    for all $q \in [d_g]$.  
    Further denote $\operatorname{Sym}(\cL_g)$ as the set of bijections from
    $[d_g]$ to $\cL_g$.
    By Assumption~\ref{assu:block-specific}, 
    $\Pi_g \given \cF_0 \sim \unif\bigl( \operatorname{Sym}(\cL_g) \bigr)$ for each $g$, with   $(\Pi_g)_{g\in [G]}$ conditionally independent given $\cF_0$. 
    Enumerate the $D \coloneqq \sum_{g=1}^Gd_g$  positions
    $(g,q)$  as  $(g_1,q_1),\ldots,(g_D,q_D)$, so that
    $Z_{g_1,q_1}\ge\cdots\ge Z_{g_D,q_D}$, breaking ties by block index.  
    This ordered score sequence, spanning all blocks,  is $\cF_0$-measurable.

    Define the filtration $(\cF_r)_{r\in \NN_0}$ by
    \[
        \cF_r \coloneqq \sigma\Bigl( \cF_0, \bigl( \Pi_{g_\ell}(q_\ell)\bigr)_{\ell=1}^r   \Bigr), \qquad r\in [D],
    \]
    with $\cF_{r'} \coloneqq \cF_D$ for all $r' > D$. Define $A_r\coloneqq  \ind \{ \Pi_{g_r}(a) \in \cT_{g_r} \text{ for all } a \in [q_r - 1] \}$, with $A_r \coloneqq 1$ by convention when $q_r=1$. 
    By this definition, $A_r = 1$ if and only if  $\{Z_{g_r, a} \colon {a \le q_r-1 }\}$ corresponds precisely to scores that belong to true null test units in block $g_r$. 

    For $r \in [D]$, denote 
    \[
        s_r \coloneqq A_r \sum_{a  =q_r}^{d_{g_r}} \ind\{ \Pi_{g_r}(a) \in \cT_{g_r} \}= A_r (h_{g_r}  -q_r+1 ) 
    \]
    as the number of null test units among $(Z_{g_r, a})_{a\ge q_r}$ multiplied by the indicator variable $A_r$. The simplification of the definition is due to block $g_r$ containing $h_{g_r}$ null tests in total. Moreover, define $w_r \coloneqq s_r / n_{g_r}$ and $X_r\coloneqq A_r \ind\{ \Pi_{g_r}(q_r) \in \cT_{g_r} \}$.  Then $A_r$, $s_r$, and $w_r$ are all $\cF_{r-1}$-measurable, while $X_r$ is $\cF_{r}$-measurable. Furthermore, the event $\{\Pi_{g_r} (q_r) \notin \cT_{g_r}\} \in \cF_r$ specifies that $Z_{g_r, q_r}$ is a control score; on this event, for every $\ell > r$ such that $g_\ell = g_r$, we have $A_\ell = s_\ell  = w_\ell = X_\ell = 0$.  Finally, as the construction above is only well-defined for $r \in [D]$, we define  $X_r  =w_r \coloneqq 0$ for $r > D$ in order to enable the use of Lemma~\ref{lem:competition-supermartingale-tail}.

    \paragraph{Verifying the conditional probabilities in~\eqref{eq:supermartingale-conditional-weights}.}
    
    On $\{A_r = 1\}$, the indices assigned to positions $1, \dots, q_r-1$ in block $g_r$ are null test indices (i.e., in $\cT_{g_r}$). Hence, $s_r$ null test indices and $n_{g_r}$ control indices are available for position $q_r$, each equally likely conditional on $\cF_{r-1}$. Hence, 
    \[
        \pi_r = \frac{s_r}{s_r+n_{g_r}} = \frac{w_r}{1 + w_r}.
    \] 
    On the other hand, on $\{A_r = 0\}$, we have $X_r = s_r = w_r = 0$, so $\pi_r  \le w_r / (1+ w_r)$ also holds.
    The choice of $\bar w = \max_{g \in [G]} b_g / n_g = \rmax$ suffices for~\eqref{eq:supermartingale-conditional-weights}, as
    $w_r = s_r / n_{g_r} \le h_{g_r} / n_{g_r} \le b_{g_r} / n_{g_r} \le \rmax$.

    \paragraph{Verifying the conditions on $C(t)$ and $V(t)$ in~\eqref{eq:supermartingale-count-representation}.}
    For $g\in [G]$, define $S_g^{\max} \coloneqq \max _{k \in [n_g]} S_{g,k}$ as the
    maximum control score in block $g$. Also define 
    \[
        \kappa_g \coloneqq \min\{ r \in [D] \colon g_r = g, \Pi_g(q_r)\notin \cT_g\};
    \]
    the $\kappa_g$th  score in $Z_{g_1,q_1}\ge \cdots \ge Z_{g_D, q_D}$ then corresponds to
    the control unit in block $g$ that attains $S_g^{\max}$. Consequently, for every $r\in[D]$ such that $g_r = g$, we have $X_r = \ind\{ r < \kappa _g \} = \ind\{ Z_{g, q_r} > S_g^{\max} \}$ and $w_r(1-X_r) = s_{\kappa_g}/n_g \cdot \ind\{r = \kappa _g \}$. 

    For $t\in \RR$, let
    \[
        r(t) \coloneqq \sum_{\ell = 1}^D \ind\{ Z_{g_\ell, q_\ell} \ge t \}. 
    \]
    Then $Z_{g_1,q_1}\ge \cdots \geq Z_{g_{r(t)},q_{r(t)}} \geq t \geq Z_{g_{r(t)+1},q_{r(t)+1}}  \ge Z_{g_D, q_D}$. Therefore,
    \begin{align*}
        \sum_{\ell=1}^{r(t)}X_\ell = \sum_{\ell=1}^D X_\ell \ind\{ Z_{g_\ell, q_\ell} \ge t \} &=   \sum_{g=1}^G \sum_{q=1}^{d_g} \ind\{ Z_{g,q} > S_g^{\max}, Z_{g,q} \ge t \}\\
        &=   \sum_{g=1}^G \sum_{j \in B_g \cap \cH_0} \ind\{ S_{j,0} > S_g^{\max}, S_{j,0} \ge t \} = V(t).
    \end{align*}
    On the other hand, 
    \begin{align*}
        \sum_{\ell=1}^{r(t)} w_\ell (1-X_\ell) &= \sum_{\ell=1}^ D w_\ell (1 - X_\ell)  \ind\{ Z_{g_\ell, q_\ell} \ge t \} = \sum_{g=1}^G \frac{s_{\kappa _g}}{n_g}  \sum_{q=1}^{d_g}\ind\{ Z_{g,q} = S_{g}^{\max}, Z_{g, q} \ge t  \} \\
        &=  \sum_{g=1}^G \frac{s_{\kappa _g}}{n_g}  \ind\{  S_{g}^{\max} \ge t  \} = \sum_{g=1}^G \sum_{j \in B_g \cap \cH_0} \frac{1}{n_g} \ind\{ S_{j,0} < S_g^{\max}, S_g^{\max} \ge t \} \le C(t),
    \end{align*}
    where the third equality is by the distinctness of all scores in $\cM_g$.   
    This verifies~\eqref{eq:supermartingale-count-representation}.
    Substituting $c_0=\bar w=\rmax$ into
    Proposition~\ref{prop:supermartingale-integrated-fdr} gives
    \begin{equation}
        \label{eq:supermartingale-shared-fdr}
        \fdr\le\alpha\int_0^1
        \frac{-\log(1-z)}{\rmax\log(1+z/\rmax)}\,dz.
    \end{equation}
    Finally, 
    compound validity of  $\tilde p_1^{\blk}(\rmax),\dots,\tilde p_m^{\blk}(\rmax)$  follows from
    Theorem~\ref{thm:shared-cbt-compound}, since increasing each
    offset $b_g/n_g$ to $\rmax$ can only increase the compound p-values.
    
    For Corollary~\ref{coro:limsup-fdr-pi2alpha}, the integrand in
    \eqref{eq:supermartingale-shared-fdr} converges to $-\log(1-z)/z$
    as $\rmax\to\infty$.
    For $\rmax\ge1$, $\log(1+y)\ge y/(1+y)$ for $y > -1$ gives
    \[
        \rmax\log(1+z/\rmax)
        \ge\frac{z}{1+z/\rmax}\ge\frac{z}{2}
    \]
    for all $z \in [0,1]$.  The dominated convergence theorem therefore allows us to conclude that 
    \begin{equation}
        \label{eq:supermartingale-shared-limit}
        \lim_{\rmax\to\infty}\int_0^1
        \frac{-\log(1-z)}{\rmax\log(1+z/\rmax)}\,dz
        =\int_0^1\frac{-\log(1-z)}z\,dz
        =\sum_{k=1}^\infty\frac1{k^2}
        =\frac{\pi^2}{6},
    \end{equation}
    as desired.\hfill$\qed$

\subsubsection{Proof of Proposition~\ref{prop:abstract-random-block-cbt-ville-fdr-control}}
\label{sec:supermartingale-random-shared-application}

    First,  the random-compound validity of $\tilde p_1^{\blk,\rh}(\rmax),\dots,\tilde p_m^{\blk,\rh}(\rmax)$ follows from Proposition~\ref{prop:block-cbt-compound-rh}. We now prove the FDR upper bound. 

    \paragraph{Defining the filtration.}
    We adapt the argument in Section~\ref{sec:supermartingale-fixed-shared-application} to account for the exchangeability of score--null-status pairs. For $g \in [G]$, replace $\cT_g$ and $\cL_g$ with
    \begin{equation} 
        \cT_g^{\rh}
        \coloneqq \{(j,0):j\in B_g\} \qquad \text{and}
        \qquad
        \cL_g^{\rh} 
        \coloneqq
        \cT_g^{\rh}
        \cup\{(g,k):k\in[n_g]\},
    \end{equation}  
    respectively. Hence, $d_g = b_g + n_g$ here, with $b_g = |B_g|$, as defined in Section~\ref{sec:shared_cbt}. We then define the initial $\sigma$-algebra in our filtration as 
    \begin{equation}
        \cF_0  
        \coloneqq
        \sigma \biggl(
        \Bigl(
        \Lbag (S_u,I_u):u\in\mathcal L_g^{\rh}\Rbag
        \Bigr)_{g\in[G]}
        \biggr). 
    \end{equation}  
    Denote the ordered scores $Z_{g, 1} > \cdots > Z_{g, d_g}$ and the bijection $\Pi_g$ as in Section~\ref{sec:supermartingale-fixed-shared-application}, with $\cL^\rh_g$ in place of $\cL_g$. Further define $J_{g,q} \coloneqq I _{\Pi_g(q)}$ as the null-status indicator associated with $Z_{g,q}$, for $q \in [d_g]$. 
    By Assumption~\ref{assu:block-specific-rh}, we can conclude---with analogous reasoning to the fixed-hypothesis setting--- that $\Pi_g \given \cF_0 \sim \unif\bigl(\operatorname{Sym}(\cL_g^\rh)\bigr)$ and $(\Pi_g)_{g\in [G]}$ are conditionally independent given $\cF_0$. 

    Define the total number of scores with true null status as $D \coloneqq \sum_{g=1}^G \sum_{q=1}^{d_g} J_{g,q}$, which is $\cF_0$-measurable, and enumerate the positions of the scores with true null status with the index pairs $(g_1, q_1), \dots, (g_D, q_D)$, such that 
    \[
        Z_{g_1, q_1} \ge \cdots \ge Z_{g_D, q_D},
    \] 
    again breaking ties by block index. This ordered sequence of true null scores is $\cF_0$-measurable. Now define the filtration
    \[
        \cF_r \coloneqq    \sigma\Bigl( \cF_0, \bigl( \Pi_{g_\ell}(q_\ell)\bigr)_{\ell=1}^r   \Bigr), \qquad r\in [D],
    \]
    with $\cF_{r'}   \coloneqq  \cF_D$ for all $r' > D$. For $r \in [D]$, we redefine $A_r$ and $s_r$ as follows:
    \[
        \begin{aligned}
        A_r
        &\coloneqq
        \ind \biggl\{
        \Pi_{g_r}(a)\in\cT_{g_r}^{\rh} 
        \text{ for every }a<q_r\text{ with }J_{g_r,a}=1
        \biggr\},\\
        s_r
        &\coloneqq
        A_r\biggl(
        b_{g_r}-\sum_{a=1}^{q_r-1}J_{g_r,a}
        \biggr),
        \end{aligned}
        \]
    with $A_r =1$ by convention whenever  
    $\{a\in[q_r-1]:J_{g_r,a}=1\}=\varnothing$. 
    Through our new $A_r$ and $s_r$, we can define $w_r \coloneqq s_r / n_{g_r}$ and $X_r \coloneqq A_r \ind\{ \Pi_{g_r}(q_r) \in \cT_{g_r}^\rh \}$, as in Section~\ref{sec:supermartingale-fixed-shared-application}.  The random variables $A_r, s_r, w_r$ are $\cF_{r-1}$ measurable, while $X_r$ is $\cF_r$-measurable. For $r > D$, set $X_r  =w_r = 0$.
    
    \paragraph{Verifying the conditional probabilities in~\eqref{eq:supermartingale-conditional-weights}.} 
    On $\{A_r= 1\}$, the indices available for position $q_r$ are the $s_r$ remaining test indices and all $n_{g_r}$ controls. Hence, 
    $\pi_r = s_r / (s_r +  n _{g_r})  =w_r/(1+w_r)$. On $\{A_r = 0\}$, we have $\pi_r = w_r / (1+w_r)$ as both sides are zero. Hence,~\eqref{eq:supermartingale-conditional-weights} is verified with the choice $\bar w = \rmax$. 

    \paragraph{Verifying the conditions on $C(t)$ and $V(t)$ in~\eqref{eq:supermartingale-count-representation}.}

    For each $g$ with $\cN_g \neq \varnothing$, we define $\kappa_g$ as in Section~\ref{sec:supermartingale-fixed-shared-application}, with $\cT_g^\rh$ replacing $\cT_g$. In addition, recall that $S_g^{\cN, *}$ denotes the maximum score of the controls in group $g$, restricted to those with null status 1; we use this in place of the fixed-hypothesis~$S_g^{\max}$. 
    Consequently, for every $r\in[D]$ such that $g_r = g$, we have $X_r = \ind\{ r < \kappa _g \} = \ind\{ Z_{g, q_r} > S_g^{\cN, *} \}$ and $w_r(1-X_r) = s_{\kappa_g}/n_g \cdot \ind\{r = \kappa _g \}$. On the other hand, consider $g$ such that $\cN_g = \varnothing$.
    Then, for every $r \in [D]$ with $g_r = g$, the score~$Z_{g_r,q_r}$ must belong to a test unit. Hence, $A_r = X_r = 1$ and $w_r(1-X_r) =0$.
    
    We take the same definition of $r(t)$ and observe from the same reasoning that
    \[ 
        \sum_{\ell=1}^{r(t)}X_\ell
         =
        \sum_{g=1}^G\sum_{q=1}^{d_g}
        J_{g,q}\,
        \ind\{Z_{g,q}>S_g^{\cN,*},\ Z_{g,q}\ge t\} 
         =
        \sum_{g=1}^G\sum_{j\in B_g}
        \ind\{I_{j,0}=1,\ V_{j,0}^{\blk}=1,\ S_{j,0}\ge t\} =V(t). 
    \]
    Furthermore, 
    \[
        \sum_{\ell = 1}^{r(t)} w_\ell (1-X_\ell) = \sum_{\substack{g\in[G]:\\ \cN_g\neq\varnothing}}
        \frac{s_{\kappa_g}}{n_g}
        \ind\{ S_g^{\cN, *} \ge t \} \le \sum_{\substack{g\in[G]:\\ \cN_g\neq\varnothing}}
        \frac{b_g }{n_g}
        \ind\{ S_g^{\cN, *} \ge t \} = C(t). 
    \]
    At this point, we substitute $c_0 = \bar w = \rmax$ into Proposition~\ref{prop:supermartingale-integrated-fdr} and conclude exactly as written in Section~\ref{sec:supermartingale-fixed-shared-application}. \hfill $\qed $

\subsubsection{FDR under the global null}
\label{sec:supermartingale-global-null-application}

Under the global null,  $\fdp = 1$ exactly when the rejection set
is nonempty. We therefore use the tail bound at a single value of $x$,
rather than integrating over it.

\begin{proposition}[FDR under the global null]
    \label{prop:supermartingale-global-null}
    Suppose that the conditions of Proposition~\ref{prop:cbt-bh-equivalence}
    hold and that $\cH_0=[m]$. With $c_0=1/\nmin$ and
    $\alpha\in(0,1)$,
    \begin{align}
        &\fdr\bigl(\hat R^\bh(
        \tilde p_1(1/\nmin),\ldots,\tilde p_m(1/\nmin))\bigr)
        \notag\\
        &\qquad\le
        \left(1+\frac{1-\alpha}{(1/\nmin)\vee\alpha}
        \right)^{-1/(\alpha\nmin)}
        =\begin{cases}
            \{1+\nmin(1-\alpha)\}^{-1/(\alpha\nmin)}
                &\text{for } \nmin<1/\alpha,\\
            \alpha^{1/(\alpha\nmin)} &\text{for }\nmin\ge1/\alpha.
        \end{cases}
        \label{eq:supermartingale-global-null}
    \end{align}
    In particular, the FDR is at most $\alpha$ whenever $\nmin\le1/\alpha$.
\end{proposition}

\begin{proof}
    On the event $\{\hat{R}^{\bh} \neq \varnothing\}$, we have 
    $\sup_{t\in\RR}Q_{c_0}(t)\ge1/\alpha$ by~\eqref{eq:supermartingale-pathwise-fdp}.  Hence by Lemma~\ref{lem:competition-supermartingale-tail} and Remark~\ref{rem:an-explicit-upper-bound} with $c_0 = \bar w = 1 /\nmin$,
    \[ 
            \fdr=\PP\bigl(\hat{R}^{\bh} \neq \varnothing\bigr) \le\PP\biggl(
            \sup_{t\in\RR}Q_{c_0}(t)\ge\frac1\alpha
            \biggr) \leq e^{-c_0 x \theta_{\bar w}(x)} \le e^{-c_0 x \lambda_x} \le 
            \biggl(1+\frac{1-\alpha}{(1/\nmin)\vee\alpha}
            \biggr)^{-1/(\alpha\nmin)}.
    \]
    The different cases in~\eqref{eq:supermartingale-global-null} follow immediately.
    Finally, the function $s\mapsto\log\bigr(1+s(1-\alpha)\bigl)$ is concave, so
    \[
        \begin{aligned}
            \log\bigl(1+\nmin(1-\alpha)\bigr)
            &=
            \log\biggl(
            (1-\alpha\nmin)\cdot1
            +\alpha\nmin\cdot\frac1\alpha
            \biggr) \ge 
            \alpha\nmin\log(1/\alpha).
        \end{aligned}
    \] 
    Multiplying by $-1/(\alpha\nmin)$ and exponentiating shows that
    the first case in~\eqref{eq:supermartingale-global-null} is at most
    $\alpha$. Evaluating the second case at $\nmin = 1/\alpha$ establishes the nominal FDR control for $\nmin \le 1/\alpha$. 
\end{proof}

%% file: sections_appendix/appd_equal_weights.tex

\section{Equal-weights compound p-values for competition-based testing}
\label{sec:appd-equal-weights}

\subsection{Equal weighting with an arbitrary common weight}

We state a general version of Proposition~\ref{prop:cbt-pval-equal-weights}
that holds when every contribution per competition in the numerator 
of the compound p-values is equally weighted by some $w_0 > 0$. 
Proposition~\ref{prop:cbt-pval-equal-weights} then follows 
by choosing $w_ 0 = 1/\nmin$.\footnote{
As Proposition~\ref{prop:cbt-pval-equal-weights-general} only proves the aggregate superuniformity property for $t \in [0,1)$, we must technically argue this property also holds when $t  =1$; this is straightforward, as   the sum of $|\cH_0|$ probabilities is always at most $m$. 
}

\begin{proposition}
    \label{prop:cbt-pval-equal-weights-general}
    Suppose that the conditions of Theorem~\ref{thm:cbt-compound} hold.
    For a deterministic $w_0>0$, define
    \begin{equation}
        \label{eq:cbt-pval-equal-weights-general}
        \tilde p_j^{(w_0)}
        \coloneqq
        \left[
        \frac{
            w_0
            +w_0\sum_{\ell=1}^m\sum_{k=1}^{n_\ell}
            \ind\{W_{\ell,k}=1,\ S_{\ell,k}\ge S_{j,0}\}
            +m\ind\{W_{j,0}=0\}
        }{m}
        \right]\wedge 1,
        \qquad j\in[m].
    \end{equation}
    Then, for every $t\in[0,1)$,
    \begin{equation}
        \sum_{j\in\cH_0}
        \PP\bigl(\tilde p_j^{(w_0)}\le t\bigr)
        \le
        \frac{m}{\nmin w_0} t. 
    \end{equation}
    Furthermore, assume no ties among the competition-winning scores. Then, for every $\alpha\in(0,1)$,
    \begin{equation}
        \fdr\left(
        \hat R^\bh\bigl(
        \tilde p_1^{(w_0)},\dots,\tilde p_m^{(w_0)}
        \bigr)
        \right)
        \le
         \frac{\alpha}{\nmin w_0}.
    \end{equation}
\end{proposition}

\begin{proof}
    For the first claim, fix $t\in [0,1)$. On the event $ \{ \tilde p_j^{(w_0)}\le t\}$, competition $j$ has a test winner, so $\tilde p_j (1 / n_j) \le \tilde p_j^{(w_0)} / (\nmin w_0)$. By Theorem~\ref{thm:cbt-compound}, 
    \[ 
        \sum_{j\in\cH_0}
        \PP\bigl(\tilde p_j^{(w_0)}\le t\bigr)
        \le
        \sum_{j\in\cH_0}
        \PP\biggl(
        \tilde p_j(1/n_j)\le
        \frac{t}{n_{\min}w_0}\wedge1
        \biggr) \le m\biggl(\frac{t}{n_{\min}w_0}\wedge1\biggr)
        \le \frac{mt}{n_{\min}w_0}. 
    \]

    To prove the FDR bound, we again use the notation of Appendix~\ref{sec:appd-proof-inflated-fdr-control}; we take $R(s)$ and $V(s)$ as in~\eqref{eq:cbt-rank-fdp-process}, but redefine $\hat V(s)$ as 
        \[
            \hat V(s) \coloneqq w_0 + w_ 0  \sum_{j=1}^m (1 - W_{j,0}) \ind\{ \rho _j \le s\}.
        \]
    The quantity $\tau$ is then defined as in~\eqref{eq:tau} but with this new $\hat V(s)$. By analogous reasoning to the proof of Lemma~\ref{lem:duality-of-cbt}, we conclude that $\hat R^\bh \coloneqq R^\bh\bigl(   \tilde p_1^{(w_0)},\dots,\tilde p_m^{(w_0)}   \bigr) = \{j \in [m]\colon W_{j,0} = 1, \rho _ j\le \tau\}$ and $\fdp(\hat R^\bh) \le \alpha V(\tau) / \hat V(\tau)$. 
    Moreover, for all $j \in [m]$, we take the $j$th-perturbed version of $R(s)$ as in the proof of Proposition~\ref{prop:cbt-counting-fdr-bound} and define the $j$th perturbed $V(s)$ as 
    \[
        \hat V^{j,-}(s)
        \coloneqq
        w _ 0 
        +\sum_{\ell \in [m] \setminus\{ j\}}
        w_ 0 (1-W_{\ell,0})\ind\{\rho_\ell\le s\};
    \]
    we can then define $\tau_j^{-}$ via~\eqref{eq:tau-perturbed} but with our new definition of $\hat V^{j, -}(\cdot)$.  Moreover, we again define  $\hat V_{j, -} \coloneqq \hat V^{j,-} (\tau_j^{-})$. 

    Following the start of the proof of Proposition~\ref{prop:cbt-counting-fdr-bound}, we have $\fdr(\hat R^\cbt)\le  \alpha \EE[\Gamma]$, where
    $\Gamma$ is defined as in~\eqref{eq:gamma-defn}. 
    Now define
    \[ 
        \cJ 
        \coloneqq
        \{j\in\cH_0:
        W_{j,0}=0,\ \rho_j\le\tau_j^-\},
        \qquad k\coloneqq|\cJ |.
    \] 
    Similar to how $\tau_j^{-}$ does not vary in value across $j\in\cJ_d$ in the proof of Proposition~\ref{prop:cbt-counting-fdr-bound}, we can show via the same reasoning that $\tau_j^{-} = \tau_{j'}^{-}$ for $j, j' \in \cJ$ (we omit the argument for brevity).
    Denote this common value of $\tau_j$ for $j\in \cJ$ as $\tau_*$. 
    We can therefore write, when $k\ge 1$,
    \begin{equation}
        \begin{split}
            \Gamma = \sum_{j \in \cJ} \frac{1/n_j}{\hat V_{j,-} } \le \sum_{j \in \cJ} \frac{1/ \nmin}{ w_0 + w_0 \sum_{\ell \in [m] \setminus\{ j\}} (1- W_{\ell, 0}) \ind\{ \rho_\ell \le \tau_* \}} \le \sum_{j \in \cJ} \frac{1/ \nmin}{ k w_0} = \frac{1}{\nmin w_0};
        \end{split}
    \end{equation}
    the above bound holds trivially when $k = 0$. We conclude that $\fdr(\hat R^\bh) \le   \alpha / (\nmin w_0)$.
\end{proof}

\subsection{Subsampling to equal-size competitions}

For $j\in[m]$, let $\cK_j$ be a
$\nmin$-subset of $[n_j]$, chosen uniformly at random
and drawn independently of all score values. 
Define
\begin{equation}
    W_{j,k}^{\mathrm{sub}}
    \coloneqq
    \ind\left\{
    S_{j,k}
    =
    \max_{a\in\{0\}\cup\cK_j}S_{j,a}
    \right\},
    \qquad k\in\{0\}\cup\cK_j,
\end{equation}
and
\begin{equation}
    \label{eq:subsampled-sss-compound-pval}
    \tilde p_j^{\mathrm{sub}}
    \coloneqq
    \left[
    \frac{
        \frac{1}{\nmin}
        +\frac{1}{\nmin}
        \sum_{\ell=1}^m\sum_{k\in\cK_\ell}
        \ind\{W_{\ell,k}^{\mathrm{sub}}=1,\
        S_{\ell,k}\ge S_{j,0}\}
        +m\ind\{W_{j,0}^{\mathrm{sub}}=0\}
    }{m}
    \right]\wedge1,
    \qquad j\in[m].
\end{equation}
Thus, we first compute the winners of the competition determined by the 
subsampled control units and the test unit of each hypothesis, then  
construct SSS+ compound p-values using these subsampled competitions. 

\begin{proposition}
    \label{prop:subsampled-bh-fdr-control}
    Suppose that Assumption~\ref{assu:individual} holds and that there are
    no ties among the scores within each competition.
    Further assume that $(\cK_j)_{j\in[m]}$ are independent of the 
    scores $\cS$. Then, for every
    $\alpha\in(0,1)$,
    \begin{equation}
        \fdr\left(
        \hat R^\bh\bigl(
        \tilde p_1^{\mathrm{sub}},\dots,\tilde p_m^{\mathrm{sub}}
        \bigr)
        \right)
        \le\alpha.
    \end{equation}
\end{proposition}

\begin{proof}
    For each $j\in\cH_0$, the retained $\nmin+1$ scores
    $(S_{j,k})_{k\in\{0\}\cup\cK_j}$ are jointly exchangeable
    conditional on $\sigma\bigl(\cG_{-j}, (\cK_\ell)_{\ell \in [m]}\bigr)$.
    Hence, we may invoke Theorem~\ref{thm:inflated-fdr-control} conditional on $(\cK_j)_{j\in[m]}$, and the result follows by the tower property of conditional expectation.  
\end{proof}

%% file: sections_appendix/appd_miscellaneous_results.tex

\section{Other theoretical results}
\label{sec:appd-misc-results}

\subsection{A counterexample for Holm's procedure}
\label{sec:appd-holm-counterexample}

Holm's   procedure~\citep{holm1979simple} first orders the compound (or strict) p-values as
\[
    \tilde p_{(1)}\le\cdots\le\tilde p_{(m)};
\]
denote $j_{(k)}$ as the index corresponding to $\tilde p_{(k)}$, with ties broken according to the order of the original index. Starting with $k=1$, we reject 
$H_{j_{(k)}}$ and continue to the next ordered p-value whenever
\[
    \tilde p_{(k)}\le\frac{\alpha}{m-k+1};
\]
at the first failed comparison, Holm's procedure stops and does not reject any of the remaining
hypotheses. 

\begin{proposition}
    \label{prop:holm-compound-counterexample}
    For every $m\ge2$, there exist independent compound p-values
    $\tilde p_1,\dots,\tilde p_m$   for which Holm's procedure at
    any level $\alpha\in(0,1)$ has
    \[
        \fwer=(m\alpha)\wedge1>\alpha.
    \]
\end{proposition}

\begin{proof}
    Let $\cH_0=\{1\}$, draw $U\sim\unif\bigl([0,1]\bigr)$, and define
    \[
        \tilde p_1=\frac{U}{m},
        \qquad
        \tilde p_2=\cdots=\tilde p_m=0.
    \]
    These random variables are independent and for every $t\in[0,1]$,
    \[
        \sum_{j\in\cH_0}\PP(\tilde p_j\le t)
        =\PP(U\le mt)
        =(mt)\wedge1
        \le mt,
    \]
    so $\tilde{p}_1,\ldots,\tilde{p}_m$ are compound p-values.   Holm's procedure first
    rejects the $m-1$ false nulls, whose compound p-values are zero. It then tests
    the sole true null at its final critical value $\alpha$. Consequently,
    \[
        \fwer
        =\PP(\tilde p_1\le\alpha)
        =(m\alpha)\wedge1
        >\alpha,
    \]
as required.
\end{proof}

The presence of a false null is necessary for such a counterexample.  Indeed, under
the global null $\cH_0 = [m]$, Holm's procedure makes at least one rejection if and only if
$\min_{j \in [m]} \tilde p_j\le\alpha/m$, whose probability is at most $\alpha$ by a
union bound and the compound validity of $\tilde p_1, \dots, \tilde p_m$.

\subsection{PRDS property of block empirical p-values}
\label{sec:appd-block-pvalues-prds}

We prove that the block-specific empirical p-values 
$p_1^\blk, \dots, p_m^\blk$ are PRDS on $\cH_0$; specifically, we show that for any increasing set $D \subseteq [0,1]^m$ and for all $\ell \in \cH_0$, we have that
\begin{equation}
    t\mapsto \PP\bigl(   (p_j^{\blk})_{j\in[m]}\in D
        \given p_\ell^{\blk}=t\bigr)
\end{equation}
is an increasing function on the support of $p_\ell^{\blk}$.  This implies
that BH applied to these strict p-values 
achieves nominal FDR control~\citep{benjamini2001control}. 

\begin{proposition}[PRDS under block-specific exchangeability]
    \label{prop:block-pvalues-prds}
    Suppose that Assumption~\ref{assu:block-specific} holds and that
    there are no ties among the scores within each block. Then the block
    empirical p-values $(p_j^\blk)_{j=1}^m$ are PRDS on $\cH_0$.
\end{proposition}

\begin{proof}
    We will use the definitions and notation from Appendix~\ref{sec:supermartingale-fixed-shared-application}. Specifically, for $g\in [G]$, we take definitions there of $\cT_g$ (with cardinality $h_g$), $\cL_g$, the multiset $\cM_g$ (with cardinality $d_g = h_g + n_g$), and the $\sigma$-algebra $\cF_0$. As in  Appendix~\ref{sec:supermartingale-fixed-shared-application}, 
    we write the scores in $\cM_g$ as   $Z_{g,1}>\cdots>Z_{g,d_g}  $, with 
    $\Pi_g:[d_g]\to\cL_g$ defined as the unique bijection satisfying $Z_{g,q}=S_{\Pi_g(q)}$ 
    for each $q \in [d_g]$.
    We  then argue via Assumption~\ref{assu:block-specific} that $\Pi_g \given \cF_0 \sim \unif\bigl( \operatorname{Sym}(\cL_g) \bigr)$ and that $(\Pi_g)_{g\in [G]}$ are independent conditional on $\cF_0$, as established previously. 

    Fix $\ell \in \cH_0$, write $g  = g(\ell)$, and define
    \[
        \cF_{-\ell} \coloneqq \sigma\Bigl( \cF_0,  
        (\Pi_h)_{h\neq g},\,
        \bigl(\Pi_g^{-1}(u)\bigr)_
        {u\in\cT_g\setminus\{(\ell,0)\}}  \Bigr)
    \]
    Thus, $\cF_{-\ell} $ records the assignments outside block  $g$ and the positions of all other null test units within block  $g$. The remaining randomness is from the $n_g + 1$ units' positions, corresponding to the $\ell$th test score and the $n_g$ control scores. 

    Let $q_1<\cdots<q_{n_g+1} $ enumerate the indices in   $[d_g]\setminus \bigl \{ \Pi_g^{-1}(u): u\in\cT_g\setminus\{(\ell,0)\}\bigr\}$   in increasing order. 
    On the event  $\{\Pi_g(q_r)=(\ell,0)\}$ , exactly $r-1$ control scores in group $g$ exceed  $S_{\ell,0}=Z_{g,q_r} $. Hence, 
    \[
        \PP\biggl( p_\ell^\blk  = \frac{r}{n_g+ 1} \bigggiven \cF_{-\ell}\biggr)
        = \PP \bigl(  \Pi_g(q_r)=(\ell,0)  \given \cF_{-\ell}\bigr )=\frac{1}{n_g+1},
        \qquad r\in[n_g+1].
    \]
    Since this conditional probability does not depend on $\cF_{-\ell}$, we conclude $p_\ell^\blk$ is independent of $\cF_{-\ell}$. 

    Now, for $r \in [n_{g}+1]$, define
    \[
        p_j^{\blk,(r)}
        \coloneqq
        \begin{cases}
        \displaystyle 
          \frac{r}{n_g+1} 
        & \text{for }j=\ell,\\[4pt]
        \displaystyle 
        \frac{1+\sum_{a\in[n_g+1]\setminus\{r\}}
        \ind\{Z_{g,q_a}\ge S_{j,0}\}}{n_g+1} 
        & \text{for }j\in B_g\setminus\{\ell\},\\ 
        p_j^{\blk} 
        & \text{for }j\notin B_g.
        \end{cases}
    \]
    This $r$-perturbed p-value collection $(p_j^{\blk, (r)})_{j \in [m]}$ is $\cF_{-\ell}$-measurable and coincides with the original block empirical p-values $(p_j^{\blk })_{j \in [m]}$ on the event $\{S_{\ell,0} = Z_{g, q_r}\}$.   

    For $r < s$, we see by considering the different cases that
    \[
        (n_g+1)\bigl(p_j^{\blk,(s)}-p_j^{\blk,(r)}\bigr)
        =
        \begin{cases}
        s-r
        & \text{for } j=\ell,\\ 
        \ind\{Z_{g,q_r}\ge S_{j,0}\}
        -\ind\{Z_{g,q_s}\ge S_{j,0}\} 
        & \text{for } j\in B_g\setminus\{\ell\},\\ 
        0 
        & \text{for }  j\notin B_g.
        \end{cases}
    \]
    Since $Z_{g, q_r} > Z_{g, q_s}$, all three cases lead to nonnegative values for $(p_j^{\blk,(s)}-p_j^{\blk,(r)})$. Hence, increasing $r$ cannot decrease any individual coordinate $p_j^{\blk, (r)}$, conditional on $\cF_0$. 

    Finally, let $D \subseteq [0,1]^m$ be any increasing set. Since $\{ p_\ell^\blk = r / (n_g + 1) \} = \{ S_{\ell, 0}  =Z_{g, q_r} \}$, and since $(p_j^{\blk, (r)})_{j \in [m]}= (p_j^{\blk })_{j \in [m]}$ on this event, we can compute:
    \begin{multline}
        \PP\biggl(   (p_j^{\blk})_{j\in[m]}\in D
        \bigggiven p_\ell^{\blk}=\frac{r}{n_g+1} \biggr)  \\
        = \EE \biggl[
        \ind\bigl\{(p_j^{\blk,(r)})_{j\in[m]}\in D\bigr\}
        \bigggiven p_\ell^{\blk}=\frac{r}{n_g+1}
        \biggr] 
        =\EE \bigl[ \ind \{(p_j^{\blk,(r)})_{j\in[m]}\in D\}  \bigr].
    \end{multline} 
    The last expectation is an increasing function in $r$, as $D$ is an increasing set and the $r$-perturbed p-values are coordinate-wise increasing. As $\ell \in \cH_0$ was arbitrarily chosen, this establishes the PRDS property of $(p_j^{\blk})_{j\in[m]}$ over $\cH_0$. 
\end{proof}

%% file: sections_appendix/appd_deferred_sims.tex

\section{Additional simulation results}
\label{sec:appd-deferred-sims}

\subsection{Varying competition sizes chosen in a range}
\label{sec:appd-binom-heterog}

Figure~\ref{fig:heterog-smooth-10-50} visualizes the FDR, FWER, and power curves for the methods of interest when $n_j$ is sampled in a range between $10$ and $50$. Specifically, independently across $j \in [m]$, we generate $n_j \sim 10 + \Bin(40, \rho_{g(j)})$, where $\rho_{g}$ for $g\in[25]$ is defined case-by-case  (see Section~\ref{sec:sims-heterog}).
The smooth configuration exhibits similar results to those of the binary configuration, and our conclusion is identical to the one written in Section~\ref{sec:sims-heterog} for the binary configuration.

\newpage 

\begin{figure}[!tb]
    \centering
    \includegraphics[width=\textwidth,clip]{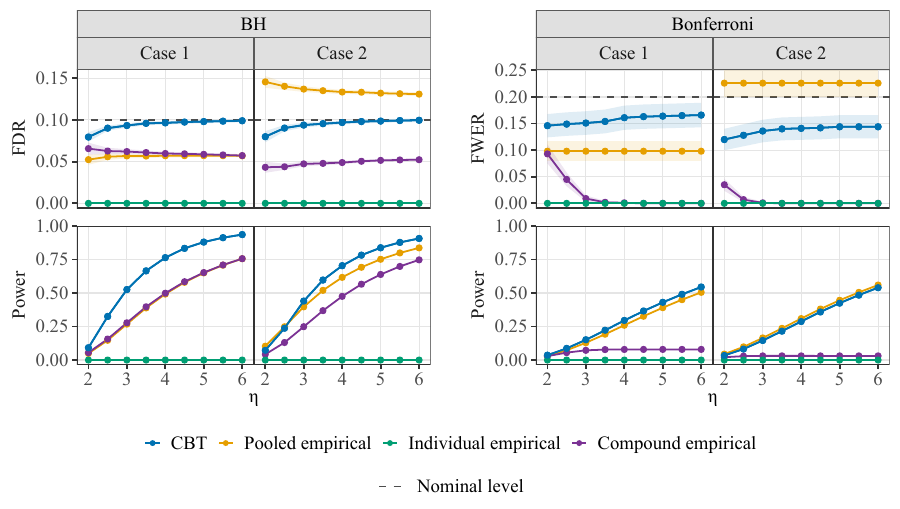}
    \caption{Power, FDR and FWER curves for the varying-size competition experiments, where 
    $n_j$ follows the smooth setting distribution. Cases~1 and~2 refer to scenarios where higher-variance clusters have more controls and lower variance clusters have more controls, as described in the text.}
    \label{fig:heterog-smooth-10-50}
\end{figure}

\begin{figure}[!tb]
    \centering
    \includegraphics[width=0.5\textwidth,clip]{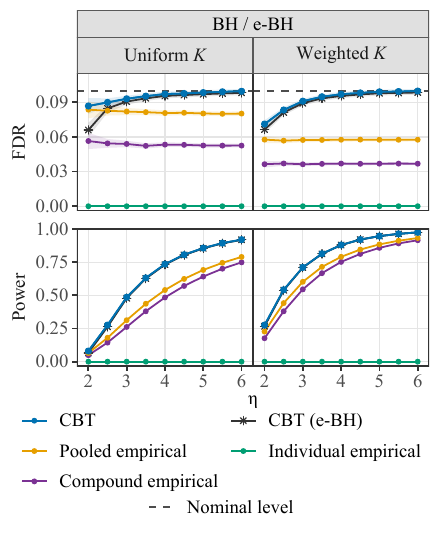}
    \caption{Power and FDR  curves for the equal-size competition experiments, including the e-BH approach. The setting is the same as that of Figure~\ref{fig:homog-n25-unif_and_weighted}.}
    \label{fig:supp-homog}
\end{figure}

\begin{figure}[!tb]
    \centering
    \includegraphics[width=\textwidth,clip]{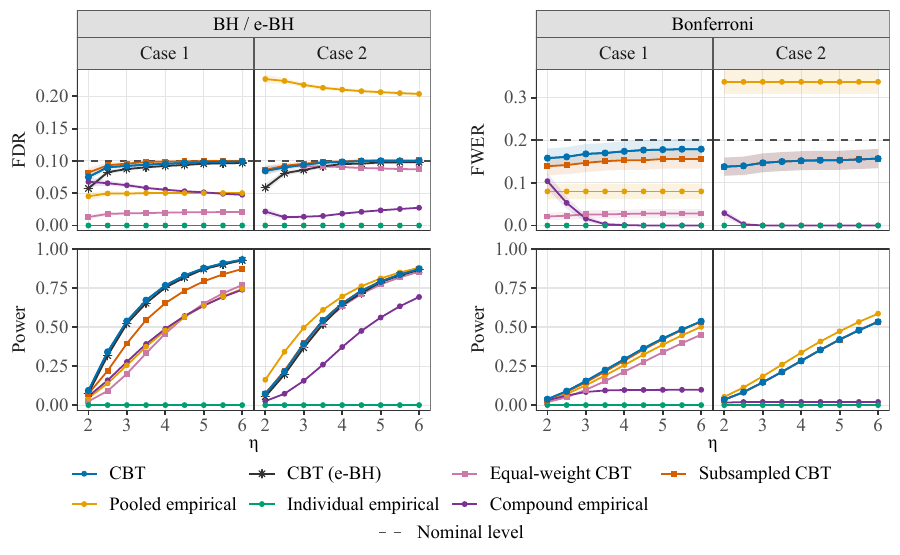}
    \caption{Power, FDR and FWER curves for the varying-size competition experiments, including the alternate procedures, for the setting of Figure~\ref{fig:heterog-binary-10-50}.}
    \label{fig:supp-heterog-binary}
\end{figure}

\begin{figure}[!tb]
    \centering
    \includegraphics[width=\textwidth,clip]{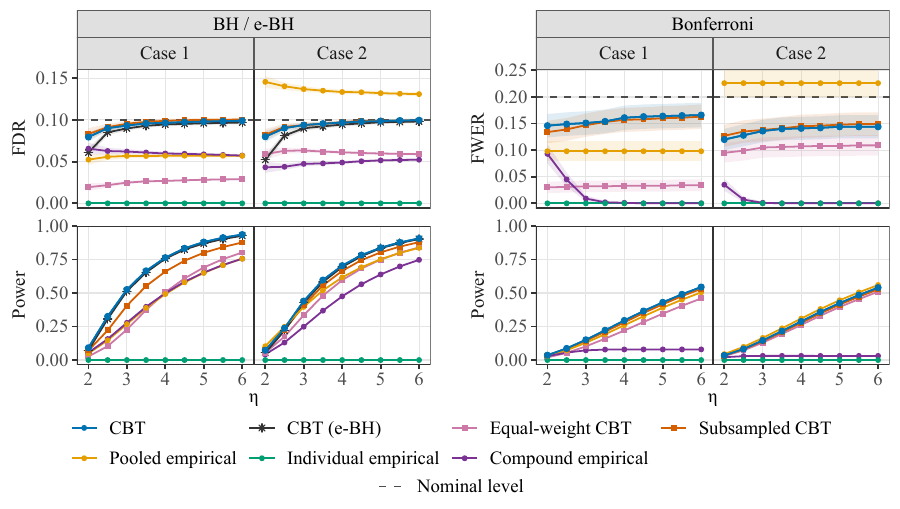}
    \caption{Power, FDR and FWER curves for the varying-size competition experiments, including the alternate procedures, for the setting of Figure~\ref{fig:heterog-smooth-10-50}.}
    \label{fig:supp-heterog-smooth}
\end{figure}
\clearpage

\subsection{Equal weighting, subsampling, and compound e-values}

In Section~\ref{sec:fdr-cbt}, we discussed three alternative methods that guarantee FDR control at the nominal level: equal-weight CBT, CBT with subsampled control groups, and e-BH with competition-based compound e-values. We now examine the performance of each   alternative and compare them to both the CBT method and baseline methods in the settings discussed in Sections~\ref{sec:sims-homog} and~\ref{sec:sims-heterog}.  

\paragraph{Results.} In the simulation setting of Section~\ref{sec:sims-homog}, first note that the competitions are of equal size, so equal-weight CBT and subsampled CBT both coincide with the original CBT method. Hence, the sole new method visualized in Figure~\ref{fig:supp-homog} is e-BH applied to the competition-based compound e-values~\eqref{eq:cbt-compound-evalue}. As discussed in Section~\ref{sec:fdr-cbt}, $\hat R^{\ebh}(\tilde e_1, \dots, \tilde e_m) \subseteq \hat R^\cbt$. Empirically, the FDR and power of the compound p-value and e-value approaches  are virtually identical. 

Figures~\ref{fig:supp-heterog-binary} and~\ref{fig:supp-heterog-smooth} visualize the Type I errors and powers of the three methods in the setting of Section~\ref{sec:sims-heterog}; the figures also include the methods described in the main text.  
Again, the e-BH approach achieves similar power to our original CBT method. Both the equal-weight and subsampled variants of CBT control the error rates at the nominal level, but achieve lower power than original CBT. 
These approaches give compound p-values with  weights of $1/\nmin$ in the numerator; this leads to a more conservative rejection set via BH, as evidenced by the figures.  

\section{Deferred simulation details}
\label{sec:appd-deferred-details}

We provide additional experiment details omitted from Sections~\ref{sec:sims} and~\ref{sec:real_data}.

\subsection{Additional details on the data-generating process and model fitting}
\label{sec:appd-details-on-dgp}

\paragraph{Generating $\bm u$.}

We explicitly state how the vector $\bm u $ is generated in the data-generating process described in Section~\ref{sec:dgp}. 
Generate $(v_1, \dots, v_{25}) \sim \cN_{25}({0, I})$ and compute
\[
u_g \coloneqq \frac{v_g-\bar v}
{\bigl(25^{-1}\sum_{h=1}^{25}(v_h-\bar v)^2 \bigr)^{1/2}}.
\]
Both $\cK$ and $\bm u$ are generated once, and are fixed across all experiments and all replications.

\paragraph{Details on the OneClassSVM.}

The OneClassSVM is fitted with default parameters on an independent held-out training dataset of 2,000 points drawn from the uniform-weighted null mixture distribution. This is done once per replication (with the held-out training dataset resampled for each). We then score the units with the negative of the fitted function, as the convention in this paper is that higher scores signify greater evidence against the null.

\paragraph{Estimating block structure via the $K$-means algorithm.}

In the shared-control competition setting of Section~\ref{sec:sims-block}, we discuss estimating the block structure of the data, which is generated by sampling the mixture distribution. We fit the \texttt{sklearn.cluster.KMeans} function in the \texttt{scikit-learn} package, with 25 clusters and 10 initializations, on the OneClassSVM held-out training data of 2,000 observations. 

Denote the fitted centers as $\hat \mu_1, \dots, \hat \mu_{25}$.  We then assign a training observation  with score $x$ to $\hat\mu_{a(x)}$, where 
\[
a(x) \coloneqq \underset{h\in[25]}{\mathrm{argmin}}\|x-\hat\mu_h\|_2.
\]
For each $h \in [25]$, compute 
$    N_h^{\mathrm{train}}    \coloneqq   \sum_{i=1}^{2000}   \ind\{a(X_i^{\mathrm{train}})=h\}.$ 
To avoid using clusters with too few assigned training observations, we define the index set for \emph{retained} fitted centers as  
\[
\hat\cI 
\coloneqq 
\{h:N_h^{\mathrm{train}}\ge20\}.
\] 
We may then construct the following map from the initial fitted groups to the retained fitted groups as
\[
    M(h) \coloneqq
    \begin{cases}
    h, & h\in\hat\cI,\\ 
    \underset{\ell\in[25]}{\mathrm{argmin}}
    \|\widehat\mu_h-\widehat\mu_\ell\|_2,
    & h\notin\hat \cI.
    \end{cases}
\]
The final assignment rule is
\[
\operatorname{gr}(x) \coloneqq M\bigl(a(x)\bigr).
\]
For each control and test unit, we therefore:
\begin{enumerate}
    \item[1.] Find its nearest center among all 25 original fitted centers.
    \item[2.] Use the map $M(\cdot)$ to compute its block membership with regard to the retained cluster centers.
\end{enumerate}
For each block $g \in [|\hat \cI|]$, we define 
\[
B_g \coloneqq \bigl\{j:\operatorname{gr}(X_{j,0})=g\bigr\},
\qquad
C_g \coloneqq \bigl\{i:\operatorname{gr}(X_i^{\mathrm{control}})=g\bigr\}.
\]
and compute the compound p-values in~\eqref{eq:shared-cbt-common-offset-pval}. 

As the held-out training set is resampled for every replication, the pairs $(B_g, C_g)_{g\in  [|\hat \cI|]}$ used to compute the compound p-values also vary across different replications. 

\subsection{Covariate-based matching algorithms}
\label{sec:appd-matching-algorithms}

We describe two algorithms for constructing blocks of tests and shared
controls. Each algorithm takes a collection of tests and controls that
are eligible to be matched to one another. When matching is restricted
to observations belonging to the same pre-existing group, such as a
histological cluster (Section~\ref{sec:application-virtual-st}) or a
cell line (Section~\ref{sec:application-crispr}), the algorithm is
applied separately within each such group.

For an application of an algorithm, index the tests by $[m]$ and the
eligible controls by $[N]$, where $m,N \in \NN$. Let
$(x_j)_{j\in[m]}$ and $(z_i)_{i\in[N]}$ denote their $p$-dimensional
covariates, standardized coordinate-wise using means and standard
deviations computed from the pooled test and control observations.
Both algorithms return pairs of nonempty index sets
$(B_g,A_g)_{g=1}^G$ satisfying
\[
    \bigsqcup_{g=1}^G B_g=[m],
    \qquad
    \bigsqcup_{g=1}^G A_g\subseteq[N].
\]
Here, $B_g$ indexes tests and $A_g$ indexes controls; the scores of
the controls indexed by $A_g$ form the shared control collection of scores 
$C_g$. Controls in groups containing no tests are unused; thus
$\sum_{g=1}^G n_g\le N$, where $n_g \coloneqq |A_g|$.
Write $a_+ \coloneqq \max\{a,0\}$ and break ties lexicographically unless stated otherwise.

\subsubsection{Adaptive recursive coarsening}
\label{sec:appd-arc}

ARC recursively partitions the pooled test and control covariates,
favoring splits with small test-to-control ratios and balanced block sizes.

\begin{algorithm}[H]
\footnotesize
\SetAlgoNlRelativeSize{-1}
\SetAlCapFnt{\normalsize}
\SetAlCapNameFnt{\normalsize}
\DontPrintSemicolon
\caption{Adaptive recursive coarsening (ARC)}
\label{alg:appd-arc}

\KwIn{test covariates $(x_j)_{j\in[m]}$;
control covariates $(z_i)_{i\in[N]}$;
 control-count stopping threshold $h\ge1$;
 maximum depth $D\ge0$;
number of candidate coordinates $s\in[p]$;
finite nonempty set of split quantile levels
$\mathcal Q\subset(0,1)$;
child-size imbalance penalty weight $\varepsilon\ge0$.}

Set $\mathcal P\gets\varnothing$ and
$\mathcal S\gets\{([m],[N],0)\}$.\;

\While{$\mathcal S\ne\varnothing$}{
    Remove a cell $(I,J,d)$ from $\mathcal S$.\;

    \If{$I=\varnothing$}{
        \textbf{continue}\;
    }
    \If{$|J|\le h$ \textnormal{\bfseries or} $d\ge D$}{
        $\mathcal P\gets\mathcal P\cup\{(I,J)\}$.\;
        \textbf{continue}\;
    }

    \BlankLine
    {\color{gray}\tcp{Construct candidate splits with controls
    for every  child containing a test.}}

    Let $L\subseteq[p]$ index the $s$ largest coordinate ranges
    among the pooled covariates in $(I,J)$.\;
    Set $\mathcal A\gets\varnothing$.\;

    \For{$(\ell,q)\in L\times\mathcal Q$}{
        Set $t$ to the empirical $q$-quantile of the pooled
        values $(x_{j\ell})_{j\in I}$ and
        $(z_{i\ell})_{i\in J}$.\;

        Set
        \[
        \begin{aligned}
            I_L&\gets\{j\in I:x_{j\ell}\le t\},
            & I_R&\gets I\setminus I_L,\\
            J_L&\gets\{i\in J:z_{i\ell}\le t\},
            & J_R&\gets J\setminus J_L.
        \end{aligned}
        \]
        Set $a\gets(I_L,J_L,I_R,J_R)$.\;

        \If{$|I_u|+|J_u|>0$ and
            $(|I_u|=0\ \text{or}\ |J_u|\ge1)$
            for both $u\in\{L,R\}$}{
            $\mathcal A\gets\mathcal A\cup\{a\}$.\;
        }
    }

    \BlankLine
    {\color{gray}\tcp{Keep a terminal cell or choose the best split.}}

    \eIf{$\mathcal A=\varnothing$}{
        $\mathcal P\gets\mathcal P\cup\{(I,J)\}$.\;
    }{
        Choose $(I_L,J_L,I_R,J_R)\in
        \argmin{a\in\mathcal A}F(a)$, where
        \[
            F(a)
            =
            \max_{u\in\{L,R\}}\frac{|I_u|}{|J_u|\vee1}
            +
            \varepsilon
            \bigl|(|I_L|+|J_L|)-(|I_R|+|J_R|)\bigr|.
        \]

        $\mathcal S\gets\mathcal S\cup
        \{(I_L,J_L,d+1),(I_R,J_R,d+1)\}$.\;
    }
}

\KwOut{shared-control blocks
$\mathcal P=\{(B_g,A_g)\}_{g=1}^G$,
where $G=|\mathcal P|$.}
\end{algorithm}

\clearpage

\subsubsection{Greedy grouping}
\label{sec:appd-gg}

GG first partitions the controls into small groups and then assigns tests
via covariate distance, with a penalty on the resulting test-to-control
ratio to discourage assigning too many tests to groups with few controls.

\begin{algorithm}[H]
\footnotesize
\SetAlgoNlRelativeSize{-1}
\SetAlCapFnt{\normalsize}
\SetAlCapNameFnt{\normalsize}
\DontPrintSemicolon
\caption{Greedy grouping (GG)}
\label{alg:appd-gg}

\KwIn{test covariates $(x_j)_{j\in[m]}$;
control covariates $(z_i)_{i\in[N]}$;
 maximum control-group size $h\ge1$;
 number of nearest controls $k\ge1$
for identifying candidate groups;
ratio penalty weight $\lambda\ge0$;
ratio threshold $r_0>0$ above which penalties apply;
penalty exponent $\gamma>0$.}

{\color{gray}\tcp{Partition controls into groups of size at most $h$.}}

Initialize the ordered list of control groups
$\mathcal L\gets([N])$.\;

\While{$\exists J\in\mathcal L:\ |J|>h$}{
    Choose such a $J$ and set
    $\ell\gets\min\argmax{a\in[p]}
    (\max_{i\in J}z_{ia}-\min_{i\in J}z_{ia})$.\;

    Stably sort $J$ as $(i_1,\ldots,i_{|J|})$
    by increasing $z_{i\ell}$.\;

    Set $v\gets\lfloor|J|/2\rfloor$,
    $J_L\gets(i_1,\ldots,i_v)$, and
    $J_R\gets(i_{v+1},\ldots,i_{|J|})$.\;

    Replace $J$ in $\mathcal L$ by $J_L,J_R$, in that order.\;
}

Label the resulting control index sets $A_1,\ldots,A_K$
in the order they appear in $\mathcal L$, and set
$n_g\gets|A_g|$, $B_g\gets\varnothing$ for $g\in[K]$.\;

\BlankLine
{\color{gray}\tcp{Identify candidate groups for each test.}}

\For{$j\in[m]$}{
    Let $N_j$ index the $\min(k,N)$ controls nearest to $x_j$
    in Euclidean distance.\;

    Set
    \[
        \mathcal G_j
        \gets\{g\in[K]:A_g\cap N_j\ne\varnothing\},
        \qquad
        d_{jg}\gets
        \min_{i\in A_g\cap N_j}\|x_j-z_i\|_2
        \quad(g\in\mathcal G_j).
    \]
}

\BlankLine
{\color{gray}\tcp{Assign tests, starting with those having
fewer candidate groups.}}

Order the tests as $(j_1,\ldots,j_m)$ by increasing
$(|\mathcal G_j|,\min_{g\in\mathcal G_j}d_{jg},j)$
lexicographically.\;

\For{$j=j_1,\ldots,j_m$}{
    Set $r_{jg}\gets(|B_g|+1)/n_g$
    for $g\in\mathcal G_j$.\;

    Choose
    \[
        g^*\in\argmin{g\in\mathcal G_j}
        \left\{
            d_{jg}+\lambda(r_{jg}-r_0)_+^\gamma
        \right\},
    \]
    breaking ties by increasing
    $(d_{jg},r_{jg},g)$ lexicographically.\;

    $B_{g^*}\gets B_{g^*}\cup\{j\}$.\;
}

Set $\mathcal P\gets
\{(B_g,A_g):g\in[K],\ B_g\ne\varnothing\}$
and relabel the retained blocks consecutively.\;

\KwOut{shared-control blocks
$\mathcal P=\{(B_g,A_g)\}_{g=1}^G$,
where $G=|\mathcal P|$.}
\end{algorithm}